\documentclass[12pt,a4]{article}
\usepackage{natbib}
\usepackage{amsmath,amssymb,indentfirst,latexsym}
\usepackage[FIGTOPCAP]{subfigure}
\usepackage[dvips]{graphics,color}
\usepackage{mathptmx}      
\usepackage{bm}
\usepackage{amsthm,amsfonts,epsfig}
\usepackage{amsmath}
\usepackage[T1]{fontenc}
\usepackage{epsfig}
\usepackage{graphicx}
\usepackage{graphics}
\usepackage{placeins}
\usepackage{amssymb,srcltx,amsthm,epsfig}
\usepackage{lscape}
\usepackage{longtable}
\usepackage[FIGTOPCAP]{subfigure}
\usepackage{here} 
\usepackage{stackengine}

\newtheorem{theorem}{Theorem}

\newcommand{\iid}{\ensuremath{\stackrel{{\textrm{iid}}}{\sim}}} 
\newcommand{\ind}{\ensuremath{\stackrel{{\textrm{ind}}}{\sim}}} 

\newcommand{\Z}{\mathbf{Z}}

\newcommand{\y}{\mathbf{y}}

\newcommand{\p}{\mathbf{p}}

\newcommand{\balpha}{\mbox{${ \bm \alpha}$}}

\newcommand{\bbeta}{\mbox{${\bm \beta}$}}
\newcommand{\btheta}{\mbox{${ \bm \theta}$}}

\newcommand{\ii}{i=1,\ldots,n}

\newcommand{\jj}{j=1,\ldots,G}

\newcommand{\sumas}{\sum^n_{i=1}}

\newcommand{\bgamma}{\mbox{${ \bm \gamma}$}}

\newtheorem{proposition}{Proposition}
\begin{document}
\title{Finite Mixture of Birnbaum-Saunders distributions using the $k$-bumps algorithm}
\author{\small Luis Benites$^{\textrm{a}}$\thanks{Correspondence to: E-mail
address: lbenitesanchez@gmail.com (Luis Benites)} ~ Roc\'{i}o Maehara$^{\textrm{a}}$ ~ Filidor Vilca$^{\textrm{b}}$ ~ Fernando Marmolejo-Ramos$^{\textrm{c}}$\\
 {\em \small $^{ \small \textrm{a}}$Departamento de Estat\'{\i}stica, Universidade de S\~{a}o Paulo, Brazil }  \\   {\em \small $^{ \small \textrm{b}}$Departamento de Estat\'{\i}stica, Universidade de Campinas, Brazil } \\
{\em \small $^{ \small \textrm{c}}$School of Psychology, The University of Adelaide, Adelaide, Australia}
}
\date{}
\maketitle
\begin{abstract}
Mixture models have  received a great deal of attention in statistics due  to the wide range of applications found  in recent years.   This paper discusses a finite mixture model of Birnbaum-Saunders distributions with  $G$ components, as an important supplement of the  work developed  by \cite{Balakrishnan:11}, who  only considered  two  components.   Our proposal enables the modeling of proper multimodal scenarios with greater flexibility, where   the identifiability of the model with $G$ components  is  proven and  an EM-algorithm for the maximum likelihood (ML) estimation of the  mixture parameters is   developed,   in  which the $k$-bumps algorithm is used  as an initialization strategy  in  the EM algorithm. The performance of the $k$-bumps algorithm as an initialization tool is evaluated through  simulation experiments. Moreover, the empirical information matrix is derived analytically to account for standard error,  and   bootstrap procedures for testing hypotheses about the number of components in the mixture are  implemented.  Finally, we perform simulation studies and analyze two real datasets to illustrate the usefulness of the proposed method.

\vspace*{0.5cm} \noindent {\bf Keywords:} Birnbaum-Saunders distribution; EM algorithm; $k$-bumps algorithm; Maximum likelihood estimation;  Finite mixture.
\end{abstract}

\section{Introduction}

Although most statistical applications are conceived to deal with unimodal data, in practice research data can exhibit heterogeneity due to skewness and multimodality. More importantly, skewness and non-distinctive shape variations can be due to intrinsic aspects of the data. Skew and distinctive shape variations that resemble multimodality can indicate not all observations come from the same parent population.  In other words, if the data come from different sub-populations, and their identifications are not known, the mixture distribution can be used quite effectively to analyze the dataset \citep[]{McLachlan2000}. Although multimodal data can be modeled with a single distribution, the quality of the model is poor in general. Hence modeling based on finite mixture distributions plays a vital role in different data analysis situations. Finite mixture models are now applied in such diverse areas such as biology, biometrics, genetics, medicine, marketing, reliability, and pattern recognition problems, among others. Some    examples of mixture models are based  on gamma, exponential, inverse  Gaussian and Weibull distributions.

The  Birnbaum-Saunders (BS) distribution, originally introduced by \cite{birnbaum1969}, is a two parameter failure time distribution for modeling fatigue failure caused under cyclic loading,  which is  derived  from   the cumulative  damage or Miner  law. This  distribution has   been  considered a more attractive alternative to the often-used Weibull, gamma, and log-normal models, because  the  BS  model fits very well within the extremes of the distribution, even when the amount of fatigue life data is small. The BS distribution, also known as fatigue life distribution, was initially used to model failure times, but has since been extended to fields such as reliability, business, engineering, survival analysis, and medical sciences; see Leiva \cite{leiva_book} and \cite{leiva2016}.  A positive  random  variable $T$   is  said  to  have  a two-parameter  BS   distribution  if its  cumulative distribution function (cdf)  can be written as
\begin{eqnarray}\label{cdfBS}
F_T(t;\alpha,\beta) = P(T \leq t) =  \Phi\big(a_t(\alpha,\beta) \big),\,\,\,  t >0, \,\, \alpha>0,\, \, \beta>0,
\end{eqnarray}	
where $a_t(\alpha,\beta) = (\sqrt{t/\beta} - \sqrt{\beta/t}\big)/\alpha$
and  $\Phi(\cdot)$ is the cdf of the  standard normal  distribution. Clearly $\beta$
is   the   median  of  the  BS  distribution.

In the  context of  finite   mixture    distributions, if   the sub-populations do not have symmetric distributions,  then  the finite  mixture of BS distributions can be used to analyze these data,  since BS distributions are  positively skewed. This can make them, good  alternatives  to   these  based  on  skewed distributions.
In reliability research, for example, populations can be heterogeneous due to at least two underlying sub-populations; one being the standard sub-population (also known as strong population) and the other being the defective sub-population. Data that arise from such heterogeneous populations are amenable to modeling by a mixture of two or more life distributions. Indeed, the mixture of BS distributions seems to be a suitable approach: multimodal distributions can be approximated very well by a mixture of distributions because sub-populations tend to not have symmetric distributions. Extensive work has been carried out regarding bimodal BS distributions; see  \cite{olmos2016} and \cite{bala2009}. However, not much work has been done on  finite mixtures of BS distributions.

The aim of this paper is to consider a finite mixture model based on the BS distributions  by extending the two-component mixture BS proposed by \cite{Balakrishnan:11}. The maximum likelihood  estimates are  obtained via the EM algorithm,  in  which  the $k$-bumps algorithm \citep{bagnato2013} is used to obtain the initial values  required by the EM algorithm. The  identifiability of the FM-BS model is  discussed  following  the model proposed  by \cite{chandra1977}. An important aspect   to  be    addressed is whether a two-component model fits the data significantly better than a one-component model. This  question  is answered  by  using   the  parametric bootstrap log-likelihood ratio statistic proposed by \cite{turner2000}.

The remainder of the paper is organized as follows. In Section 2, we briefly report basic results for the BS distribution and present the  finite mixture  BS (FM-BS) distribution  along with its   properties. In Section 3, we deal with the parameter estimation of the FM-BS distribution through an  EM  algorithm,  as well as the  starting  values  and  stopping   rule used  in the  algorithm. Moreover,  an approximation of the observed information matrix  for  obtaining   the   standard error of  the ML    estimates  is   presented.  In Section 4 and 5, numerical samples using both simulated and real datasets are given to illustrate the performance of the proposed model. Finally, Section 6 contains our concluding remarks.

\section{Finite mixture BS model}
First,  we  recall  that from \eqref{cdfBS}, a positive random variable  $T$ is distributed as a BS distribution if  its  probability density function (pdf) is
\begin{equation}\label{pdfbs}
f_{T}(t;\alpha,\beta) = \phi\big(a_t(\alpha,\beta)\big) A_t(\alpha,\beta), \,\,\, t>0,
\end{equation}
where $a_t(\alpha,\beta) = (\sqrt{t/\beta} - \sqrt{\beta/t}\big)/\alpha$, and $A_t(\alpha,\beta) = t^{-3/2} (t + \beta)/(2\alpha \beta^{1/2})$ is the derivative of $a_t$ with respect to $t$.  This  distribution  is   denoted   by $T\sim {\rm BS}(\alpha,\beta)$, where  $\alpha$ and $\beta$ are the shape and scale parameters, respectively.
The BS distribution is related to the  normal distribution by means of the representation  of  $T$  give  by $T = \beta  \left(1 + 2X^2 + 2X\sqrt{1 + X^2}\right)$, where  $X \sim \mbox{N}(0, \alpha^2/4)$. The mean and  variance are    given  respectively  by
${\rm E}(T)= \beta \big(1 + \alpha^2/2\big)$ and ${\rm Var}(T) = (\alpha \beta)^2\big(1 + 5\alpha^2/4\big)$. Note  from (\ref{cdfBS}),  it  is  easy  to  see that $\beta$ is the median of the distribution of $T$. Moreover, the mode (denoted by $m$) is   obtained  as the solution of the nonlinear equation
\begin{equation}
(\beta - m) (m + \beta)^2 = \alpha^2 \beta m (m + 3\beta),
\label{eqmodaalpha}
\end{equation}
where  $m < \beta$.  Then, from the above equation,   $\alpha$ can  be    expressed   in terms of $m$ and $\beta$,  and   consequently  the pdf of the BS distribution can   be  re-parameterized   in term  of  parameters $m$ and $\beta$ as  follows
\begin{equation}
f_T(t;m,\beta) = \frac{1}{\sqrt{2 \pi}} \exp\left[-\frac{1}{2}\frac{\beta m (m + 3\beta)}{(\beta - m )} \left(\frac{a_t(1,\beta)}{m + \beta} \right)^2\right]\frac{t^{-3/2}(t+\beta)}{2 (m + \beta)^2}\sqrt{\frac{m + 3\beta}{\beta - m}}.
\label{eqBSp}
\end{equation}
\begin{figure}[h!]
\centering
\subfigure[$m=1$]{\includegraphics[scale=0.25]{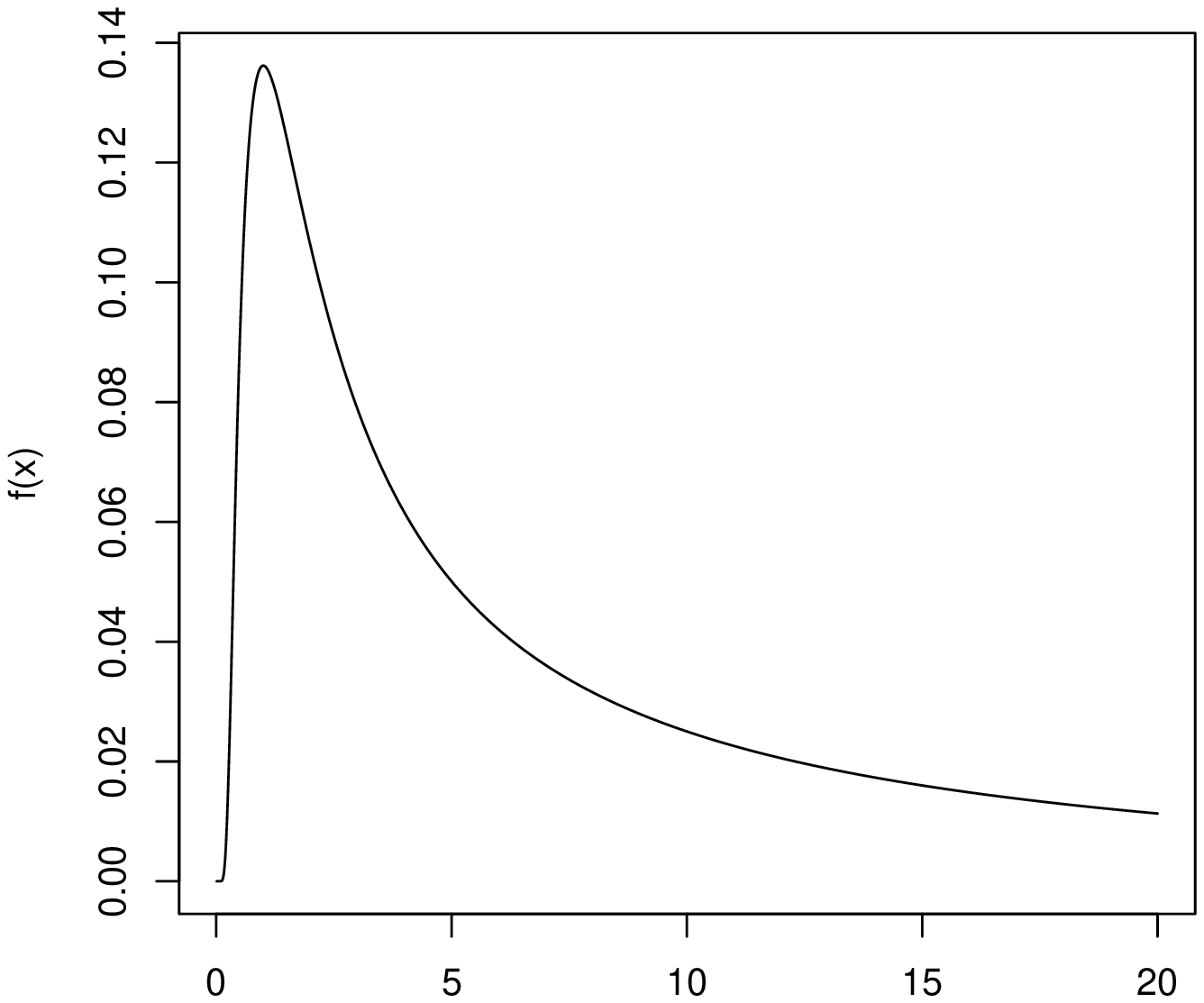}\label{}}~\subfigure[$m=2$]{\includegraphics[scale=0.25]{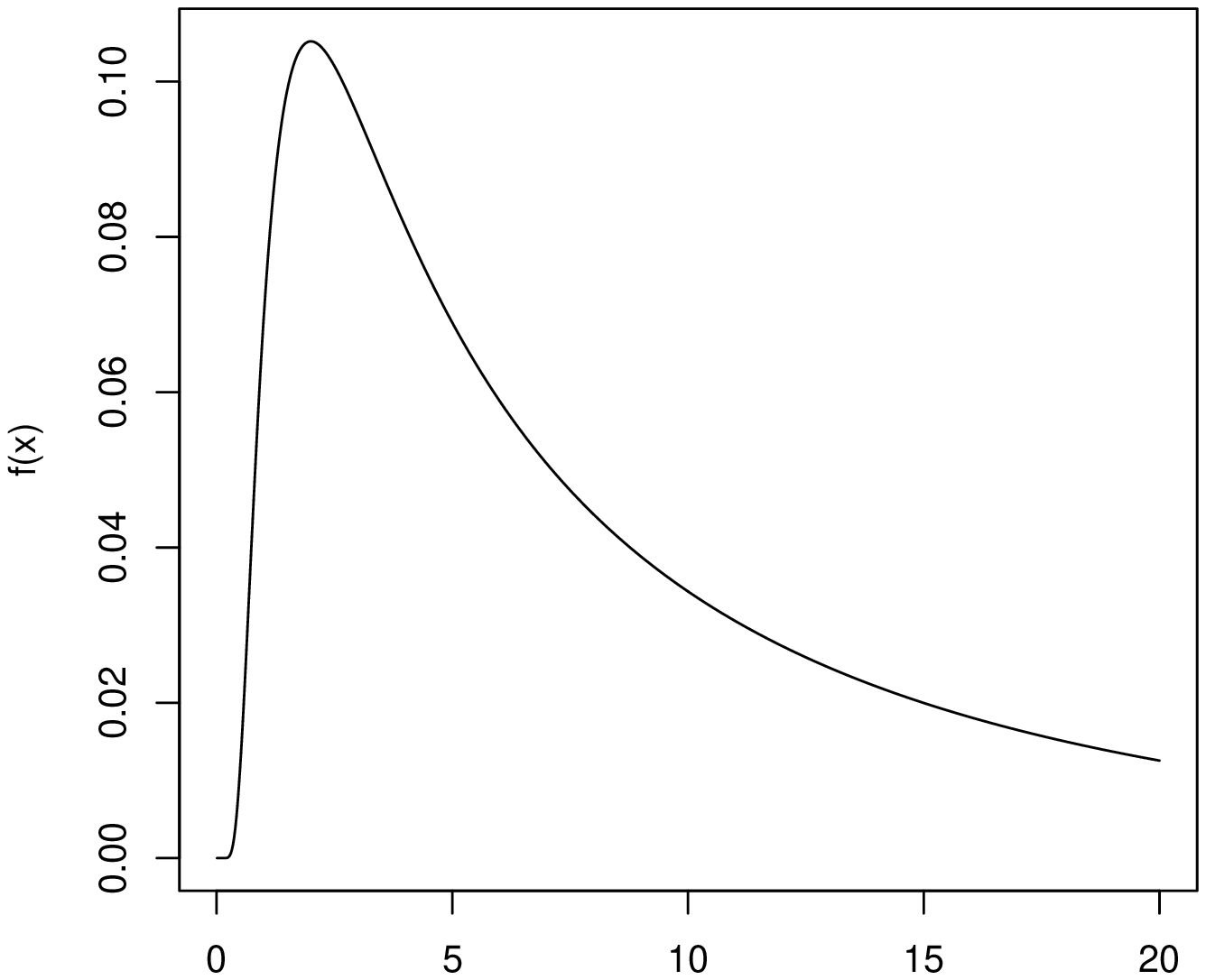}\label{}}~\subfigure[$m=3$]{\includegraphics[scale=0.25]{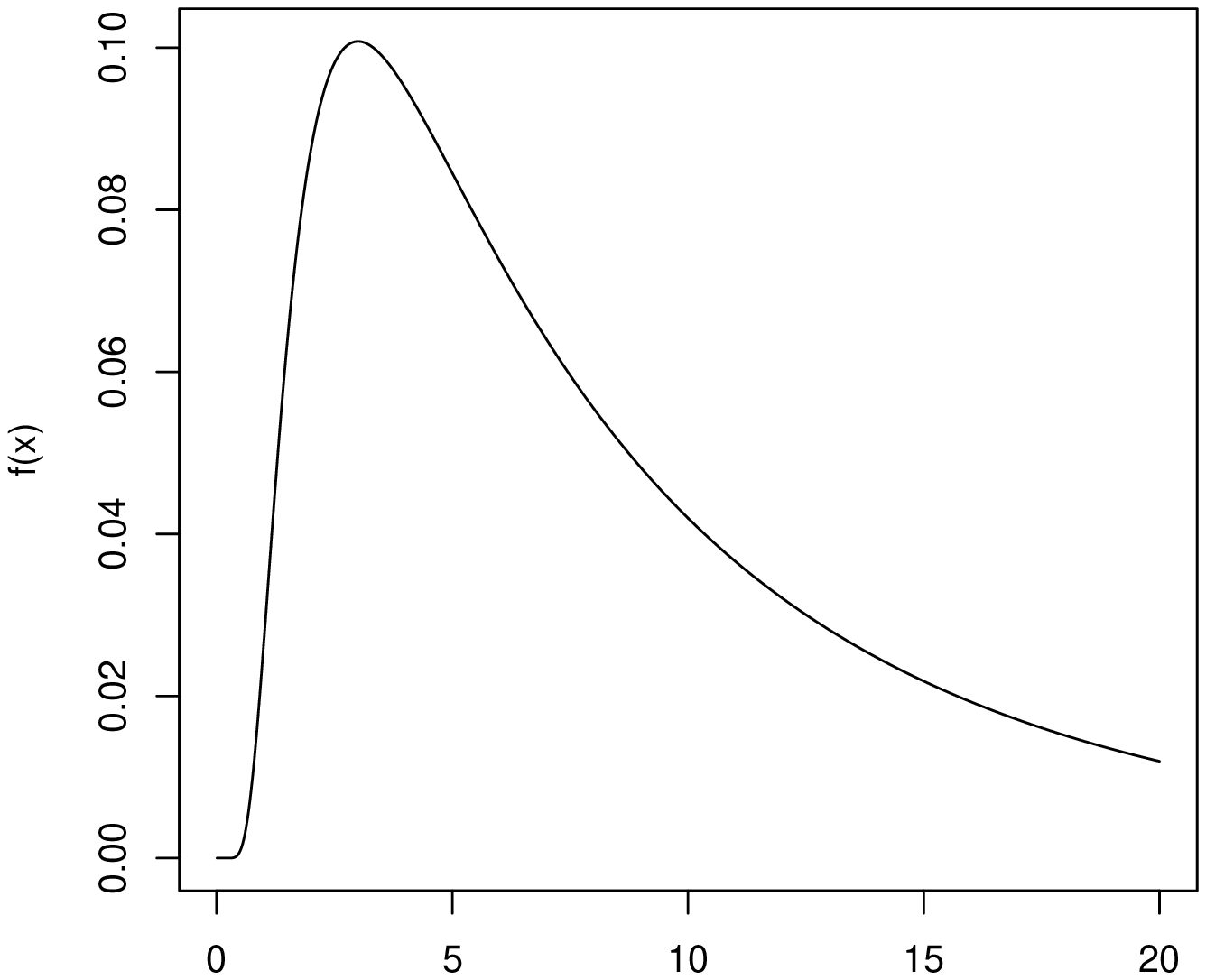}\label{}}~\subfigure[$m=6$]{\includegraphics[scale=0.25]{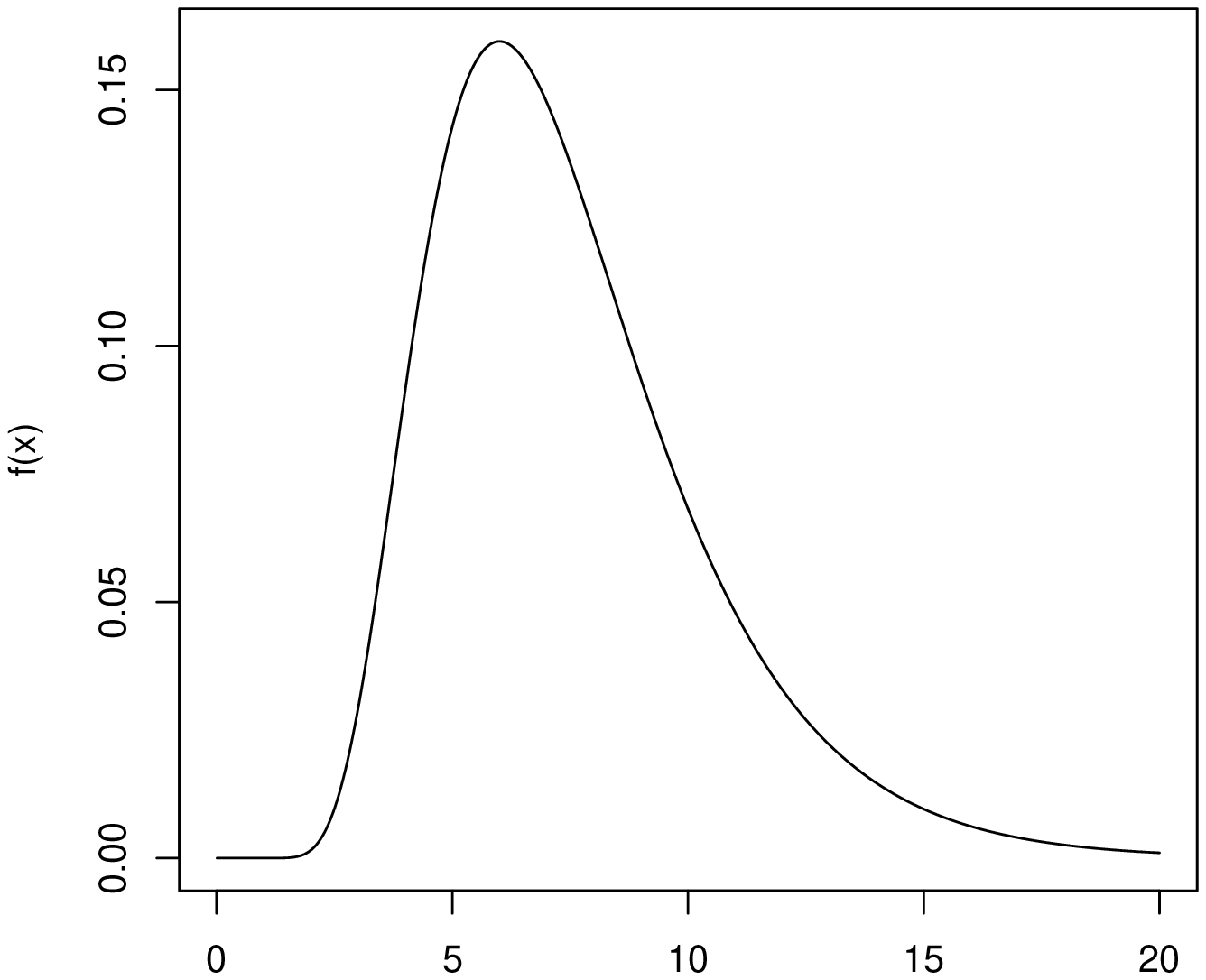}\label{}}
 	\caption{Unimodal BS densities with $\beta = 7$ and varying modes ($m$). \label{figmode}}
\end{figure}
A feature of the BS density in  \eqref{eqBSp} is that  its asymmetry changes according to the value of $m$,  as  can  be  seen  in  Figure \ref{figmode}. As $m$ increases, the density becomes more symmetric around $\beta$.

Now, following   the  idea   of  \cite{Balakrishnan:11},  we  define  mixtures of Birnbaum-Saunders  distributions of the form:
\begin{eqnarray}
f(y;\mathbf{p}, \balpha,\bbeta) &=& \sum_{j=1}^G p_j f_{T_j}(y;\alpha_j,\beta_j), \quad y \in \mathbb{R}_+, \label{eqFM-BS}
\end{eqnarray}
where $p_j$ is the mixing parameter of the jth sub-population which is constrained to be positive with the constraint $\sum_{j=1}^G p_j=1$, and  $f_{T_j}(	\cdot;\alpha_j,\beta_j)$ is   the pdf  of sub-population $j$ of the  $ {\rm BS}(\alpha_j,\beta_j)$ distribution, with $\alpha_j>0$, $\beta_j>0$, $j=1,\dots,G$. Moreover, $\mathbf{p}=(p_1,\ldots, p_G)^{\top}$, $\balpha = (\alpha_{1},\ldots,\alpha_{G})^{\top}$ and $\bbeta = (\beta_{1},\ldots,\beta_{G})^{\top}$. One of the $p_j$ is redundant because these probabilities add up to 1. We assume that the number of components $G$ is known and fixed. 
A  positive  random  variable $Y$,  with density (\ref{eqFM-BS}), is called a finite mixture of Birnbaum-Saunders (FM-BS) model, and  will be denoted by $ Y\sim\text{FM-BS}(\mathbf{p},\balpha,\bbeta)$. The pdf of the FM-BS can take different shapes as  can  be  seen in  Figure \ref{fig1fmbs}. Some properties of the FM-BS distribution can be derived by using the close relationship between the distribution $T_j$  and   normal   distribution.

\begin{theorem}\label{theo;1}
If $Y \sim {\rm FM}{\mbox -}{\rm BS}(\mathbf{p},\balpha,\bbeta)$, then
{\textcolor{black}{
\begin{itemize}
\item[(i)]  $c Y \sim {\rm FM}{\mbox -}{\rm BS}(\mathbf{p}, \balpha,c\bbeta)$, where $c \in \mathbb{R}_{+}$;
\item[(ii)] $Y^{-1} \sim {\rm FM}{\mbox -}{\rm BS}(\mathbf{p}, \balpha,\bbeta^{-1})$, where $\bbeta^{-1}=(1/\beta_1, \ldots, 1/\beta_G)^{\top}$;
\item[(iii)]  For  $\beta_1=\beta_2=\ldots=\beta_G=\beta$, $\beta/Y$ and $Y/\beta$  have  the  same  distribution;
\item[(iv)] The   cdf  of  $Y$  is $F_Y(y)=\sum_{j=1}^G p_j \Phi\big(a_y(\alpha_j,\beta_j)\big)$;
\item[(v)] If  $W=\log(Y)$, then    the   pdf  of   $W$ is
\begin{eqnarray*}
f_{W}(w;\mathbf{p}, \balpha,\bbeta) &=& \sum_{j=1}^G p_j f_{W_j}(w;\alpha_j,\gamma_j), \quad w \in \mathbb{R}_+, \label{eqFM-logBS}
\end{eqnarray*}
where  $f_{W_j}(w;\alpha_j,\gamma_j)=(1/2)\phi\big(\xi_2(w;\alpha_j,\gamma_j)\big)\xi_1(w;\alpha_j,\gamma_j)$,  with $\xi_2(w;\alpha_j,\gamma_j)=\frac{2}{\alpha_j}\sinh\Big(\frac{w-\gamma_j}{2}\Big)$ and  $\xi_1(w;\alpha_j,\gamma_j)=\frac{2}{\alpha_j}\cosh\Big(\frac{w-\gamma_j}{2}\Big)$, $\gamma_j=\log(\beta_j)$.
\end{itemize}
}}
\end{theorem}
As mentioned,  many   properties of  the  FM-BS  distribution    can  be  obtained by  using properties   of  the  normal  distribution,  and 	 from    other   results that come  from  the usual BS distribution  and  an  associated   distribution, such as the sinh-normal  distribution \citep{rieck1989}.
\begin{theorem}\label{teo1}
If $Y \sim {\rm FM}{\mbox -}{\rm BS}(\mathbf{p},\balpha,\bbeta)$, then
$$ \rm{E}(Y^s)=\sum_{j=1}^G p_j \exp(\gamma_j s)\left[ \frac{K_{(2s+1)/2}\big(\alpha^{-2}_j\big)+K_{(2s-1)/2}\big(\alpha^{-2}_j\big)}{2K_{1/2}\big(\alpha^{-2}_j\big)} \right], $$
where  $ K_\nu(\cdot)$ denotes  the modified Bessel function of the third kind.  Moreover
\begin{eqnarray*}
		\rm{E}(Y) &=& \sum_{j=1}^G p_j\beta_{j} \left(1 + \alpha_{j}^2/2 \right) \, \, \, {\rm and}\, \, \, \rm{E}(Y^2) =\sum_{j=1}^G p_j \beta_j^2\left(1 + 2 \alpha_j ^2 + 3 \alpha_j^4/2\right).
\end{eqnarray*}
\end{theorem}	

\begin{theorem}\label{teo1}
Let $Y \sim {\rm FM}{\mbox -}{\rm BS}(\mathbf{p}, \balpha,\bbeta)$. Then the mode (modes) and median of the FM-BS distribution  are obtained, respectively, by solving the  nonlinear equations with respect to $y$
\begin{itemize}
\item[] Mode: $\sum_{j=1}^G p_j \phi\big(a_y(\alpha_j,\beta_j)\big)  \left[a_y(\alpha_j,\beta_j) A^2_y(\alpha_j,\beta_j) + y^{-5/2}(y + 3\beta_j)/(4\alpha_j \beta_j^{1/2})\right] = 0$
\item[] Median: $\sum_{j=1}^G p_j \Phi\big(a_y(\alpha_j,\beta_j)\big) = 0.5$.
\end{itemize}
\end{theorem}	
Table \ref{mode_median} displays the mode and median of the FM-BS distribution based on different parametric choices. The values of the parameters $p_1$, $\alpha_1$, $\alpha_2$, $\beta_1$ and $\beta_2$, in Table \ref{mode_median} are chosen to demonstrate the unimodal and bimodal cases for the probability function of the mixture model. From Table \ref{mode_median}, we see that the mode is slightly affected by variation in the values of the mixing proportion $p_1$, for the unimodal and bimodal case. In addition, the median decreases when $p_1$ increases for the unimodal and bimodal cases.
\begin{table}[H]
\begin{center}
\caption{The mode(s) and median of the FM-BS.}\label{mode_median}
 		\vskip 3mm
 		\small{\begin{tabular}{l@{\hskip 1.1in}l@{\hskip 1.1in}c@{\hskip 1.1in}}
 				\hline
 				$\btheta = (p_1,\alpha_1,\alpha_2,\beta_1,\beta_2)$  & Mode(s) &  Median \\ \hline
 				(0.2,0.5,0.75,3,7)	&   2.8649          &  5.7670    \\
 				(0.3,0.5,0.75,3,7)	&   2.6698          &  5.1786   \\
 				(0.4,0.5,0.75,3,7)	&   2.5521          &  4.6549   \\ \\
 				(0.2,0.25,0.35,3,7)	&   2.9756, 3.9871  &  6.2635    \\
 				(0.3,0.25,0.35,3,7)	&   2.8938, 4.5233  &  5.7541    \\
 				(0.4,0.25,0.35,3,7)	&   2.8625, 4.9819  &  5.0735   \\   	   	   	   	   	
 				\hline
 \end{tabular}}
 \end{center}
 \end{table}
\begin{figure}[H]
	\centering
	{\includegraphics[scale=0.5]{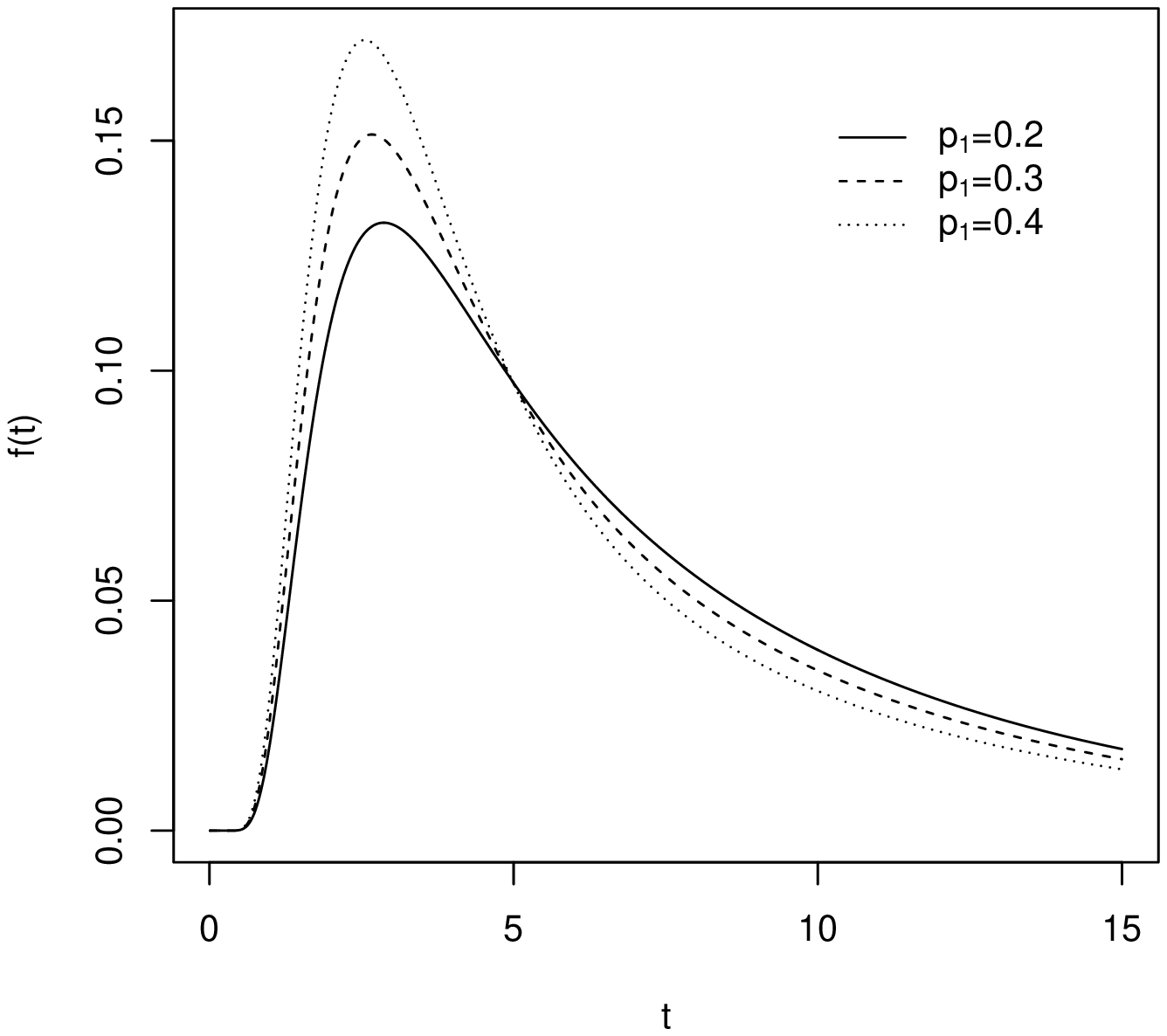}}~{\includegraphics[scale=0.5]{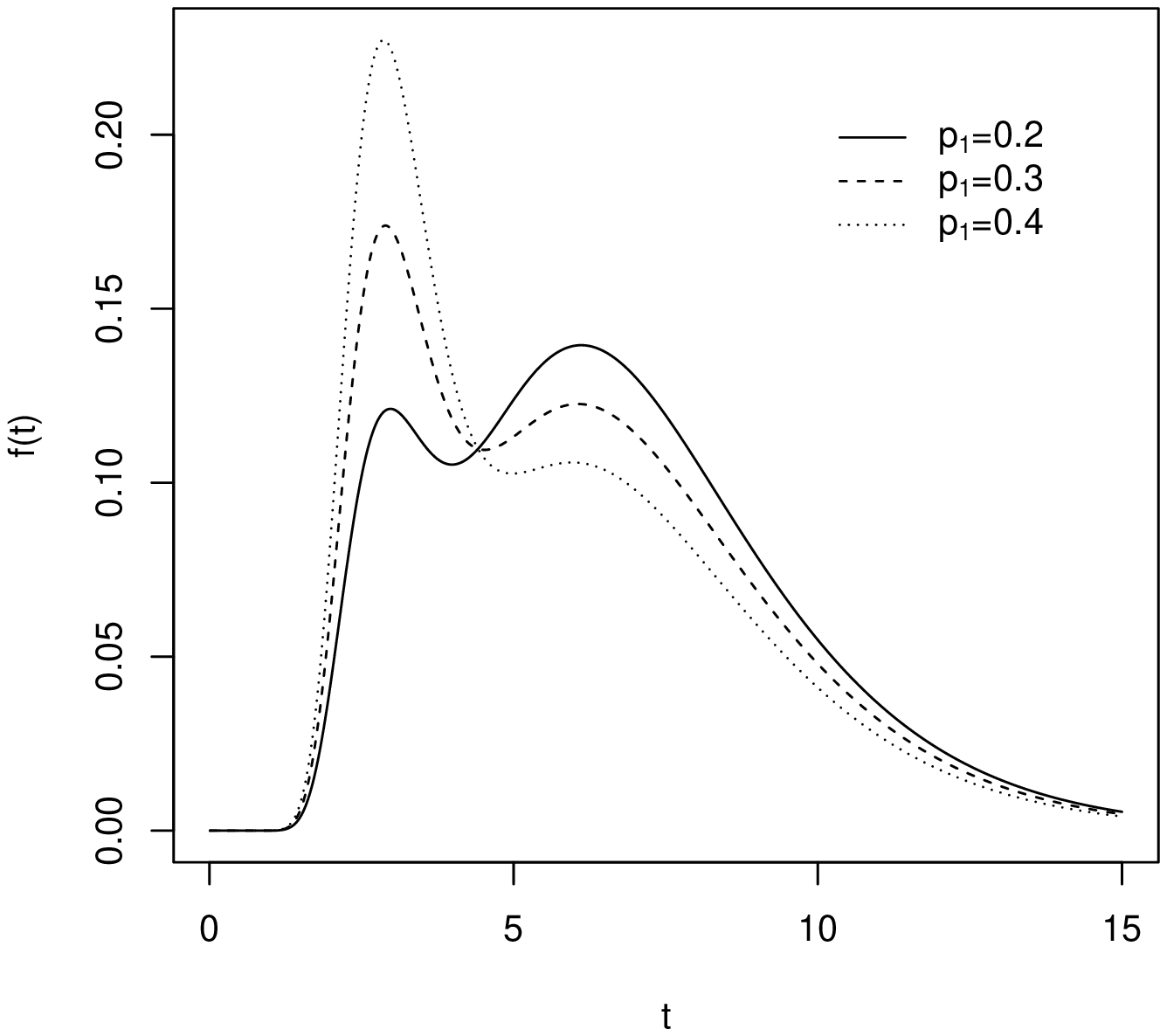}}	
	\caption{ Plots of the density function for some  values of $p_1$; (left panel) $(\alpha_1,\alpha_2,\beta_1,\beta_2)=(0.2,0.5,0.75,3,7)$ and (right panel) $(\alpha_1,\alpha_2,\beta_1,\beta_2)=(0.25,0.35,3,7)$. \label{fig1fmbs}}
\end{figure}
\begin{theorem}\label{sfANDhrf}
Let $Y \sim {\rm FM}{\mbox -}{\rm BS}(\mathbf{p},\balpha,\bbeta)$. Then,
\begin{itemize}
\item[(i)] the survival function (sf) and the hazard  function (hf)  of $Y$ are, respectively,
$$S_Y(y) = \sum_{j=1}^G p_j S_{T_j}(y) \quad \text{and} \quad 	h_{Y}(y) = \sum_{j=1}^G p_j f_{T_j}(y;\alpha_j,\beta_j)/S_Y(y),$$
where $S_{T_j}(y) =  1-\Phi\big(a_y(\alpha_j,\beta_j)\big)$;
\item[(ii)] for $G=2$ the hr function $h_Y(\cdot)$   satisfies
\begin{equation*}
\lim_{y \to \infty}h_Y(y) =   \begin{cases}
	 \frac{1}{2\alpha_1^2 \beta_1}, & \mbox{if } \alpha_2^2\beta_2 < \alpha_1^2\beta_1; \\
	 \frac{d}{2\alpha_1^2\beta_1}  + \frac{1 - d}{2\alpha_2^2\beta_2} , & \mbox{if } \alpha_2^2\beta_2 = \alpha_1^2\beta_1; \\
	 \frac{1}{2\alpha_2^2 \beta_2}, & \mbox{if } \alpha_2^2\beta_2 > \alpha_1^2\beta_1, \end{cases}
\end{equation*}
where $d = {p}\Big({p + (1-p)\exp(1/\alpha_2^2 - 1/\alpha_1^2)\alpha_1 \beta_1^{1/2}/(\alpha_2 \beta_2^{1/2})}\Big)^{-1}$.
\end{itemize}
\end{theorem}
\begin{proof}[\textbf{Proof:} ] Part (i)  is  easy to  prove.  For part (ii),  we  start  by considering the hazard function, which can also be written as
	$$h_Y(y) = w(y) h_{T_1}(y) + \big(1 - w(y)\big) h_{T_2}(y),$$
where $w(y)={pS_{T_1}(y)}/\big[pS_{T_1}(y)+(1-p)S_{T_2}(y)\big].$ Applying L'H\^{o}pital's rule for $w(y)$,  the derivative of  $w(y)$  can  be  expressed as
\begin{eqnarray*}
\frac{d}{dy}w(y) &=&\frac{p}{p + (1-p)\frac{\phi(a_y(\alpha_2,\beta_2))A_y(\alpha_2,\beta_2)}{\phi(a_y(\alpha_1,\beta_1))A_y(\alpha_1,\beta_1)}}.
	\label{eqDer}
\end{eqnarray*}
We have that if $y \rightarrow \infty$, then  $A_y(\alpha_2,\beta_2)/A_y(\alpha_1,\beta_1) \rightarrow \alpha_1 \beta_1^{1/2}/(\alpha_2 \beta_2^{1/2})$ and
\begin{equation*}
	\frac{\phi\big(a_y(\alpha_2,\beta_2)\big)}{\phi\big(a_y(\alpha_1,\beta_1)\big)} \rightarrow   \begin{cases}
	0, & \mbox{if } \alpha_2^2\beta_2 < \alpha_1^2\beta_1; \\
	\exp\left(1/\alpha_1^2 - 1/\alpha_2^2\right) , & \mbox{if } \alpha_2^2\beta_2 = \alpha_1^2\beta_1; \\
	\infty, & \mbox{if } \alpha_2^2\beta_2 > \alpha_1^2\beta_1. \end{cases}
\end{equation*}
So, from  the above results we obtain that as $y\rightarrow \infty$,
\begin{equation*}
	w(y) \rightarrow   \begin{cases}
	1, & \mbox{if } \alpha_2^2\beta_2 < \alpha_1^2\beta_1; \\
	d , & \mbox{if } \alpha_2^2\beta_2 = \alpha_1^2\beta_1; \\
	0, & \mbox{if } \alpha_2^2\beta_2 > \alpha_1^2\beta_1; \end{cases}
\end{equation*}
where $$d = {p}\left(    {p + (1-p)\exp(1/\alpha_2^2 - 1/\alpha_1^2)\alpha_1 \beta_1^{1/2}/(\alpha_2 \beta_2^{1/2})} \right)^{-1},$$
and hence part (ii) of Theorem \ref{sfANDhrf}  has been proved.
\end{proof}
\begin{figure}[t!]
	\centering
	{\includegraphics[scale=0.5]{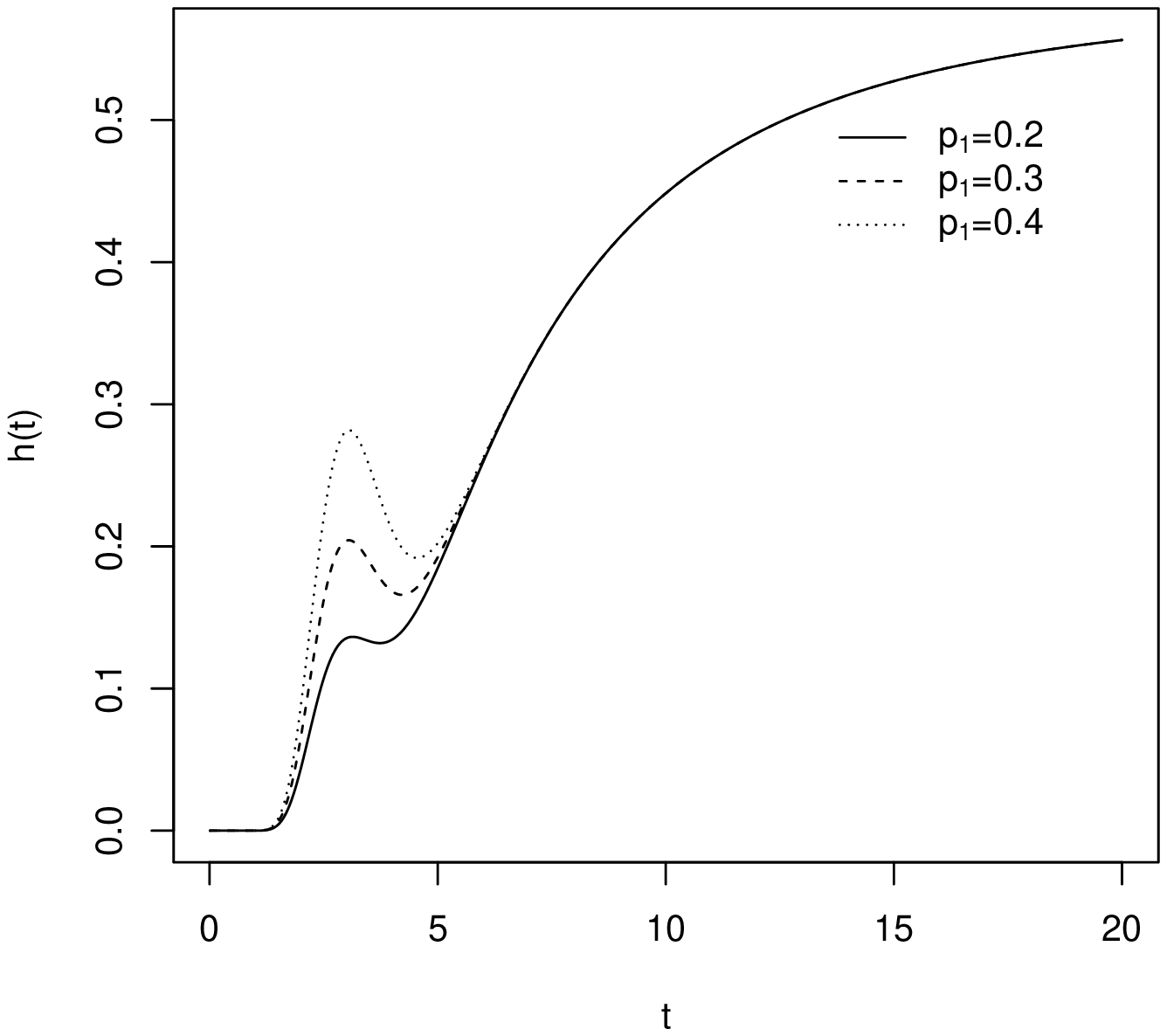}}~{\includegraphics[scale=0.5]{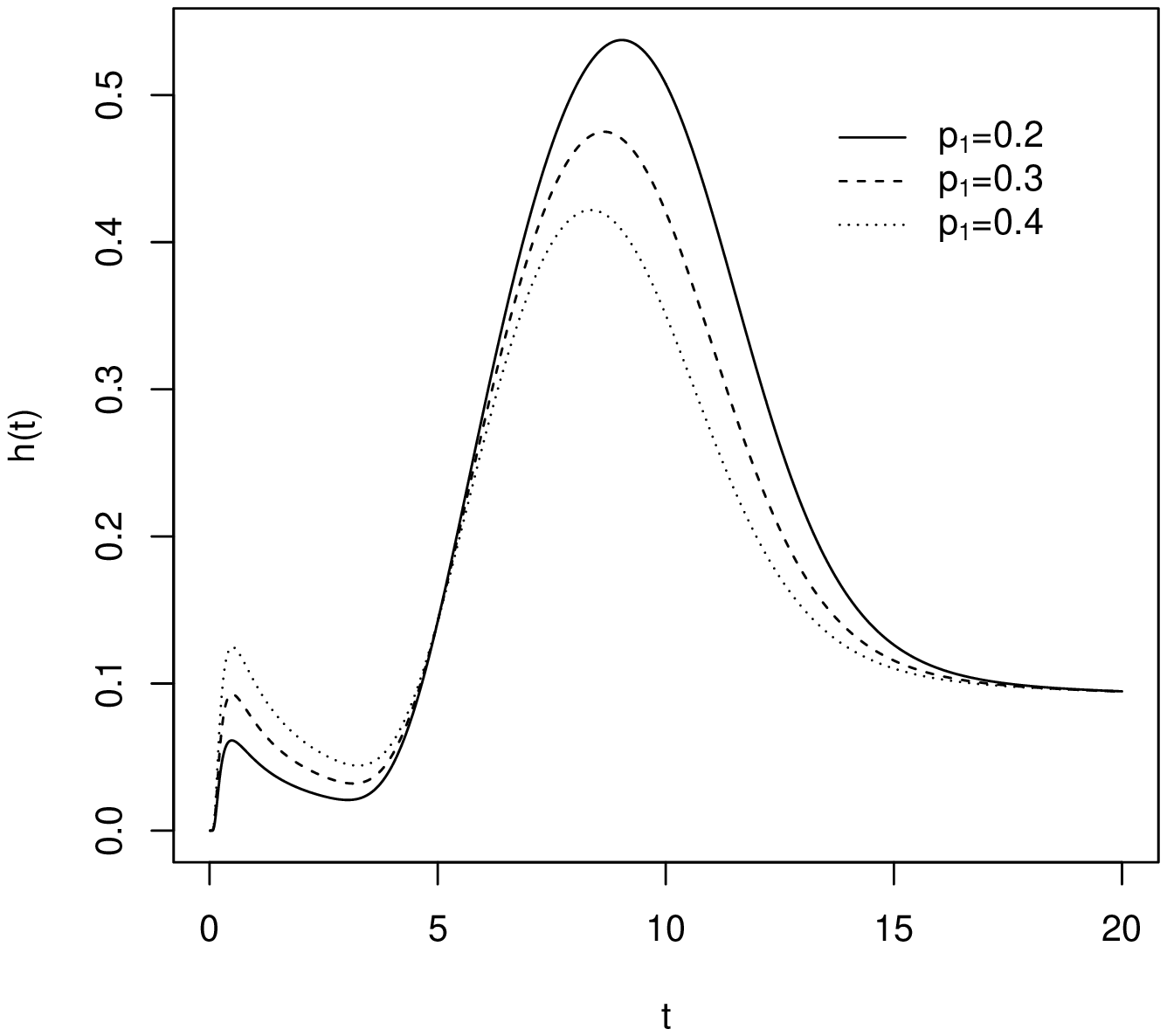}}\\
	\caption{Some hr functions  of the FM-BS. (left panel, unimodal) with $(\alpha_1,\alpha_2,\beta_1,\beta_2)=(0.25,0.35,3,7)$ and (right panel, bimodal)  with $(\alpha_1,\alpha_2,\beta_1,\beta_2)=(1.5,0.25,3,7)$. \label{fig2fmbs}}
\end{figure}
Figure \ref{fig2fmbs} displays  two different hr functions (unimodal and bimodal), considering three mixing proportions and fixed parameters $\alpha_1$, $\alpha_2$, $\beta_1$, $\beta_2$. Note that Figure \ref{fig2fmbs} (left panel) considers small $\alpha_1$ and $\alpha_2$ and Figure \ref{fig2fmbs} (right panel) considers large $\alpha_1$ and small $\alpha_2$.
\begin{figure}[t!]
\centering
{\includegraphics[scale=0.5]{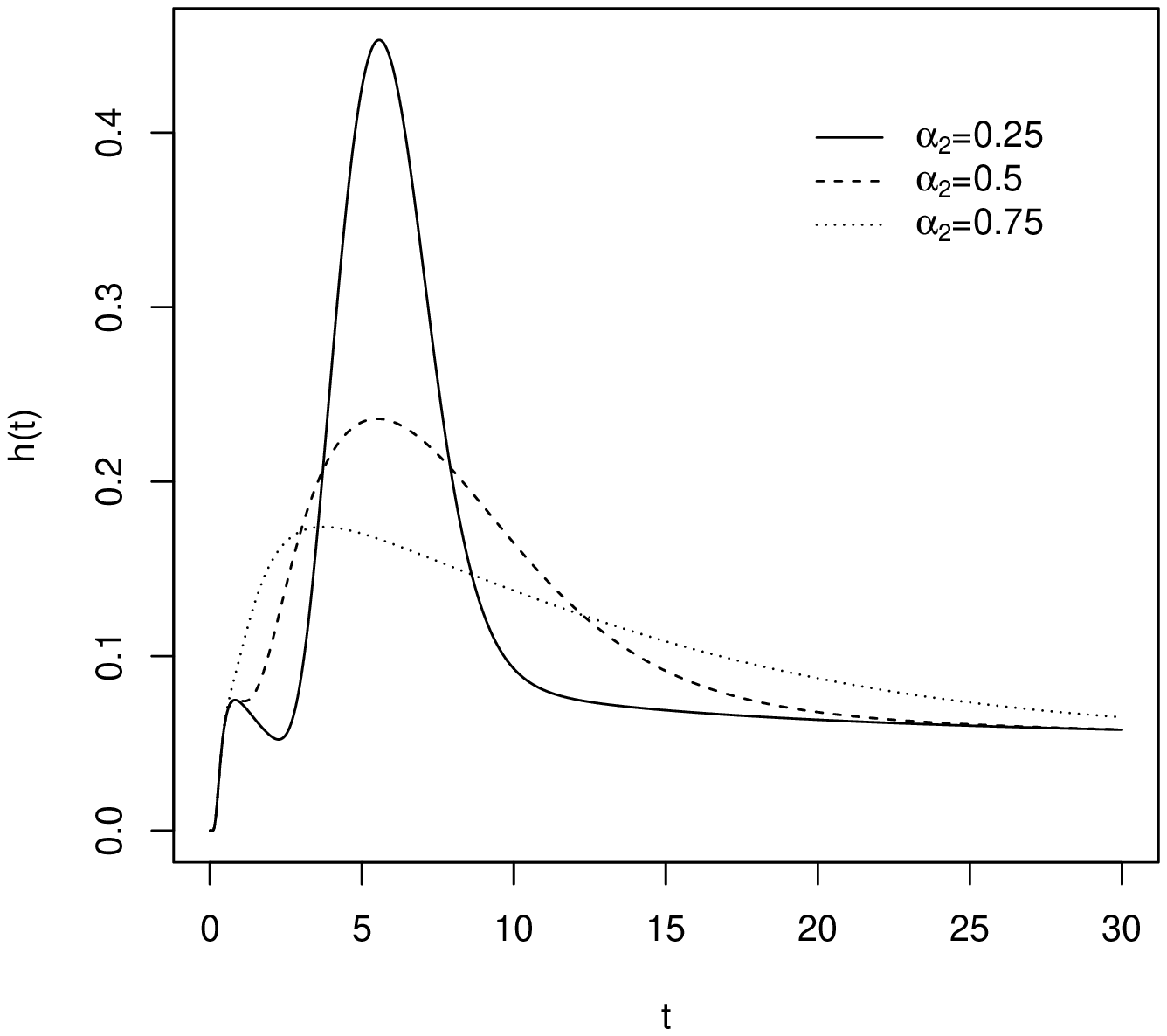}}~{\includegraphics[scale=0.5]{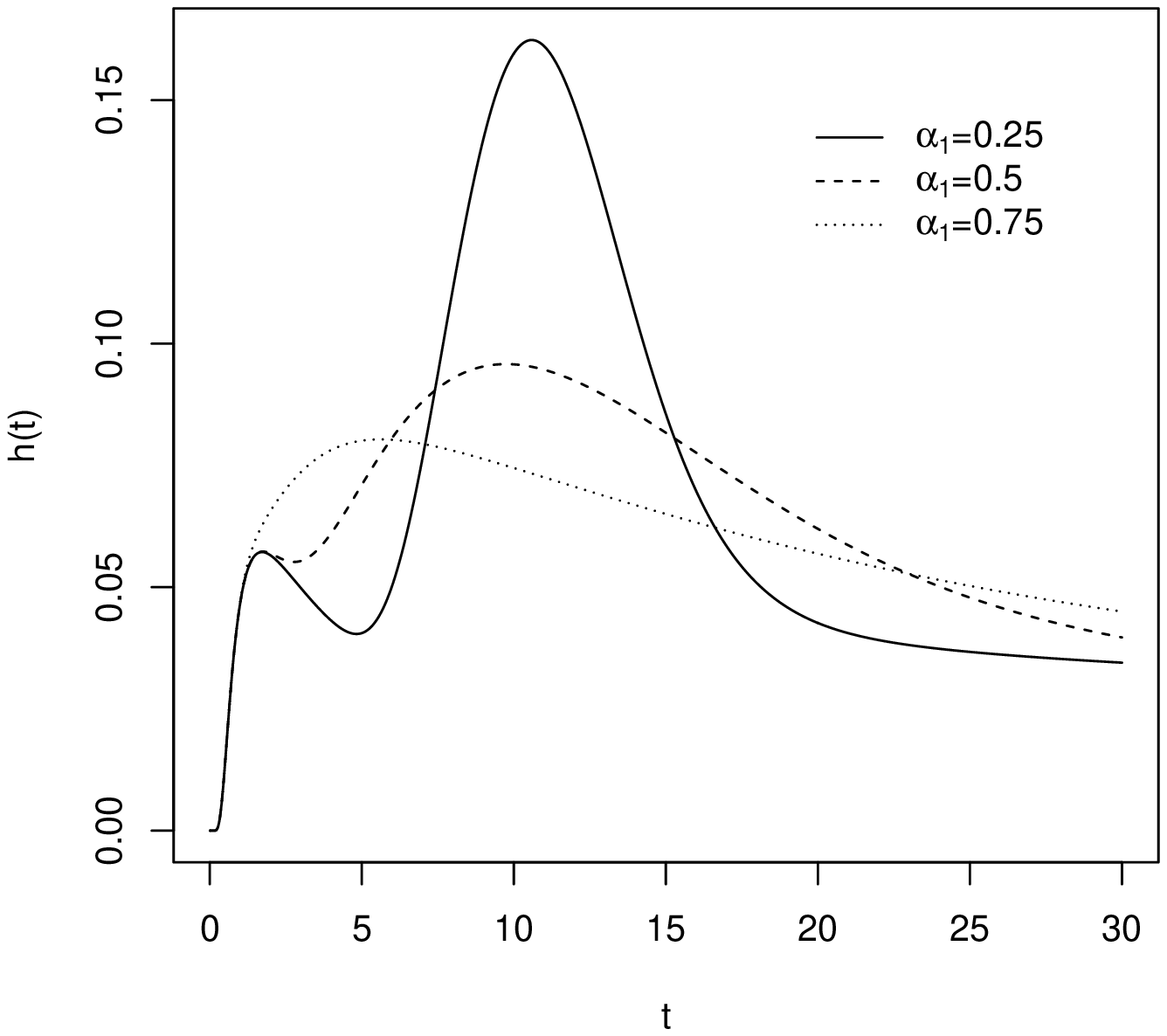}}\\
\caption{Some hr  functions  of the FM-BS. (left panel) with $\alpha_1=1.5$, $\beta_1=\beta_2=5$ and (right panel) with $\alpha_2=1.5$, $\beta_1=\beta_2=10$. In all cases, $p_1=0.4$. \label{fig2fmbs2}}
\end{figure}
Figure \ref{fig2fmbs2} shows the convergence of the hr function (unimodal and bimodal) for the  FM-BS model considering two conditions (the first and the last of Theorem \ref{sfANDhrf}) for different parameter combinations. Note that the left plot in Figure \ref{fig2fmbs2} satisfies the $\alpha_2^2 \beta_2  < \alpha_1^2 \beta_1$ condition and the right plot in Figure \ref{fig2fmbs2} satisfies the $\alpha_2^2 \beta_2  > \alpha_1^2 \beta_1$ condition.

Inspired   by \cite{hussaini1997},  we  estimate the stress-strength reliability $R = P(Y < X)$, when the random variables $X$ and $Y$ are independent and each has  a   FM-BS distribution.  Specifically, suppose that $X \sim {\rm FM}{\mbox -}{\rm BS}(\mathbf{p},\balpha,\bbeta)$ and $Y\sim {\rm FM}{\mbox -}{\rm BS}(\mathbf{q},\bgamma,\btheta)$ are independent such that   their pdfs  are   given  by
$$h_X(x) = \sum_{j=1}^{G_1} p_j f_{T_j}(x;\alpha_j,\beta_j)\quad \text{and} \quad h_Y(y) = \sum_{l=1}^{G_2} q_l g_{T_l}(y;\gamma_l,\theta_l),$$
respectively. Then, the stress-strength reliability $R$  can  be  expressed  as
\begin{equation}
R = \int_{0}^{\infty} \int_{0}^{x} h_Y(y) h_X(x) dy dx = \sum_{j=1}^{G_1} \sum_{l=1}^{G_2} p_j q_l R_{jl},
\label{stressReliability}
\end{equation}
where
\begin{eqnarray*}
R_{jl}& =& \int_{0}^{\infty} \int_{0}^{x}g_{T_l}(y;\gamma_l,\theta_l) f_{T_j}(x;\alpha_j,\beta_j)dydx=\int_{0}^{\infty} \phi\big(c_x(\alpha_j,\beta_j)\big)\Phi\big(c_x(\gamma_l,\theta_l) \big)C_x(\alpha_j,\beta_j)dx,
\end{eqnarray*}
for $j=1,\ldots,G_1$, $l=1,\ldots,G_2$.
The expressions for $ R_{jl}$,  in  \eqref{stressReliability},  can be obtained using numerical methods (e.g., using the \texttt{integrate()} function in the R software).

\subsection{Identifiability}
A very important concept associated with mixture models is identifiability, which the foundation for estimation problems. It's important as testing hypotheses and classification of random variables can only be discussed if the class of all finite mixtures is identifiable. This identifiability issue has been discussed by several authors, among them, \cite{teicher1967}, \cite{yakowitz1968} and \cite{chandra1977}.  In this section, we use the results from \cite{chandra1977} to show that the class of all finite mixing distributions relative to the ${\rm BS}(\alpha,\beta)$ distribution is identifiable, which we present briefly: Let $\psi$ be a transform associated with each $F_j \in \Psi$, where $\Psi$ is the class of distribution functions, having the domain of definition $D_{\psi_j}$ with linear map $M: F_j \rightarrow \psi_j$, $j=1,\ldots,G$. If there exists a total ordering ($\preceq$) of $\Psi$ such that:
\begin{itemize}
	\item[(i)] $F_1 \preceq F_2$, $(F_1,F_2 \in \Psi)$ implies $D_{\psi_1} \subseteq D_{\psi_2}$; and 
	\item[(ii)]	 for each $F_1 \in \Psi$, there exists some $s_1 \in D_{\psi_1}$ where $\psi_1(s) \neq 0$ such that $\lim\limits_{s \rightarrow s_1} \psi_2(s)/\psi_1(s) = 0$ for $F_1 < F_2$, ($F_1,F_2 \in \Psi$),
\end{itemize}
then the class of all finite mixing distributions is identifiable relative to $\Psi$.
\begin{proposition}
The class of all finite mixing distributions relative to the ${\rm BS}(\alpha,\beta)$ distribution for small $\alpha$ is identifiable.
\end{proposition}
\begin{proof}[\textbf{Proof:}] Let $T$  be a  positive random variable  that  follows  a  BS distribution with   cdf  as given in (1.1). By using Chandra's approach, we   verify   conditions (i)-(ii). First,  the $s$th moments of the $j$th BS component, considering small $\alpha < 0.5$ \citep[see,][]{rieck1989}, is  given by
\begin{equation}\label{fgm}
\psi_j(s) = E(T_j^s) = \beta_j^s \left[1 + \frac{\alpha_j^2 s^2}{2} + \frac{\alpha_j(s^4 - s^2)}{8} + \frac{\alpha_j(s^9 - 5s^4 + 4s^2)}{48}\right],\,\,\, j=1,2.	
\end{equation}
From \eqref{fgm} we can see $D_{\psi_1}=D_{\psi_2}=(-\infty,\infty)$. Now we satisfy the two conditions of the previous theorem as follows: \\
\noindent{\it Condition 1:} Ordering the family $\Psi$ of all cdf's lexicographically by $F_1 \leq F_2$ if $\alpha_1 > \alpha_2$ or  $\alpha_1 = \alpha_2$ implies that $\beta_1 > \beta_2$. So, we can simply prove that $D_{\psi_1} \subseteq D_{\psi_2}$. \\
	
\noindent{\it Condition 2:}  If we take $s_1 = + \infty$ and consider $\alpha_1 > \alpha_2$ and $\beta_1 > \beta_2$, we have
	\begin{eqnarray*}
		\lim\limits_{s \rightarrow +\infty} \frac{\psi_2(s)}{\psi_1(s)}&=&\lim\limits_{s \rightarrow +\infty} \frac{\beta_2^s \left[1 + \frac{\alpha_2^2 s^2}{2} + \frac{\alpha_2(s^4 - s^2)}{8} + \frac{\alpha_2(s^9 - 5s^4 + 4s^2)}{48}\right]}{\beta_1^s \left[1 + \frac{\alpha_1^2 s^2}{2} + \frac{\alpha_1(s^4 - s^2)}{8} + \frac{\alpha_1(s^9 - 5s^4 + 4s^2)}{48}\right]}\\
		&=& \lim\limits_{s \rightarrow +\infty}  \left\{\exp \left[ \log \left(\frac{\beta_2}{\beta_1}\right)\right]\right\}^s\\
		&=& 0,
	\end{eqnarray*}
	and hence the identifiability of the finite mixture ${\rm BS}(\alpha,\beta)$ has been proved.
\end{proof}

\section{Maximum  likelihood estimation}
In this section, we deal with the estimation problem for the model by using the EM  (expectation maximization) algorithm that finds the maximum likelihood (ML)  in the presence of missing data, which  was introduced by \citep{Dempster_EM} to  obtain the ML estimates  of the unknown parameter $\btheta$.   Moreover,  we  discuss starting values  and  the stopping rule of the EM algorithm, and how the standard errors were obtained.

\subsection{Parameter estimation via the EM algorithm}\label{sectionem}

Here,   we describe how to  implement  the expectation conditional maximization (ECM) algorithm \citep{Meng93} for the  ML  estimation of the parameters of the FM-BS model. The basic idea of the ECM is that the maximization (M) step of EM is replaced by several computationally simple conditional maximization (CM) steps. For notational convenience, let $\mathbf{y} = (y_1, \ldots, y_n)$ be the observations  vector and  $\mathbf{Z} = (\Z_1, \ldots, \Z_n)$ the set of latent component-indicators $\Z_j = (Z_{1j},\ldots, Z_{Gj})$, $j=1,\ldots, n$, whose values are a set of binary variables with
\[ Z_{ij} =
\begin{cases}
1  & \quad \text{if } \mathbf{Y}_j\,\text{belongs to group}\,\, j,\\
0  & \quad \text{} \text{otherwise,}\\
\end{cases}
\]
in  which  $\sum_{j=1}^{G} Z_{ij}=1$. Given the mixing probabilities{\color{black}{ $p_1, \ldots, p_G$}}, the component indicators $\Z_1,\ldots$ $, \Z_n$ are independent, with multinomial densities $f(\mathbf{z}_i)=p^{z_{i1}}_1p^{z_{i2}}_2...(1-p_1-\ldots-p_{G-1})^{z_{iG}}$, which  are  denoted  by  $\mathbf{Z}_i\sim {\rm Multinomial}(1;p_1\ldots,p_G)$. Note that $P(Z_{ij}=1)=1-P(Z_{ij}=0)=p_{j}$.
These  results are  used to build the ECM  algorithm,  since  the FM-BS model can be represented  hierarchically as
{\color{black}{
\begin{eqnarray}
Y_i|Z_{ij}=1 &\ind& {\rm BS}(\alpha_j, \beta_j),\label{HierqMIX1}\\
\textbf{Z}_i &\iid& {\rm Multinomial}(1,p_1,\ldots,p_G) \quad (i=1,\ldots,n).\label{HierqMIX3}
\end{eqnarray}}}
According to  (\ref{HierqMIX1}) and (\ref{HierqMIX3}), the complete data log-likelihood function of $\btheta=(p_1,\ldots, p_{G-1}, \alpha_1,$ $\ldots, \alpha_G,\beta_1,\ldots,\beta_G)$ given $(\mathbf{y},\mathbf{Z})$, aside from additive constants, is
\begin{eqnarray*}\label{log-likecomple}
\ell_c(\btheta|\mathbf{y},\mathbf{Z})&=&\sum_{i=1}^n\sum_{j=1}^G
z_{ij}\left[\log{p_j}- \log(\alpha_j) - \frac{1}{2}\log(\beta_j) + \log(y_i + \beta_j) -  \frac{1}{2}a^2_{y_i}(\alpha_j,\beta_j) \right].\nonumber \\
\end{eqnarray*}
Hence, the expected value of the complete data log-likelihood $\ell_c(\btheta|\mathbf{y},\mathbf{Z})$, evaluated with $\btheta = \btheta^{(k)}$, is  the $Q$-function given  by
$Q(\btheta|{\btheta}^{(k)})= \rm{E}\big(\ell_c(\btheta|\mathbf{y},\mathbf{Z}) | \mathbf{y},{\btheta}^{(k)}\big)$. To evaluate the $Q$-function, the necessary conditional expectations include ${\widehat z}_{ij}^{(k)}=\rm{E}\big(Z_{ij}|, y_i,\btheta^{(k)}\big)$.
By using known properties of conditional expectation, we obtain
\begin{eqnarray*}
{\widehat z}_{ij}^{(k)}&=&\widehat{p}^{(k)}_j f_{T_j}(y_i;{\alpha}^{(k)}_j,{\beta}^{(k)}_j)/\sum_{j=1}^G
	\widehat{p}^{(k)}_j f_{T_j}(y_i;{\alpha}^{(k)}_j,{\beta}^{(k)}_j).\label{zij}
\end{eqnarray*}
Therefore, the $Q$-function  can be written as
\begin{eqnarray*}
Q(\btheta|\widehat{\btheta}^{(k)})&=& \sum_{i=1}^n\sum_{j=1}^G \left[{\widehat z}_{ij}^{(k)} \left(\log{p_j}- \log \alpha_j - \frac{1}{2} \log \beta_j  + \frac{1}{2}\log(y_i + \beta_j) -  \frac{1}{2}a_{y_i}(\alpha_j,\beta_j)^2 \right) \right].
\end{eqnarray*}
In summary, the implementation of the ECM algorithm for ML estimation of the parameters of the FM-BS model proceeds as follows:\\
\noindent{\textbf{E-step}}: Given $\btheta = \widehat{\btheta}^{(k)}$, compute $\widehat{z}_{ij}^{(k)}$ for $\ii$ and $\jj$.\\
\noindent{\textbf{CM-step 1}}: Fix $\beta_j^{(k)}$ and update $\alpha_j^{(k)}$ and $p_j^{(k)}$ as
\begin{eqnarray}
\widehat{\alpha}_j^{2(k)} &=& \frac{\sumas \widehat{z}_{ij}^{(k)}a_{y_i}(1,\widehat{\beta}^{(k)}_j)}{\sumas \widehat{z}_{ij}^{(k)}}\, \, {\rm and}\, \,\,
{\widehat{p}}^{(k)}_j  =\frac{\sumas \widehat{z}_{ij}^{(k)}}{n}; \nonumber
\end{eqnarray}
\noindent{\textbf{CM-step 2}}: Fix $\widehat{\p}^{(k+1)}$, $\widehat{\balpha}^{(k+1)}$ and update $\widehat{\bbeta}^{(k+1)}$ using
\begin{eqnarray*}
\widehat{\bbeta}^{(k+1)}  = \arg\max_{\beta} \, \, Q(\widehat{\p}^{(k+1)}, \widehat{\balpha}^{(k+1)}, \bbeta, | \widehat{\btheta}^{(k)}).
\end{eqnarray*}

\subsection{Starting values of the EM algorithm} \label{notesimplementation}
It is well known that mixture models can provide a multimodal log-likelihood function. In this sense, the method of maximum likelihood estimation through the EM algorithm may not give maximum global solutions if the starting values are far from the real parameter values. Thus, the choice of starting values for the EM algorithm in the mixture context plays a key role in parameter estimation, since good  initial  values  for  the  optimization  process hasten or enable the convergence. The adopted starting values are summarized as follows:
\begin{itemize}
	\item Initialize the zero-one membership indicator $\widehat{Z}_j^{(0)} = \{\widehat{z}_{ij}^{(0)} \}_{i=1}^G$ according to a partitional clustering method.
	\item The initial values for mixing probabilities, component locations and scale can be specified as
	$${\widehat{p}}^{(0)}_j = \frac{\sumas \widehat{z}_{ij}^{(0)}}{n}, \quad \widehat{\alpha}_j^{2(0)} = \frac{\sumas \widehat{z}_{ij}^{(0)}a_{y_i}(1,\widehat{\beta}^{(0)}_j)}{\sumas \widehat{z}_{ij}^{(0)}},\quad \widehat{\bbeta}^{(0)}  = \arg\max_{\beta} \, \, Q(\widehat{\p}^{(0)}, \widehat{\balpha}^{(0)}, \bbeta, | \widehat{\btheta}^{(0)}).$$
\end{itemize}
It is well known that the success of EM-type  algorithms largely depends on the initialization values, so they have some limitations. For instance, label switching can get trapped at a local maximum or converge to the boundary of the parameter space. Unfortunately, when the partitions provided  by  partitional clustering methods, for example $k$-means \citep{Basso2010} and  $k$-medoids \citep{Kaufman1990} algorithms, are used to initialize the EM-algorithm, the final estimates may  change every time the algorithm is executed. As the EM-algorithm inherits the random initialization of the partitional clustering algorithm, we recommend using an algorithm that always gives the same initial values.   Following \cite{bagnato2013},  we used the $k$-bumps algorithm; an algorithm that always provides the same final partition.\\ 

\noindent \textbf{Remark 1.}
\begin{enumerate}
\item[(a)] \textit{The $k$-bumps algorithm can be summarized via the following steps:}
	\begin{itemize}
		\item[-] \textit{detect $k$ bumps $B_j$ of the observed data $\y$;}
		\item[-] \textit{find the maximum point $m_j^{(0)}$ for each $B_j$, where $m_j$ is the mode, for $j=1,\ldots,G$;}
		\item[-] \textit{assign each observation to the cluster with the closest maximum point.}
	\end{itemize}
\textit{After obtaining $m_j^{(0)}$, $\beta_j^{(0)}$ can be calculated. Using equation (2.3), the initial value  of the parameter $\alpha_j^{(0)}$ (the shape parameter) can be estimated. More details about the $k$-bumps algorithm are available in \cite{bagnato2013}.}
\item[(b)]  \textit{Only for the $k$-means and $k$-medoids algorithm, for each group $j$, we  utilize the modified moment estimates proposed by \cite{Ng:03}, which is implemented in the function} \verb"mmmeth()" \textit{in the {\bf R package} \textbf{bssn} \citep{rocio2015} to obtain the initial values for $\alpha$ and $\beta$}.
\end{enumerate}

\subsection{Stopping rule} \label{stoppingrule}
To assess the convergence of the EM algorithm,  the two most useful ways of confirming convergence are: (i) the difference between two successive log-likelihood values is less than a user-specified error tolerance; or  (ii) all parameter estimates are changing by a very small degree. As suggested by \cite{Andrews2011}, we adopt the Aitken acceleration-based stopping criterion \citep[][Chap. 4.9]{McLachlan2008}:
\begin{equation}
|\ell^{(k+1)} - \ell^{(k+1)}_{\infty}|<\varepsilon,
\label{stoppingrule}
\end{equation}
 deciding when to terminate computations
where $\ell^{(k+1)}$ is the observed log-likelihood evaluated at $\btheta^{(k+1)}$, $\varepsilon$ is the desired  tolerance ($\varepsilon = 10^{-6}$ will be used to decide when to terminate computations),  and the asymptotic estimate of the log-likelihood at iteration is \citep[see,][]{Bohning1994}
$$\ell^{(k+1)}_{\infty} = \ell^{(k)} +  \big(\ell^{(k+1)} - \ell^{(k)}\big)/(1-c^{(k)}),$$
with $c^{(k)}$ denoting  the Aitken's acceleration at the $k$th iteration, given by $c^{(k)}=\displaystyle\frac{\ell^{(k+1)}-\ell^{(k)}}{\ell^{(k)} - \ell^{(k-1)}}$.  Assuming convergence to the ML   estimator  $\widehat{\btheta}$, also $\ell^{(k+1)}_{\infty}$ is the asymptotic estimate of the log-likelihood at iteration $k+1$. Note that the above procedure is also applicable to the simple case $(G = 1)$ by treating $Z_{ij} = 1$.

\subsection{Standard error approximation}\label{imtheo}

This section presents an outline of the standard errors of the ML estimates from the FM-BS model, which are obtained in a simple way by differentiating the log-likelihood function twice and obtaining the inverse. However, this is  somewhat complex to carry out. By  assuming   the  usual   regularity  conditions, these  guarantee that the ML   estimates  solve the gradient equation  and that the Fisher information exists  according to  \cite{louis1982}. So, the variance estimates  are  obtained   from  the diagonal of the inverse of the empirical information matrix, defined as:
\begin{equation}
\mathbf{I}_o = \sum_{i=1}^{n} s(y_i|\btheta)s^{\top}(y_i|\btheta) - n^{-1}S(\mathbf{y}|\btheta)S^{\top}(\mathbf{y}|\btheta),
\label{IMapprox}
\end{equation}
where $S(\mathbf{y}|\btheta)=\sum_{i=1}^{n} s(y_i|\btheta)$, with  $s(y_i|\btheta) = {\partial \log f(y_i|\btheta)}/{\partial \btheta}$  being  the empirical score
function for the $i$th individual. Substituting the ML estimates $\widehat{\btheta}$ by $\btheta$ in  (\ref{IMapprox}), $\mathbf{I}_{o}$   reduces to
\begin{equation}\label{oim}
\mathbf{I}_{o}=\sumas \widehat{\mathbf{s}}_i
\widehat{\mathbf{s}}^{\top}_i,
\end{equation}
where $\widehat{\mathbf{s}}_i$ is an individual score vector given by
$\widehat{\mathbf{s}}_i = (\widehat{s}_{i,p_1}, \ldots, \widehat{s}_{i,p_{G-1}}, \widehat{s}_{i,\alpha_1},\ldots,\widehat{s}_{i,\alpha_G},\widehat{s}_{i,\beta_1}, \ldots, \widehat{s}_{i,\beta_G})^{\top}.$
Explicit expressions for the elements of $\widehat{\mathbf{s}}_i$ are given by
\begin{eqnarray*}
\widehat{s}_{i,p_j}&=& \frac{f_{T_j}(y_i, \alpha_j,\beta_j)-f_{T_G}(y_i, \alpha_G,\beta_G)}{f(y_i;\mathbf{p},\balpha,\bbeta)},\,\, \, \,  \widehat{s}_{i,\alpha_j}= \frac{p_j D_{\alpha_j}\big(f_{T_j}(y_i;\alpha_j,\beta_j)\big)}{f(y_i;\mathbf{p}, \balpha,\bbeta)}, \\ \widehat{s}_{i,\beta_j} &=&\frac{p_j D_{\beta_j}\big(f_{T_j}(y_i;\alpha_j,\beta_j)\big)}{f(y_i;\mathbf{p}, \balpha,\bbeta)}, \, j=1,\dots,G,
\end{eqnarray*}
where
$D_{\alpha_j}\big(f_{T_j}(y_i;\alpha_j,\beta_j)\big) = {\partial f_{T_j}(y_i;\alpha_j,\beta_j) }/{\partial \alpha_j}$  and $D_{\beta_j}(f_{T_j}\big(y_i;\alpha_j,\beta_j)\big) = {\partial f_{T_j}(y_i;\alpha_j,\beta_j)}/{\partial \beta_j}$.
For simplicity of notation, we omit the index $i$ in the expressions without causing any confusion:
$$
D_{\delta_j}\big(f_{T_j}(y;\alpha_j,\beta_j)\big) = \phi\big(a_y(\alpha_j,\beta_j)\big) \left[\frac{\partial A_y(\alpha_j,\beta_j) }{\partial \delta_j} - \frac{\partial a_y(\alpha_j,\beta_j)}{\partial \delta_j } a_y(\alpha_j,\beta_j)A_y(\alpha_j,\beta_j)  \right],
$$
where $\delta_j = \alpha_j, \beta_j$ {\rm and}
$$\frac{\partial a_y(\alpha_j,\beta_j)}{\partial \alpha_j} = -\frac{1}{\alpha_j}a_y(\alpha_j,\beta_j), \quad \frac{\partial A_y(\alpha_j,\beta_j) }{\partial \alpha_j} = -\frac{1}{\alpha_j}A_y(\alpha_j,\beta_j),$$
$$\frac{\partial a_y(\alpha_j,\beta_j)}{\partial \beta_j} = -\frac{1}{2 \alpha_j \beta_j}\left(\sqrt{\frac{y}{\beta_j}} + \sqrt{\frac{\beta_j}{y}}\right), \quad \frac{\partial A_y(\alpha_j,\beta_j)}{\partial \beta_j} = \frac{y^{-3/2}(\beta_j - y)}{4\alpha_j \beta_j^{1/2}}.$$
Standard errors of $\widehat{\btheta}$ are extracted from the square root of the diagonal elements of the inverse of equation (\ref{oim}). The information based approximation in that equation is asymptotically applicable. 
\begin{figure}[H]
\centering
{\includegraphics[scale=0.35]{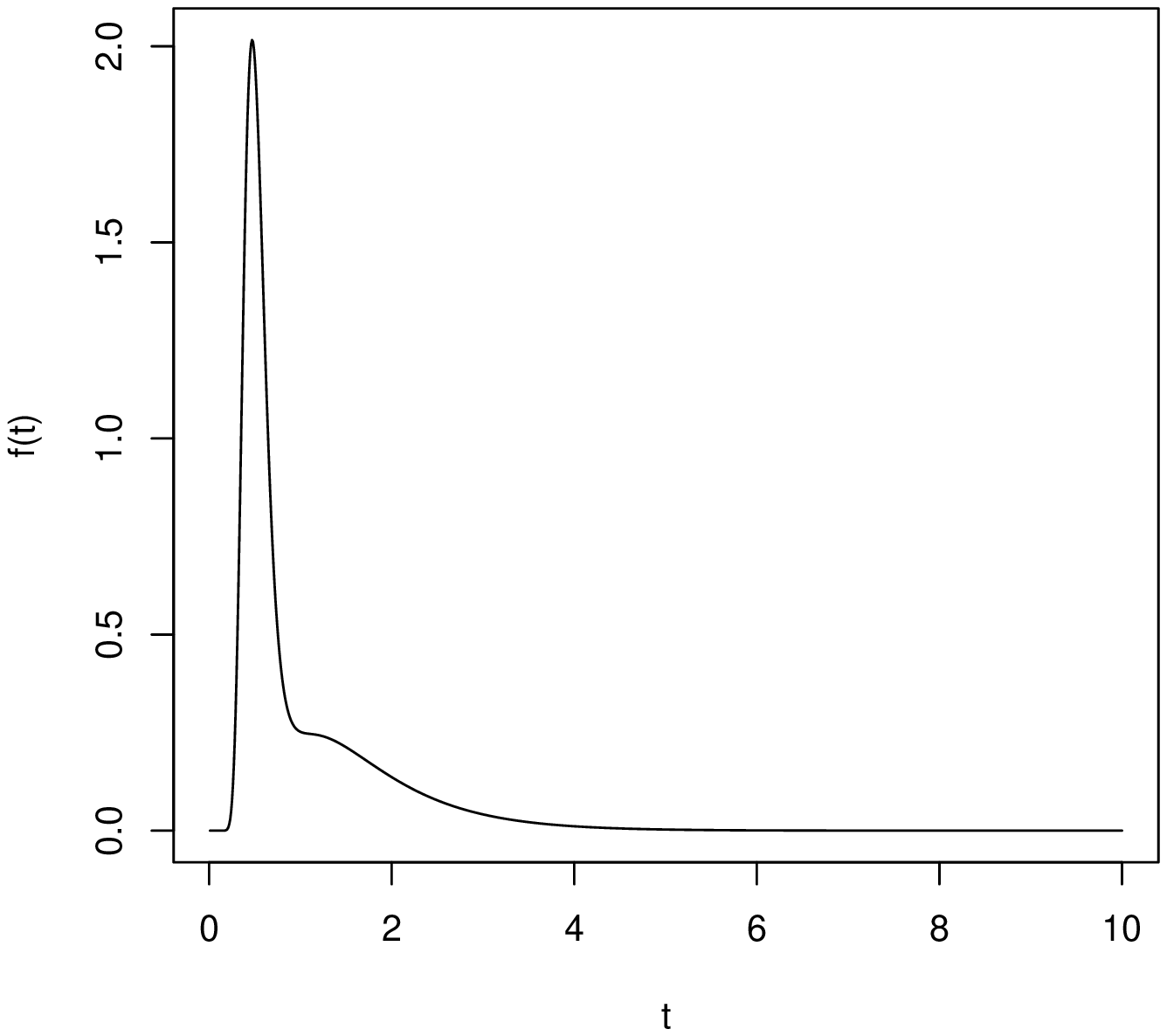}}~{\includegraphics[scale=0.35]{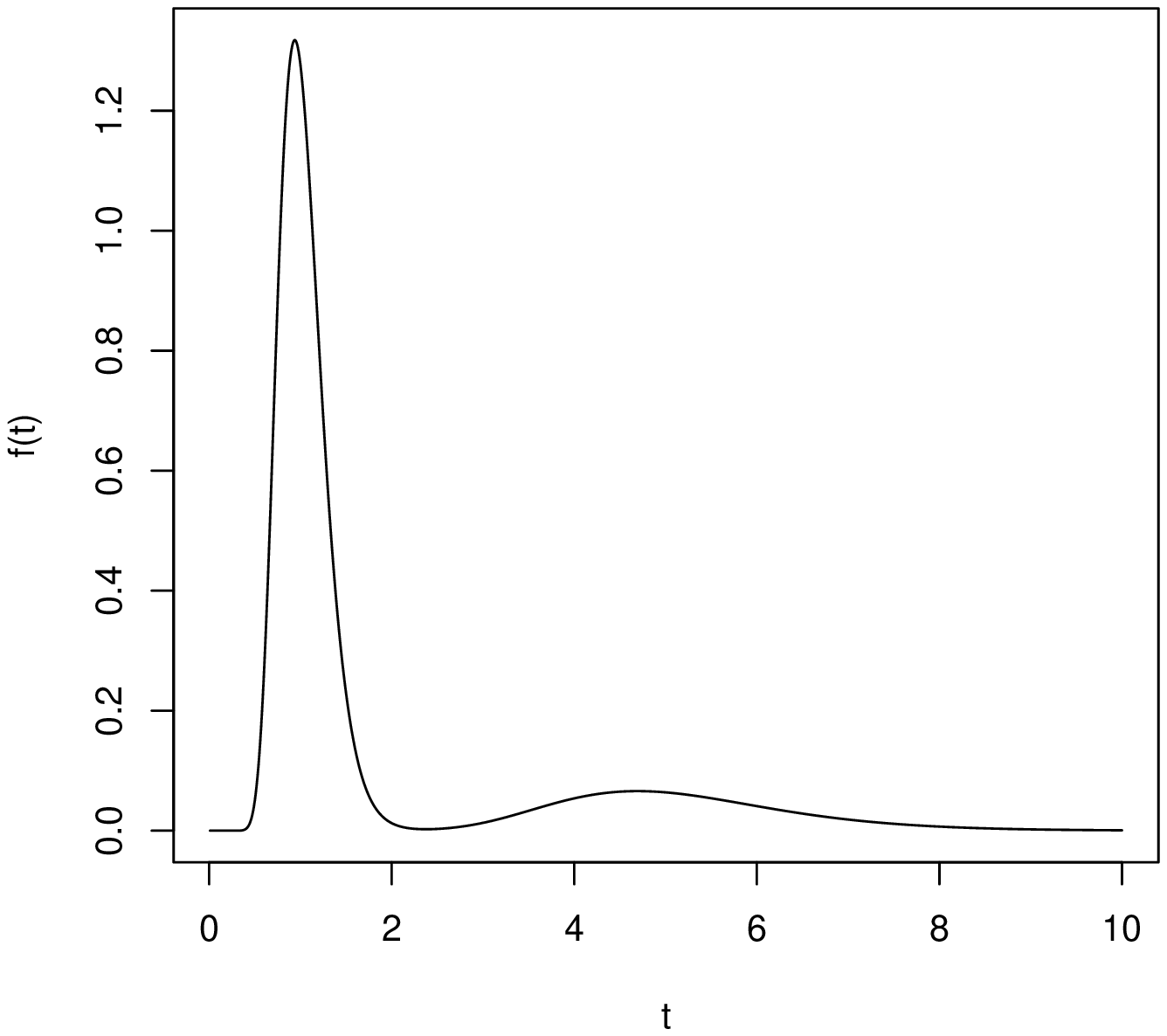}}\\
	\caption{Target mixture densities from which the data were simulated: For (left panel) $\btheta = (0.6,0.25,0.5,0.5,1.5)$ and  for (right panel) $\btheta =(0.8,0.25,0.25,1.0,5.0)$.  \label{fig-apliAIS1}}
\end{figure}
\section{Simulation study}
In this section, we run a simulation study  to evaluate the performance of the ML estimators with different partitional clustering methods for initialization of the  EM algorithm proposed in  Section 4. To perform these numerical experiments, we used the statistical computing environment R \citep{Rmanual.2016}. Specifically, the goals are to evaluate the accuracy of the estimates based on the EM algorithm for the FM-BS models and evaluate the consistency of the standard errors of the estimates (Study 1). Another goal is to show that our proposed EM algorithm estimates  provide good asymptotic properties (Study 2).
In all cases,  the simulation data were artificially generated from the following models with two components ($G=2$):
\begin{equation}
\label{defG=2}
f(y;\mathbf{p}, \balpha,\bbeta) = p_1f_{T_1}(y;\alpha_1,\beta_1) + (1-p_1)f_{T_2}(y;\alpha_2,\beta_2).
\end{equation}
One thousand random samples of sample size $n=75, 100, 500, 1000$ and $5000$ were generated for the  FM-BS model under scenarios of poorly separated  (PS) components  and
well separated (WS) components
\begin{itemize}
	\item[] Scenario 1: $p_1=0.6$, $\alpha_1=0.25$, $\alpha_2=0.50$, $\beta_1=0.5$, $\beta_2=1.5$  (PS components);
	\item[] Scenario 2: $p_1=0.8$, $\alpha_1=0.25$, $\alpha_2=0.25$, $\beta_1=1.00$, $\beta_2=5$ (WS components).
\end{itemize}
Our analyses were performed with a 3.40GHz Intel Core i7 processor with 31.9GB of RAM. The R code for the $k$-bumps algorithm is available in \cite{bagnato2013}, but we adapted it to our context, while for the $k$-medoids algorithm we used the {\bf R package} \texttt{ClusterR}. Figure \ref{fig-apliAIS1} (left panel and right panel) shows the mixture densities generated using Scenarios 1 and 2, respectively.
\begin{table}[t!]
	\begin{center}
		\caption{Comparison of average CPU time (in seconds) for FM-BS model under various sample sizes. The log-likelihood values are in parentheses.}\label{CPUtimetable}
		\vskip 3mm
		\scriptsize{\begin{tabular}{c@{\hskip 0.25in}c@{\hskip 0.25in}c@{\hskip 0.25in}c@{\hskip 0.25in}|c@{\hskip 0.25in}c@{\hskip 0.25in}c@{\hskip 0.25in}c@{\hskip 0.25in}c}
				\hline
				n    & $k$-means   &  $k$-medoids &  $k$-bumps & $k$-means   &  $k$-medoids &  $k$-bumps\\ \hline
				& \multicolumn{3}{c}{$\btheta=(0.6,0.25,0.5,0.5,1.5)$}  &  \multicolumn{3}{c}{$\btheta=(0.8,0.25,0.25,1.0,5.0)$} \\ \hline
				75   &  0.0819   & 0.0606     & 0.9300      &  0.0173   & 0.0209     & 0.8203\\
				& (-52.269) & (-51.8858) & (-39.9022)  & (-94.3336)& (-94.0662) & (-61.3887) \\
				100  &  0.0860   & 0.0700     & 1.0397      &  0.0159   & 0.0261     & 0.8106\\
				& (-72.299) & (-71.5438) & (-53.9624)  &(-125.8039)& (-130.7341)& (-81.1129) \\
				500  &  0.2652   & 0.2229     & 2.1783     &  0.0470   & 0.1017     & 1.4188  \\
				&(-366.1519)& (-358.7772)& (-274.4587) &(-623.2053)& (-645.317) & (-420.7042)\\
				1000 &  0.493    & 0.5403     & 3.4901      &  0.0891   & 0.2829     & 2.1967 \\
				&(-740.8493)& (-698.7736)& (-555.5299) &(-1252.646)& (-1263.654)& (-840.9494)\\
				5000 &  2.7449   & 5.3584     & 15.3325     &  0.3996   & 4.1035     & 8.7529 \\
				&(-3673.296)& (-3531.957)& (-2777.863) &(-6317.537)& (-6327.193)& (-4228.801)\\		
				\hline
\end{tabular}}
\end{center}
\end{table}
This is the first step to ensure that the estimation procedure works satisfactorily. Table \ref{CPUtimetable} shows a comparison of the average CPU times for the FM-BS model under different sample sizes and considering $k$-means, $k$-medoids and $k$-bumps algorithms in  order to obtain the initial values.

\subsection{Study 1: Parameter recovery and consistency of the standard errors of the estimates}
In this section the goal is to show the ML estimation of $\btheta=(p_1,\alpha_1,\alpha_2, \beta_1,\beta_2)^{\top}$ through  the EM algorithm   considering  the stopping criterion   given in  (\ref{stoppingrule}).  The mean values of the estimates across the 1000 Monte Carlo samples are computed, and the results are presented in Tables \ref{sim1escenario1} and \ref{sim1escenario2}, where the ML estimates of the parameters are close to the true values and become closer as the sample size increases. These tables show  the information matrix standard error (IM SE), the Monte Carlo standard deviation (MC Sd) and  the coverage probability,  reported   to  examine the consistency of the approximation method given in Subsection \ref{imtheo} for the standard errors (SE) of the ML estimates of parameters $\btheta$.  From  these tables, we note that the standard errors provide relatively close results (MC Sd and IM SE), which indicates that the proposed asymptotic approximation for the variances of the ML estimates is reliable. Moreover, the tables report the coverage probability (COV) and the percentage of coverage of the resulting 95\% confidence intervals (CI)  assuming asymptotic normality. The COV is defined as $COV(\widehat{\theta})=(1/m)\sum_{j=1}^{m} I (\theta \in [\widehat{\theta}_L,\widehat{\theta}_U])$, where $I$ is the indicator function such that $\theta$ lies in the interval $[\widehat{\theta}_L,\widehat{\theta}_U]$, with $\widehat{\theta}_L$ and $\widehat{\theta}_U$ being the estimated lower and upper bounds of the 95\% CI, respectively. The COV for the parameters is quite stable for both scenarios, which indicates that the proposed asymptotic approximation for the variance estimates of the ML estimates is reliable. This can also be seen in the COV of parameters, since in general a confidence interval above 95\% coverage is maintained for each parameter.
\begin{table}[H]
	\begin{center}
		\caption{Study 1: Mean fit of the FM-BS model based on different initialization algorithms and samples sizes when  $\btheta=(0.6,0.25,0.5,0.5,1.5)$.}\label{sim1escenario1}
		\vskip 3mm
		\tiny{\begin{tabular}{c@{\hskip 0.12in}c@{\hskip 0.12in}|c@{\hskip 0.12in}c@{\hskip 0.12in}c@{\hskip 0.12in}c@{\hskip 0.12in}c@{\hskip 0.12in}|c@{\hskip 0.12in}c@{\hskip 0.12in}c@{\hskip 0.12in}c@{\hskip 0.12in}c@{\hskip 0.12in}|c@{\hskip 0.12in}c@{\hskip 0.12in}c@{\hskip 0.12in}c@{\hskip 0.12in}c@{\hskip 0.12in}c@{\hskip 0.12in}c}
		 	\hline  \rule{0pt}{3ex}
				n  &Measure& $\widehat{p}_1$ & $\widehat{\alpha}_1$  &  $\widehat{\alpha}_2$ &  $\widehat{\beta}_1$ &  $\widehat{\beta}_2$ & $\widehat{p}_1$ & $\widehat{\alpha}_1$  &  $\widehat{\alpha}_2$ &  $\widehat{\beta}_1$ &  $\widehat{\beta}_2$ & $\widehat{p}_1$ & $\widehat{\alpha}_1$  &  $\widehat{\alpha}_2$ &  $\widehat{\beta}_1$ &  $\widehat{\beta}_2$ \\[0.1cm] \hline \rule{0pt}{3ex}
				&       &            \multicolumn{5}{c}{$k$-means}   &            \multicolumn{5}{c}{$k$-medoids}   &            \multicolumn{5}{c}{$k$-bumps}\\[0.1cm] \hline \rule{0pt}{3ex}
				75   & Mean  & 0.6253 & 0.2592 & 0.4380  & 0.5154  & 1.5566 & 0.6282 & 0.2601 & 0.4433  & 0.5167  & 1.5653 & 0.6017 & 0.2456 & 0.4613  & 0.5051  & 1.5416\\
				& IM SE & 0.0940 & 0.0419 & 0.1280  & 0.0276  & 0.2602 & 0.0965 & 0.0427 & 0.1472  & 0.0290  & 0.2859 & 0.1054 & 0.0436 & 0.1410  & 0.0264  & 0.2273\\  				 				
				& MC Sd & 0.0676 & 0.0740 & 0.1024  & 0.0626  & 0.2429 & 0.0724 & 0.0702 & 0.0994  & 0.0623  & 0.2374 & 0.0925 & 0.0443 & 0.1058  & 0.0295  & 0.2279\\  				
				& COV   & 92.8\% & 93.6\% & 84.6\%  & 92.8\%  & 89.6\% & 91.8\% & 93.2\% & 87.6\%  & 93.0\%  & 91.2\% & 90.6\% & 90.6\% & 86.8\%  & 90.0\%  & 89.6\%\\
				100  & Mean  & 0.6232 & 0.2562 & 0.4512  & 0.5088  & 1.5730 & 0.6257 & 0.2600 & 0.4534  & 0.5126  & 1.5670 & 0.6009 & 0.2452 & 0.4703  & 0.5025  & 1.5320\\				 				
				& IM SE & 0.0721 & 0.0532 & 0.0823  & 0.0432  & 0.2035 & 0.1058 & 0.0378 & 0.1257  & 0.0253  & 0.2553 & 0.0899 & 0.0366 & 0.1206  & 0.0227  & 0.2402\\
				& MC Sd & 0.0656 & 0.0485 & 0.0920  & 0.0339  & 0.1900 & 0.0633 & 0.0582 & 0.0919  & 0.0533  & 0.2130 & 0.0765 & 0.0353 & 0.0939  & 0.0228  & 0.2049\\ 				
				& COV   & 92.2\% & 93.4\% & 86.2\%  & 95.4\%  & 90.6\% & 92.2\% & 95.6\% & 88.2\%  & 94.6\%  & 92.0\% & 93.0\% & 93.0\% & 86.2\%  & 94.0\%  & 89.2\%\\ 				 				 				
				500  & Mean  & 0.6063 & 0.2515 & 0.4880  & 0.5019  & 1.5203 & 0.6051 & 0.2513 & 0.4876  & 0.5017  & 1.5237 & 0.5995 & 0.2497 & 0.4968  & 0.5013  & 1.5023 \\
				& IM SE & 0.0386 & 0.0156 & 0.0515  & 0.0106  & 0.1116 & 0.0381 & 0.0156 & 0.0509  & 0.0106  & 0.1102 & 0.0403 & 0.0158 & 0.0536  & 0.0107  & 0.1145\\ 				
				& MC Sd & 0.0366 & 0.0165 & 0.0469  & 0.0105  & 0.1060 & 0.0345 & 0.0158 & 0.0448  & 0.0102  & 0.1022 & 0.0421 & 0.0166 & 0.0516  & 0.0111  & 0.1083 \\ 				 				
				& COV   & 92.2\% & 92.0\% & 91.2\%  & 95.4\%  & 93.2\% & 96.2\% & 94.2\% & 91.4\%  & 95.4\%  & 94.4\% & 92.4\% & 93.4\% & 91.2\%  & 95.2\%  & 93.6\%\\ 				 				 				
				1000 & Mean  & 0.6057 & 0.2524 & 0.4900  & 0.5021  & 1.5204 & 0.6083 & 0.2527 & 0.4850  & 0.5028  & 1.5315 & 0.5999 & 0.2494 & 0.4990  & 0.5006  & 1.5010\\
				& IM SE & 0.0269 & 0.0110 & 0.0361  & 0.0076  & 0.0787 & 0.0262 & 0.0109 & 0.0352  & 0.0076  & 0.0770 & 0.0283 & 0.0110 & 0.0376  & 0.0076  & 0.0813\\                       			
				& MC Sd & 0.0265 & 0.0110 & 0.0345  & 0.0076  & 0.0743 & 0.0261 & 0.0110 & 0.0335  & 0.0074  & 0.0717 & 0.0284 & 0.0109 & 0.0374  & 0.0077  & 0.0846 \\ 				
				& COV   & 92.2\% & 94.8\% & 92.2\%  & 95.0\%  & 94.0\% & 92.0\% & 94.0\% & 90.4\%  & 94.8\%  & 92.8\% & 94.2\% & 96.0\% & 94.2\%  & 94.4\%  & 94.4\% \\ 				 				 				 				
				5000 & Mean  & 0.6042 & 0.2511 & 0.4946  & 0.5011  & 1.5121 & 0.6058 & 0.2516 & 0.4914  & 0.5015  & 1.5219 & 0.6016 & 0.2507 & 0.4967  & 0.5006  & 1.5058\\
				& IM SE & 0.0120 & 0.0048 & 0.0161  & 0.0034  & 0.0354 & 0.0117 & 0.0048 & 0.0158  & 0.0034  & 0.0348 & 0.0122 & 0.0048 & 0.0164  & 0.0034  & 0.0357\\
				& MC Sd & 0.0125 & 0.0048 & 0.0168  & 0.0035  & 0.0370 & 0.0120 & 0.0050 & 0.0157  & 0.0036  & 0.0347 & 0.0119 & 0.0049 & 0.0158  & 0.0033  & 0.0336\\ 				
				& COV   & 90.8\% & 95.2\% & 91.4\%  & 93.4\%  & 92.2\% & 91.6\% & 92.8\% & 89.6\%  & 91.6\%  & 89.6\% & 94.8\% & 94.2\% & 93.2\%  & 95.2\%  & 95.0\%\\ 				 				 				 				 				 				  				
				\hline  \rule{0pt}{3ex}
		\end{tabular}}
	\end{center}
\end{table}
\begin{table}[H]
	\begin{center}
		\caption{Study 1: Mean fit of the FM-BS model based on different initialization algorithms and samples sizes when $\btheta=(0.8,0.25,0.25,1.0,5.0)$.}\label{sim1escenario2}
		\vskip 3mm
	\tiny{\begin{tabular}{c@{\hskip 0.12in}c@{\hskip 0.12in}|c@{\hskip 0.12in}c@{\hskip 0.12in}c@{\hskip 0.12in}c@{\hskip 0.12in}c@{\hskip 0.12in}|c@{\hskip 0.12in}c@{\hskip 0.12in}c@{\hskip 0.12in}c@{\hskip 0.12in}c@{\hskip 0.12in}|c@{\hskip 0.12in}c@{\hskip 0.12in}c@{\hskip 0.12in}c@{\hskip 0.12in}c@{\hskip 0.12in}c@{\hskip 0.12in}c}
		\hline  \rule{0pt}{3ex}
		n  &Measure& $\widehat{p}_1$ & $\widehat{\alpha}_1$  &  $\widehat{\alpha}_2$ &  $\widehat{\beta}_1$ &  $\widehat{\beta}_2$ & $\widehat{p}_1$ & $\widehat{\alpha}_1$  &  $\widehat{\alpha}_2$ &  $\widehat{\beta}_1$ &  $\widehat{\beta}_2$ & $\widehat{p}_1$ & $\widehat{\alpha}_1$  &  $\widehat{\alpha}_2$ &  $\widehat{\beta}_1$ &  $\widehat{\beta}_2$ \\[0.1cm] \hline \rule{0pt}{3ex}
		&       &            \multicolumn{5}{c}{$k$-means}   &            \multicolumn{5}{c}{$k$-medoids}   &            \multicolumn{5}{c}{$k$-bumps}\\[0.1cm] \hline \rule{0pt}{3ex}
				75   & Mean  & 0.7978 & 0.2475 & 0.2404  & 1.0005  & 5.0332 & 0.7990 & 0.2468 & 0.2338  & 1.0007  & 5.0189 & 0.7979 & 0.2461 & 0.2384  & 1.0003  & 5.0069 \\
				& IM SE & 0.0460 & 0.0247 & 0.0594  & 0.0325  & 0.3454 & 0.0458 & 0.0246 & 0.0581  & 0.0326  & 0.3425 & 0.0466 & 0.0232 & 0.0528  & 0.0328  & 0.3224 \\  				 				
				& MC Sd & 0.0451 & 0.0226 & 0.0490  & 0.0314  & 0.3717 & 0.0469 & 0.0231 & 0.0477  & 0.0323  & 0.3438 & 0.0324 & 0.0332 & 0.0603  & 0.0292  & 0.2273 \\  				
				& COV   & 94.0\% & 93.4\% & 92.8\%  & 96.4\%  & 92.2\% & 93.0\% & 92.4\% & 89.4\%  & 95.0\%  & 92.6\% & 92.4\% & 93.4\% & 91.0\%  & 94.6\%  & 93.4\% \\
				100  & Mean  & 0.7982 & 0.2460 & 0.2401  & 0.9997  & 5.0141 & 0.7994 & 0.2482 & 0.2399  & 0.9994  & 5.0063 & 0.7981 & 0.2473 & 0.2416  & 1.0013  & 5.0098 \\				 				
				& IM SE & 0.0398 & 0.0208 & 0.0480  & 0.0279  & 0.2904 & 0.0398 & 0.0212 & 0.0480  & 0.0281  & 0.2929 & 0.0899 & 0.0366 & 0.1206  & 0.0227  & 0.2402 \\
				& MC Sd & 0.0414 & 0.0203 & 0.0431  & 0.0280  & 0.2985 & 0.0404 & 0.0197 & 0.0407  & 0.0280  & 0.2929 & 0.0791 & 0.0305 & 0.0899  & 0.0226  & 0.2053 \\ 				
				& COV   & 94.4\% & 93.4\% & 91.2\%  & 94.4\%  & 91.4\% & 93.4\% & 96.2\% & 93.2\%  & 95.2\%  & 93.2\% & 94.2\% & 93.2\% & 91.8\%  & 94.4\%  & 93.2\% \\ 				 				 				
				500  & Mean  & 0.8000 & 0.2492 & 0.2481  & 0.9991  & 4.9953 & 0.7990 & 0.2493 & 0.2482  & 1.0005  & 5.0029 & 0.7991 & 0.2494 & 0.2491  & 1.0012  & 5.0024 \\
				& IM SE & 0.0179 & 0.0090 & 0.0189  & 0.0124  & 0.1260 & 0.0179 & 0.0090 & 0.0189  & 0.0125  & 0.1258 & 0.0179 & 0.0090 & 0.0190  & 0.0125  & 0.1265 \\ 				
				& MC Sd & 0.0175 & 0.0089 & 0.0186  & 0.0121  & 0.1228 & 0.0171 & 0.0089 & 0.0187  & 0.0127  & 0.1244 & 0.0185 & 0.0092 & 0.0176  & 0.0125  & 0.1259 \\ 				 				
				& COV   & 94.8\% & 95.2\% & 94.0\%  & 96.4\%  & 96.2\% & 96.6\% & 94.8\% & 94.0\%  & 93.0\%  & 95.8\% & 94.0\% & 94.2\% & 96.2\%  & 95.2\%  & 96.0\% \\ 				 				 				
				1000 & Mean  & 0.7986 & 0.2497 & 0.2485  & 0.9994  & 4.9992 & 0.7995 & 0.2498 & 0.2496  & 1.0001  & 5.0054 & 0.7980 & 0.2496 & 0.2494  & 1.0002  & 4.9980 \\
				& IM SE & 0.0127 & 0.0064 & 0.0131  & 0.0088  & 0.0882 & 0.0127 & 0.0064 & 0.0132  & 0.0088  & 0.0891 & 0.0127 & 0.0064 & 0.0131  & 0.0088  & 0.0885 \\                       			
				& MC Sd & 0.0113 & 0.0061 & 0.0131  & 0.0087  & 0.0922 & 0.0130 & 0.0061 & 0.0128  & 0.0092  & 0.0938 & 0.0128 & 0.0063 & 0.0127  & 0.0089  & 0.0846 \\ 				
				& COV   & 97.4\% & 95.4\% & 93.0\%  & 95.2\%  & 91.6\% & 94.8\% & 95.0\% & 95.2\%  & 93.6\%  & 93.2\% & 94.8\% & 94.6\% & 93.4\%  & 94.8\%  & 96.0\%  \\ 				 				 				 				
				5000 & Mean  & 0.8001 & 0.2499 & 0.2501  & 1.0000  & 4.9983 & 0.8004 & 0.2500 & 0.2494  & 1.0001  & 4.9983 & 0.8003 & 0.2500 & 0.2499  & 1.0003  & 5.0018 \\
				& IM SE & 0.0057 & 0.0028 & 0.0058  & 0.0039  & 0.0396 & 0.0057 & 0.0028 & 0.0058  & 0.0039  & 0.0396 & 0.0057 & 0.0028 & 0.0058  & 0.0039  & 0.0396 \\
				& MC Sd & 0.0057 & 0.0028 & 0.0057  & 0.0037  & 0.0393 & 0.0058 & 0.0028 & 0.0056  & 0.0042  & 0.0374 & 0.0056 & 0.0029 & 0.0057  & 0.0038  & 0.0386 \\ 				
				& COV   & 95.8\% & 94.4\% & 96.4\%  & 95.6\%  & 95.2\% & 94.4\% & 97.0\% & 94.2\%  & 92.6\%  & 95.2\% & 95.6\% & 93.4\% & 95.6\%  & 95.6\%  & 94.8\% \\ 			 				 				 				 				 				  				
				\hline \rule{0pt}{3ex}
		\end{tabular}}
	\end{center}
\end{table}
Note that when the samples are poorly separated, the estimates using the $k$-bumps algorithm obtain better estimates for small sample sizes compared with the $k$-means and $k$-medoids algorithms. When the samples are well separated the estimates are good regardless of sample size using all algorithms. In general, the results suggest that the proposed FM-BS model produces satisfactory estimates, as expected.
\begin{figure}[t!]
	\centering
	\includegraphics[scale=0.35]{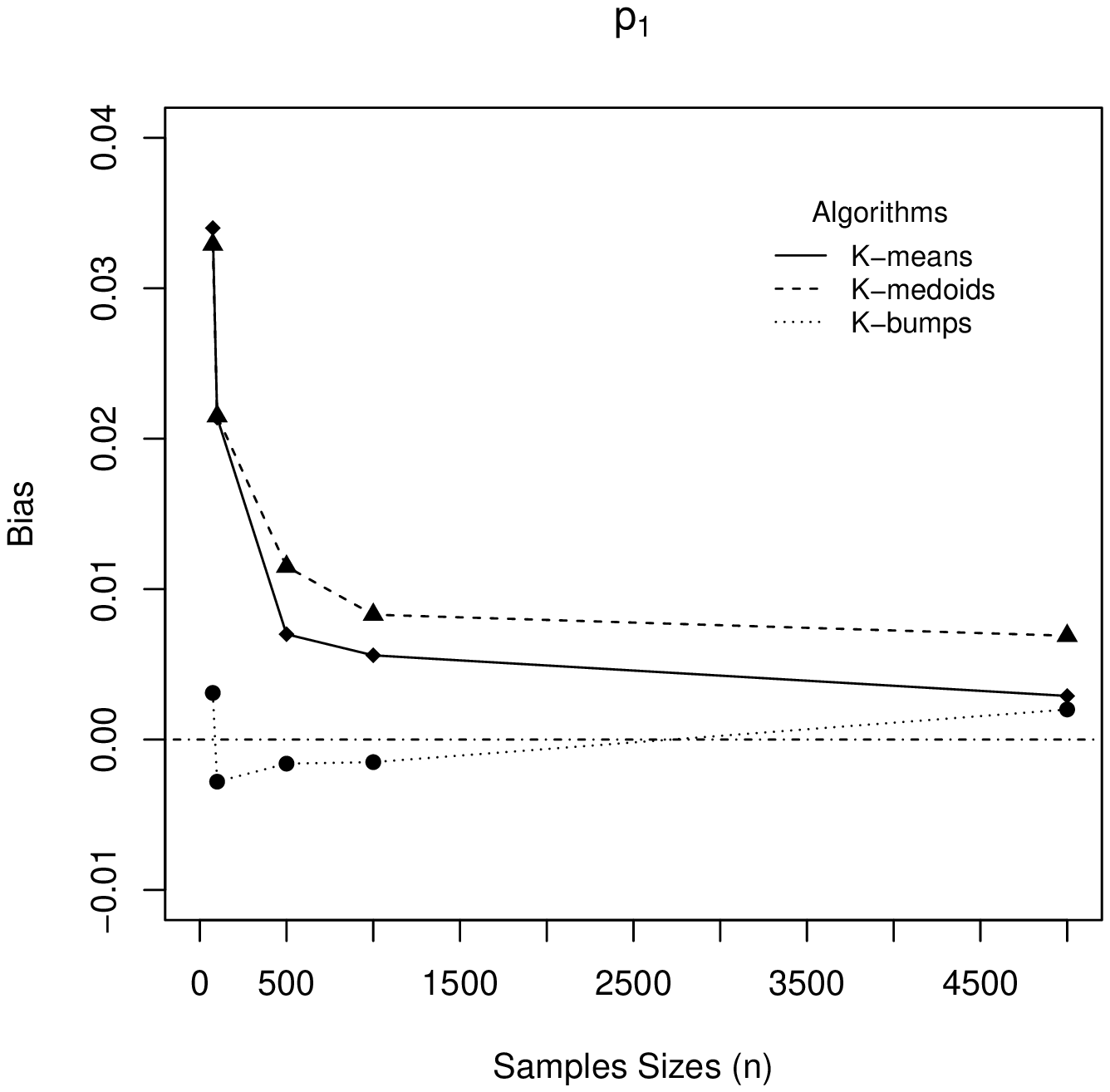}~\includegraphics[scale=0.35]{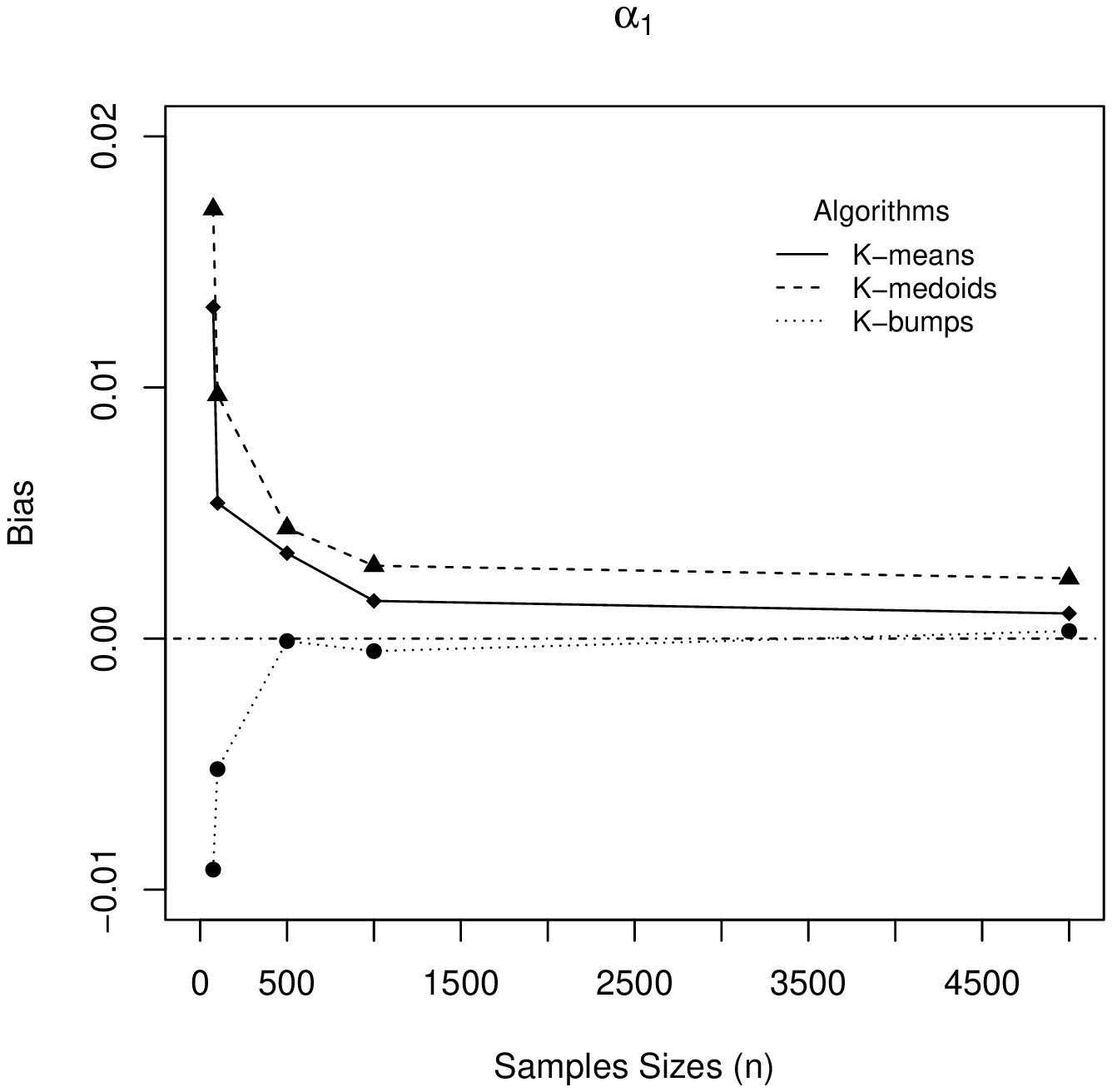}~\includegraphics[scale=0.35]{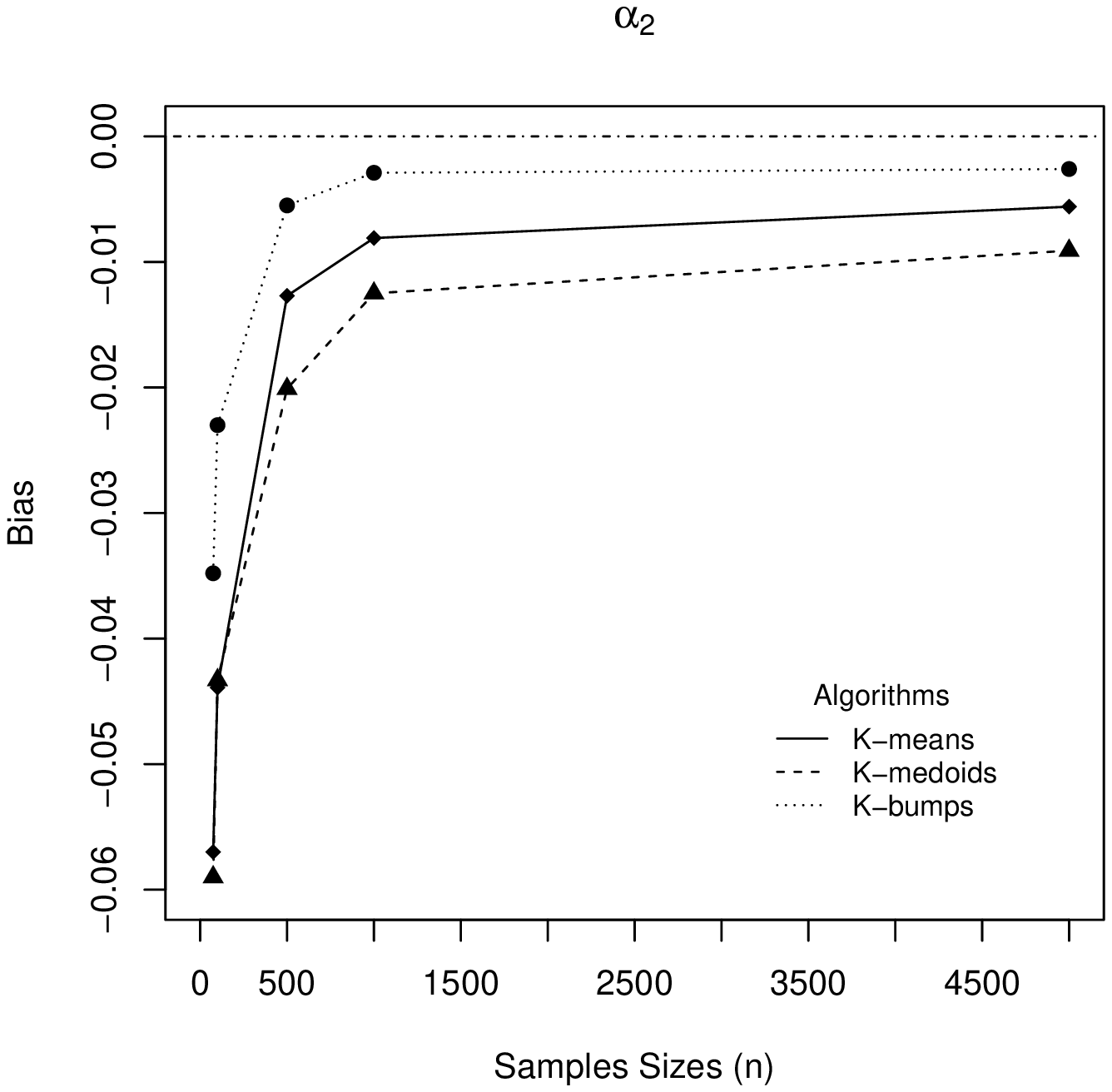}\\
	\includegraphics[scale=0.35]{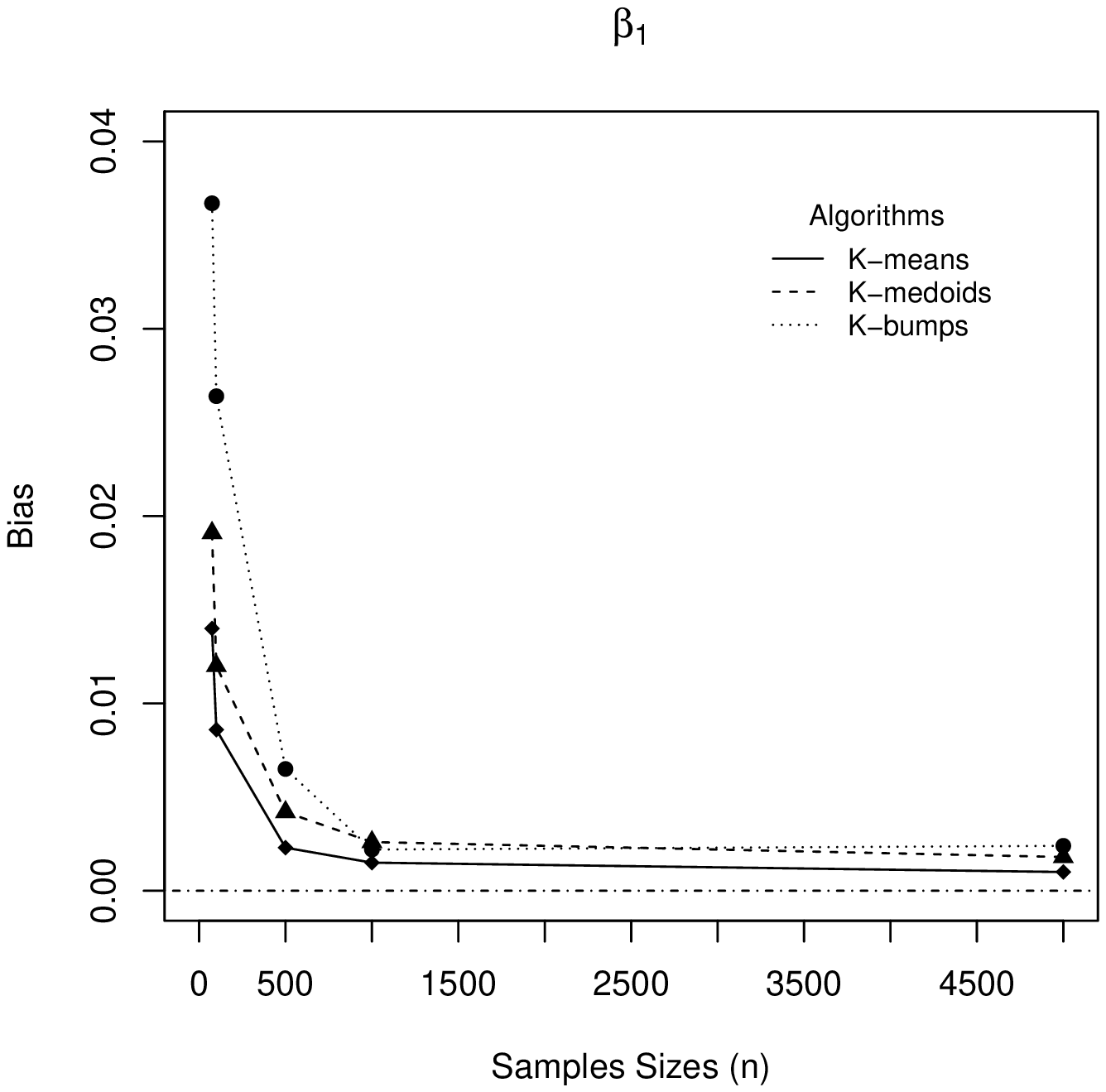}~\includegraphics[scale=0.35]{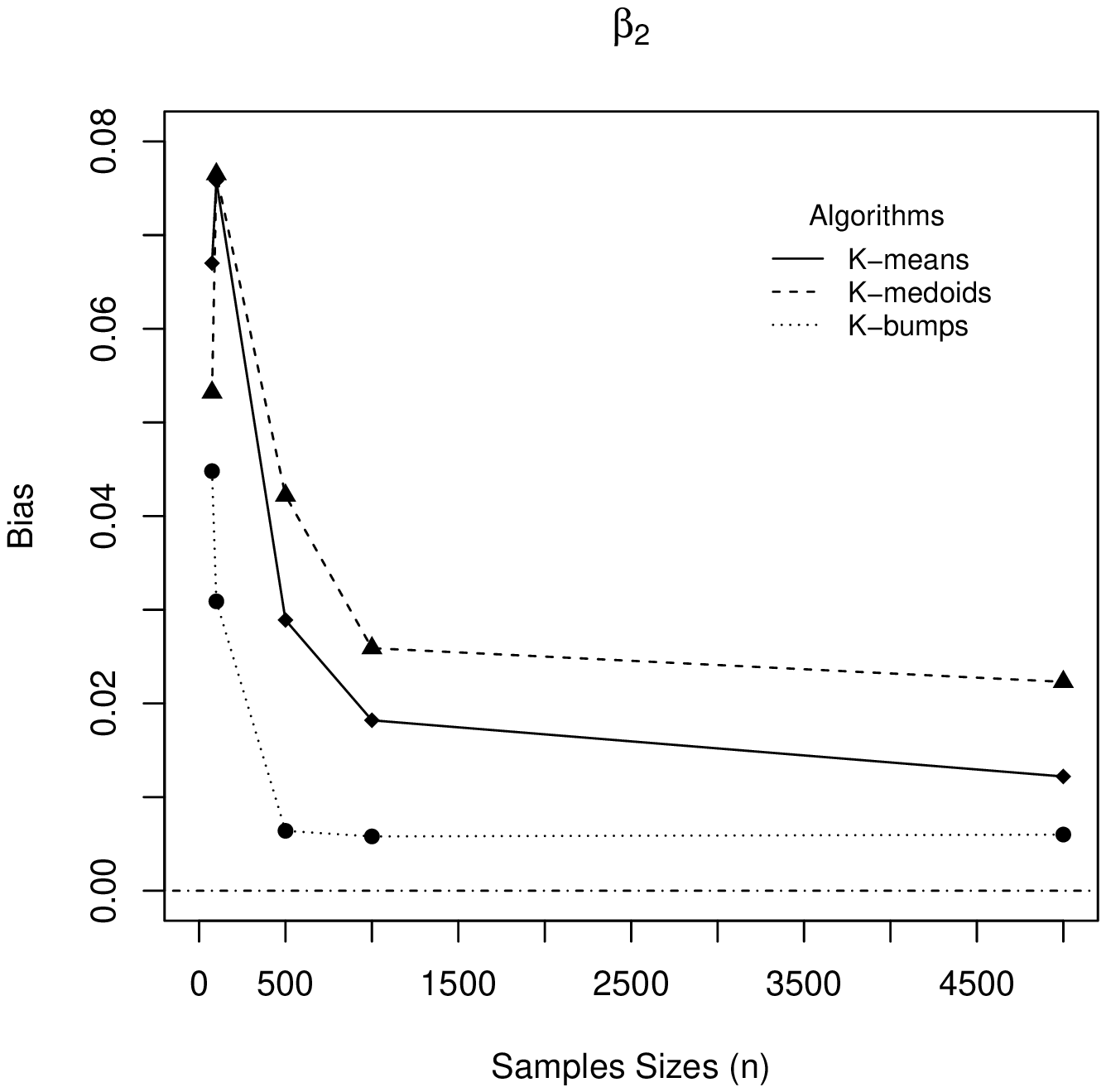}
	\caption{Study 2. (Scenario 1) Average bias of parameter estimates in the FM-BS model with different algorithms.\label{biasRMSE1}}
\end{figure}

\begin{figure}[!t]
	\centering
	\includegraphics[scale=0.35]{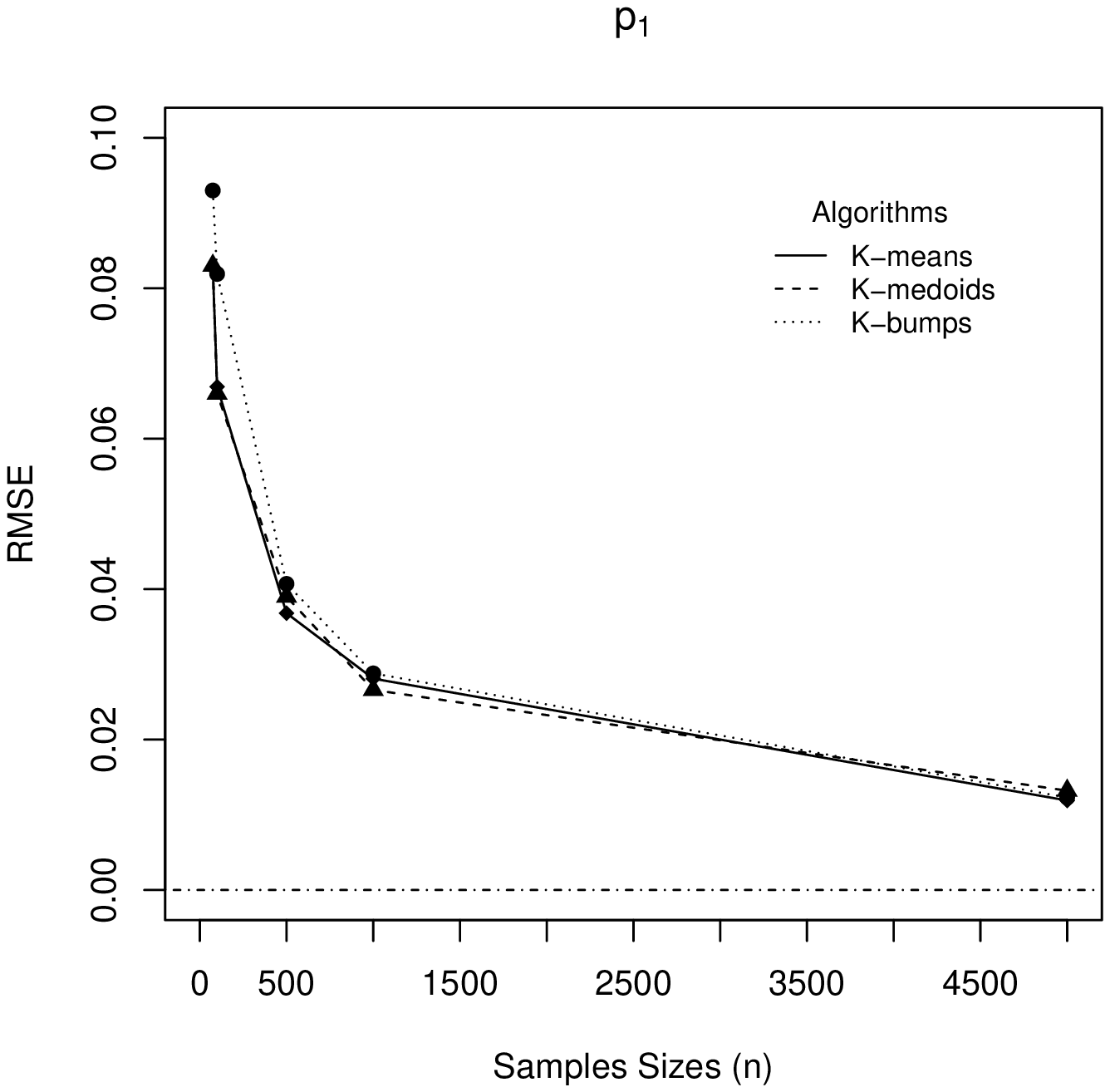}~\includegraphics[scale=0.35]{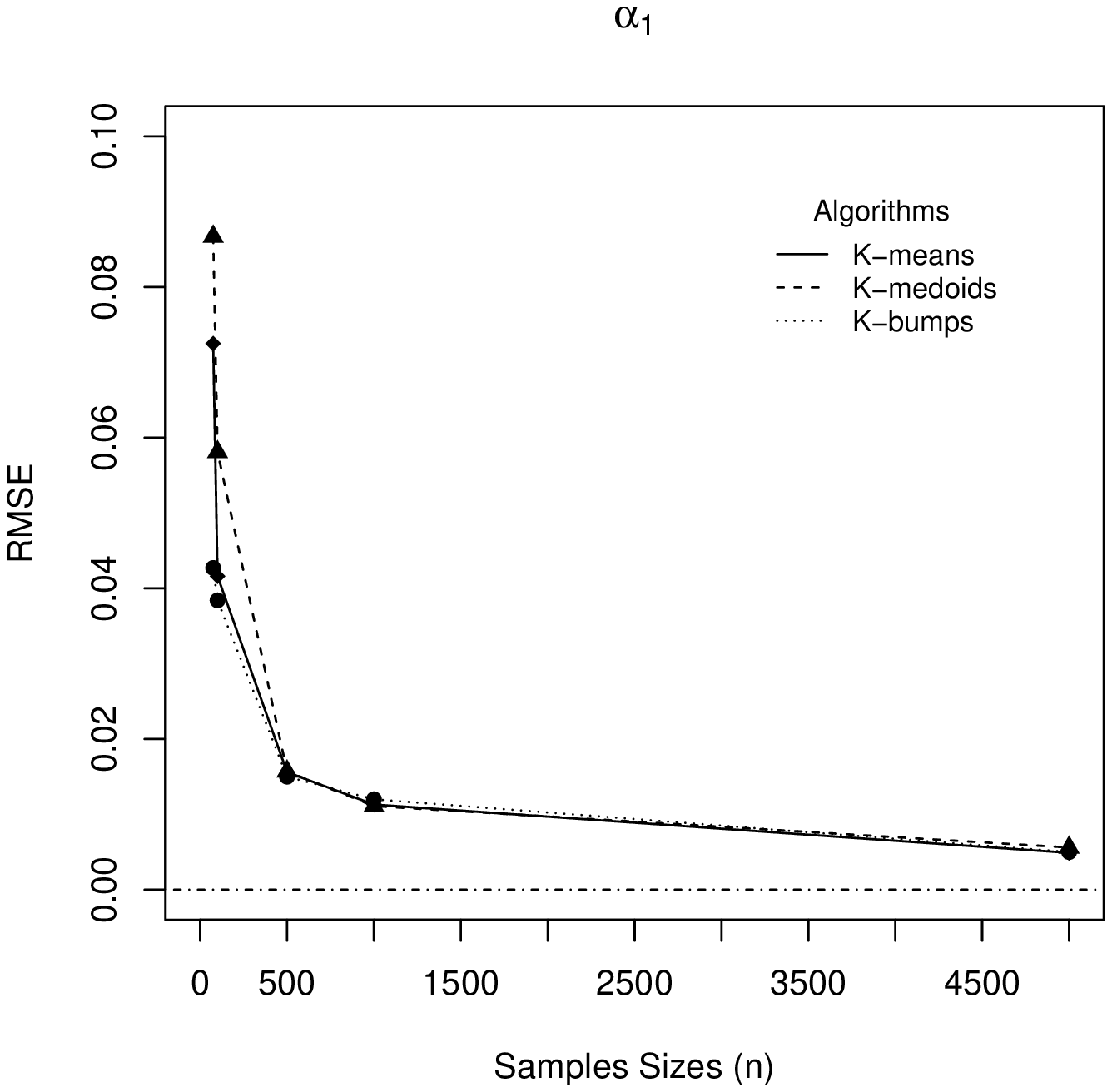}~\includegraphics[scale=0.35]{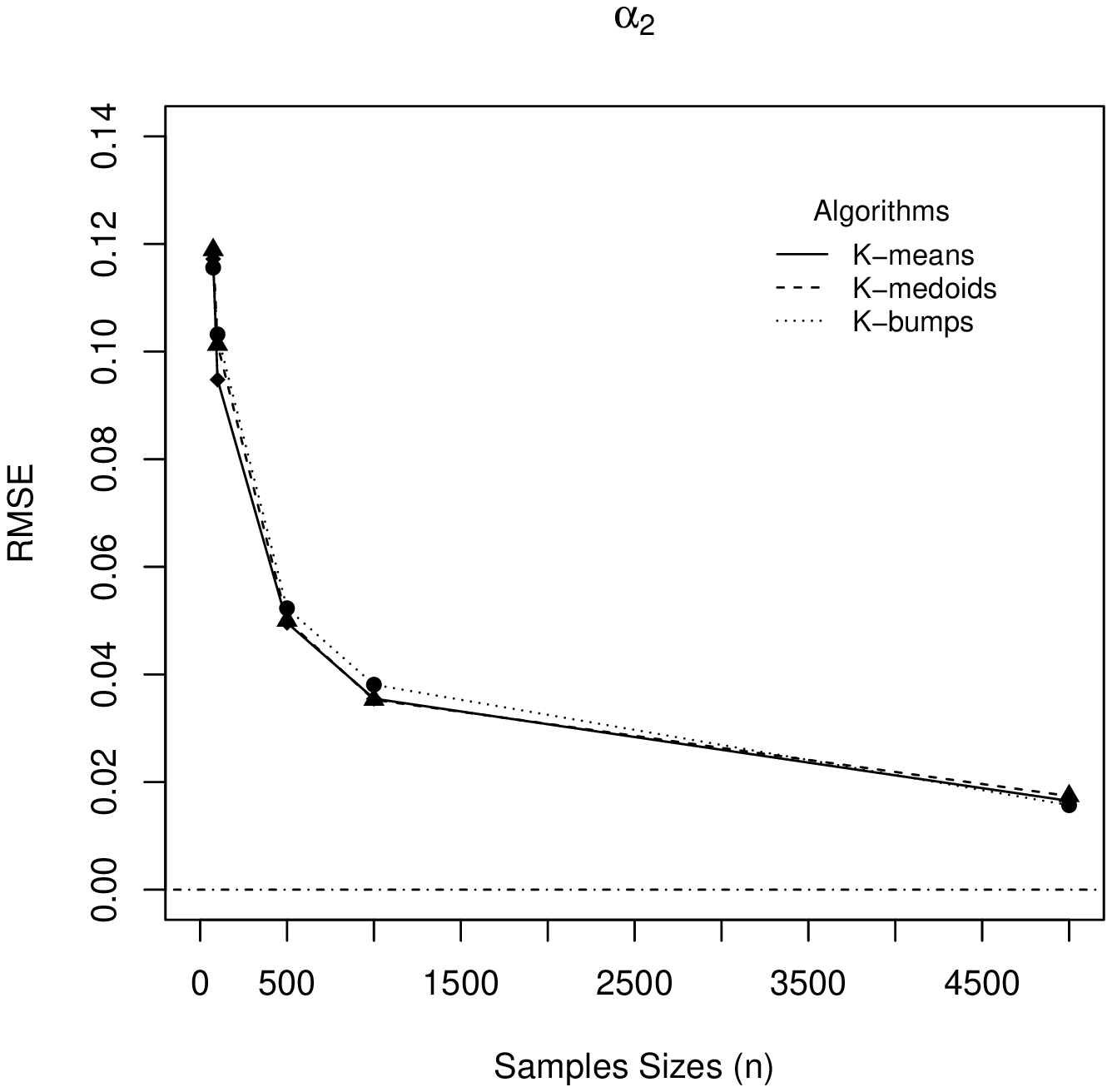}\\
	\includegraphics[scale=0.35]{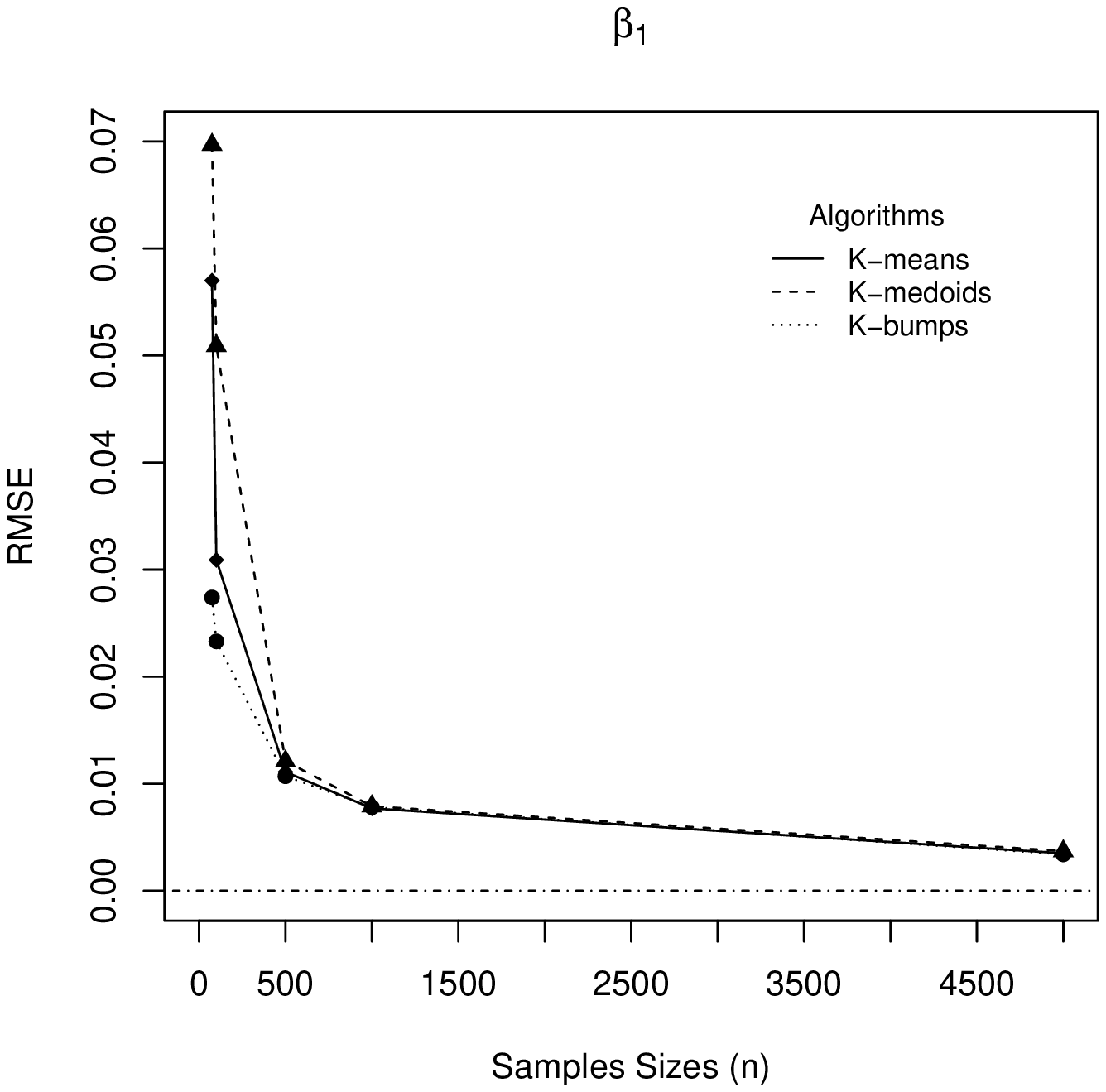}~\includegraphics[scale=0.35]{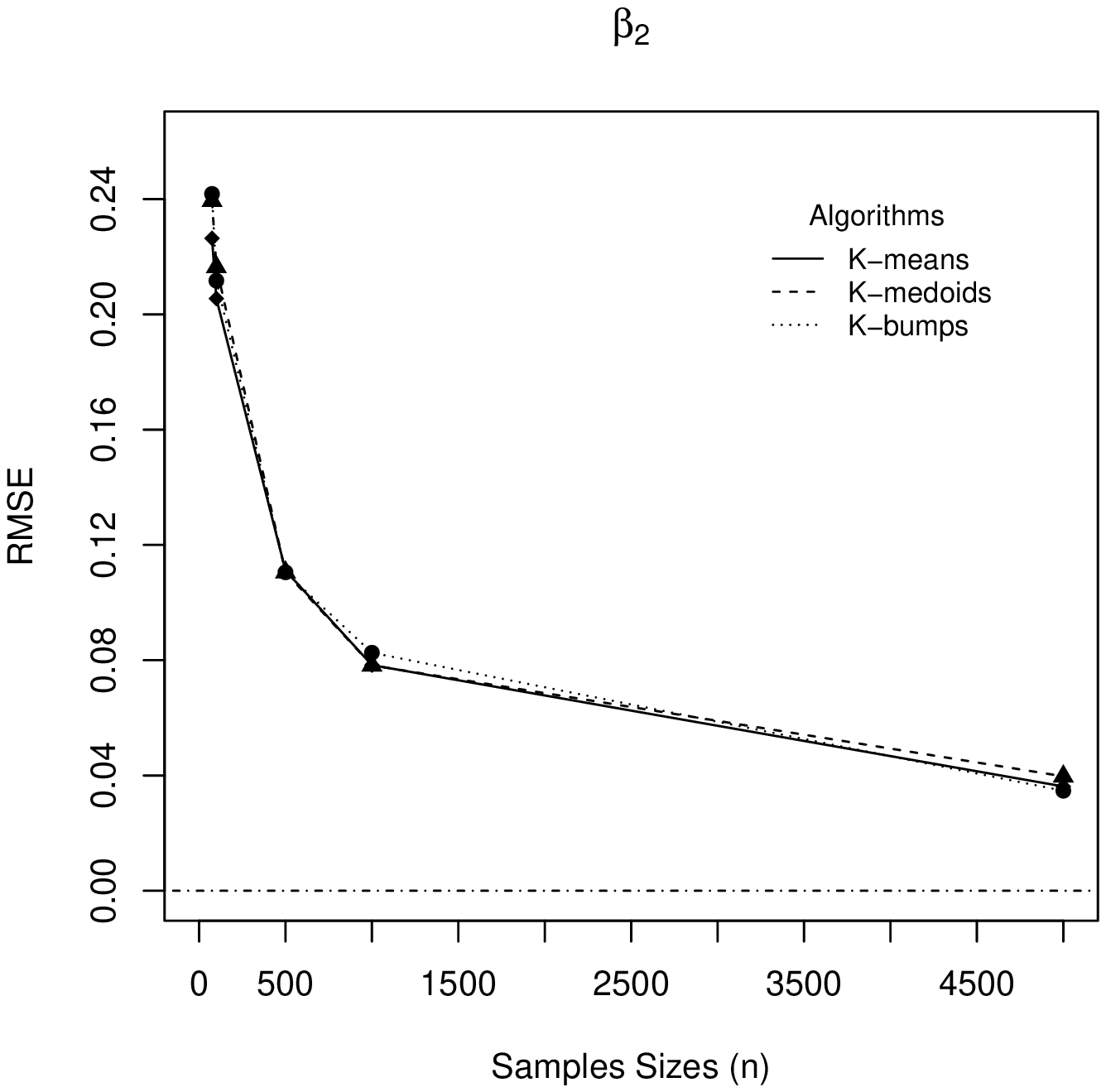}
	\caption{Study 2. (Scenario 1) Average RMSE of parameter estimates in the FM-BS model with different algorithms.\label{biasRMSE2}}
\end{figure}
\subsection{Study 2: Asymptotic properties of the EM estimates }
In this section the goal is to show the asymptotic properties of the EM estimates. Our strategy is to generate artificial samples for the model  in (\ref{defG=2}). Various settings of sample sizes were chosen ($n=75, 100,500,1000$ and $5000$). The true values of the parameters in this study are  as   in Study 1 (Scenario 1). For each combination of parameters and sample sizes, we generated 1000 random samples from the FM-BS model. To evaluate the estimates obtained by the proposed EM algorithm, we compared the bias (Bias) and the root mean square error (RMSE) for each parameter over the 1000 replicates. They are defined as
$$\text{Bias}(\theta_i)=\frac{1}{1000}\sum_{j=1}^{1000}(\widehat{\theta}^{(j)}_i - \theta_i) \quad \text{and}\quad \text{RMSE}(\theta_i)=\sqrt{\frac{1}{1000}\sum_{j=1}^{1000}(\widehat{\theta}^{(j)}_i - \theta_i)^2},$$
where $\widehat{\theta}_i^{(j)}$ is the estimates of $\theta_i$ from the jth sample. The results for $p_1$, $\alpha_1$, $\alpha_2$, $\beta_1$ and $\beta_2$ are shown in Figures \ref{biasRMSE1} and \ref{biasRMSE2}. As a general rule, we can say that Bias and RMSE tend to approach zero when the sample size increases, which indicates that the estimates based on the proposed EM algorithm under the FM-BS model provide good asymptotic properties.

\section{Real applications}
In order to illustrate the proposed method, we consider two real datasets, and  adopt the Bayesian information criterion \citep[BIC,][]{Schwarz78} and Akaike information criterion (AIC) to select the number of components in mixture models ($\text{AIC} = -2\ell(\widehat{\btheta})+ 2\rho$ and $\text{BIC} = -2\ell(\widehat{\btheta})+\rho \log n$,  where  $\rho$ is the number of the parameters  in the model).


\subsection{Real dataset I}
We applied the proposed method to data corresponding to the enzymatic activity in the blood and representing the metabolism of carcinogenic substances among 245 unrelated individuals that  were  studied previously by \cite{bechtel:93},  who fitted a mixture of two skewed distributions and \cite{Balakrishnan:11}, who considered three different mixture models based on the BS and length-biased models. Recently, a bimodal BS (BBS)  model  has been considered by  \cite{olmos2016} to fit these data.  Here we perform the EM algorithm described in Section \ref{sectionem} to carry out the ML estimation for the FM-BS model. The competing models are compared using the AIC, BIC and the associated rate of convergence, $r$, which is assessed in practice as
\begin{equation}
r = lim_{t\to\infty}  \frac{\| \btheta^{(t+1)} - \btheta^{(t)}\|}{\|\btheta^{(t)} - \btheta^{(t-1)}\|}.\label{rates}
\end{equation}
The rate of convergence depends on the fraction of missing information, and a greater value of $r$ implies slower convergence; see \cite{Meng94}. Models with lower convergence rates and BIC are considered more preferable. Table \ref{tableENZYMEindicators}  shows the ML estimates  obtained by fitting the FM-BS model ($G = 1-4$ components) and BBS model. Note that the estimation procedure for fitting the FM-BS model does not converge properly for $G \geq 3$.  Figure 9 (a) and (b)  display histograms with estimated pdfs for the data superimposed with $G=1-4$ components. In Figure 9 (c) and (d) we show the cumulative and estimated survival functions and the empirical survival function of enzyme data for four fitted FM-BS models respectively.
\begin{figure}[!t]
	\centering
	{\includegraphics[scale=0.37]{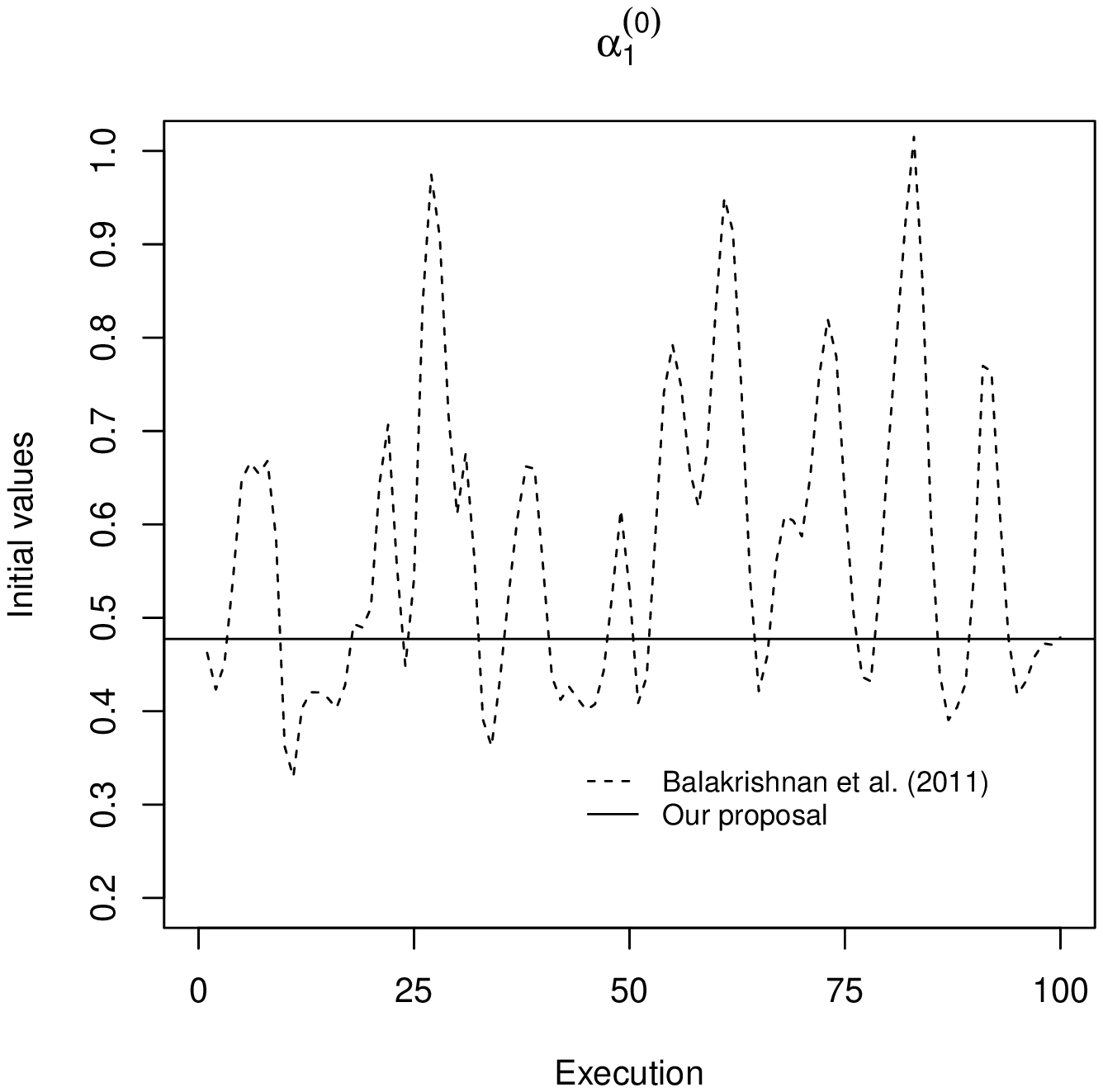}}~{\includegraphics[scale=0.37]{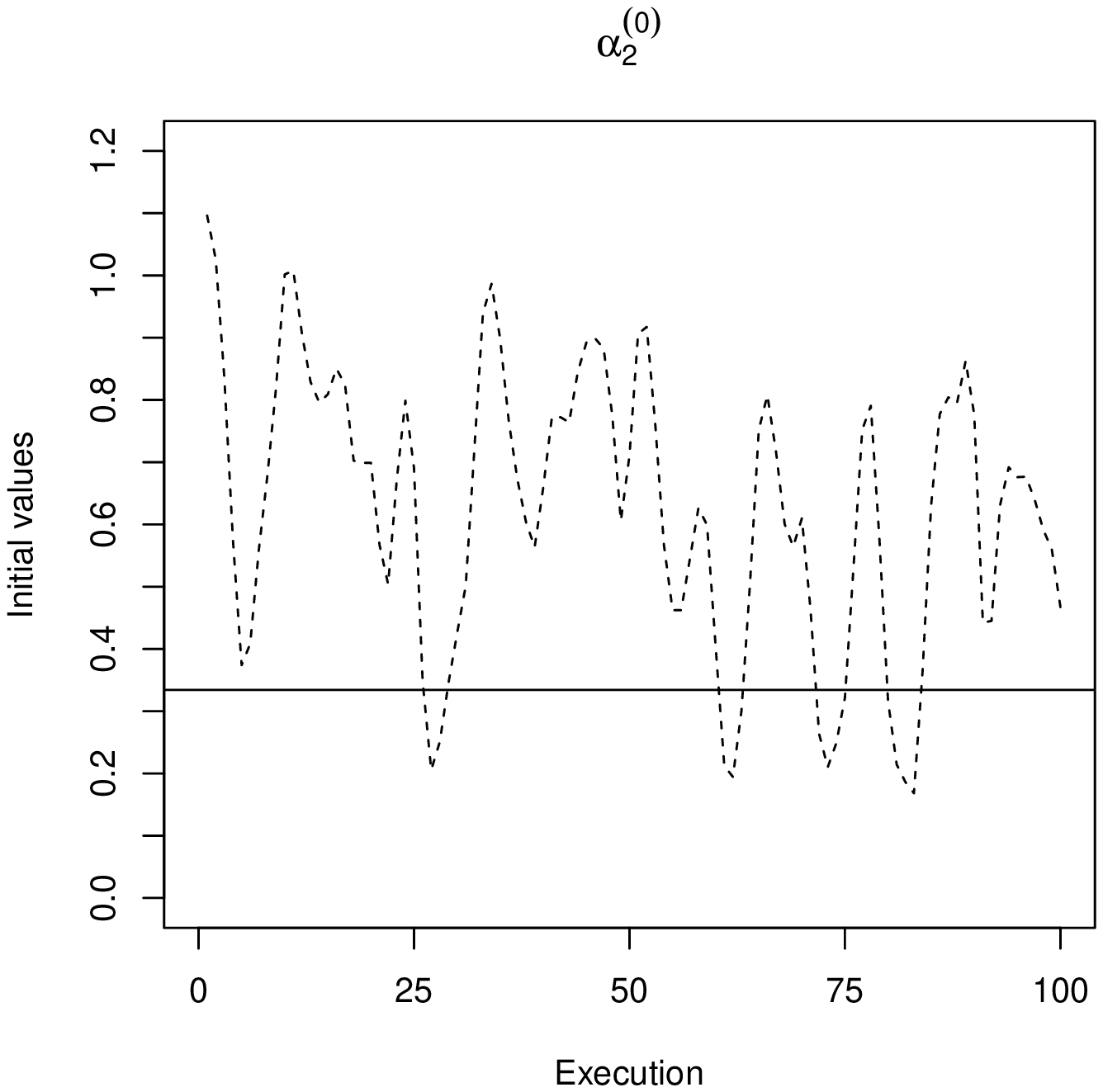}}~{\includegraphics[scale=0.37]{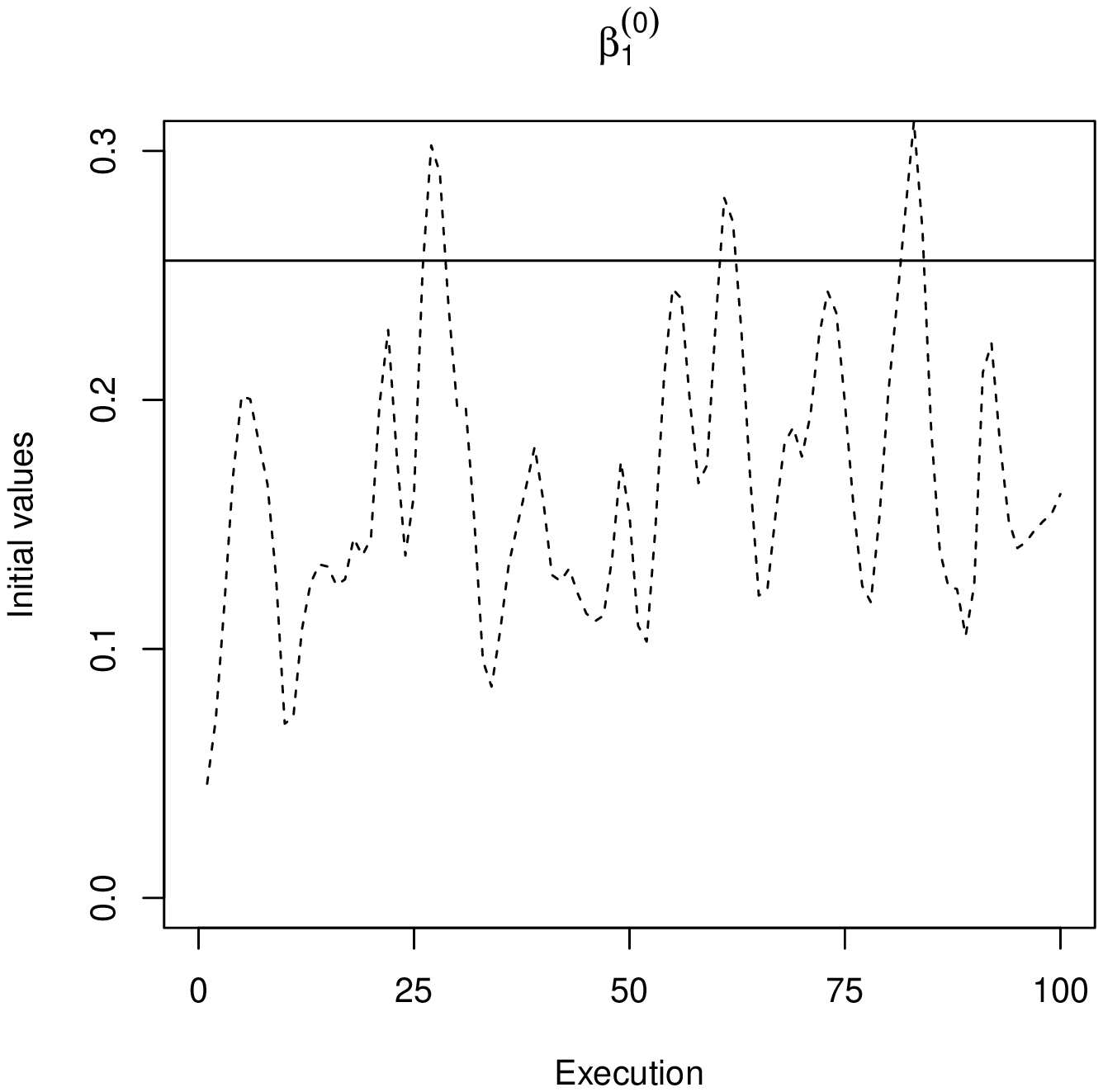}} \\
	{\includegraphics[scale=0.37]{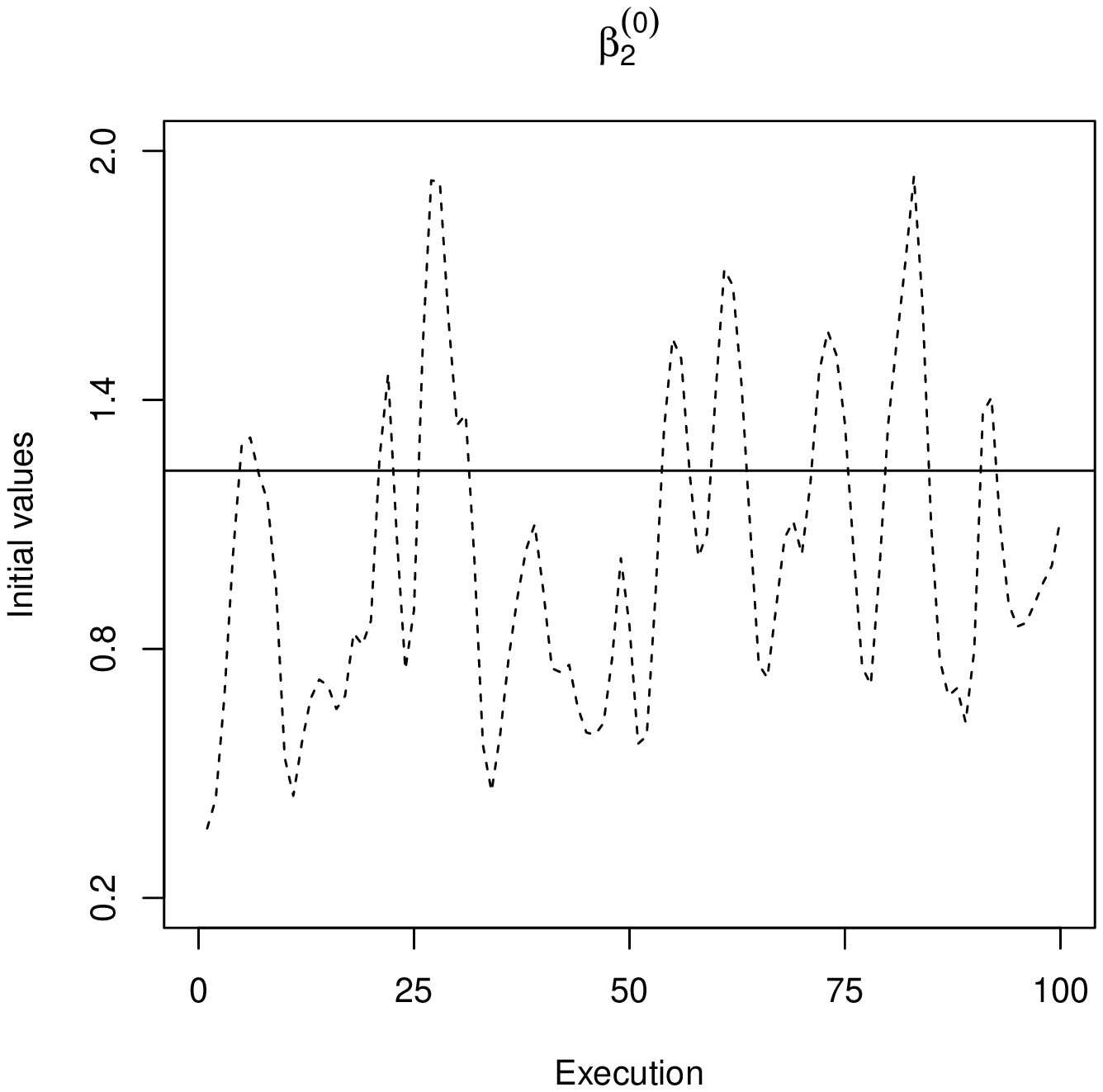}}~{\includegraphics[scale=0.37]{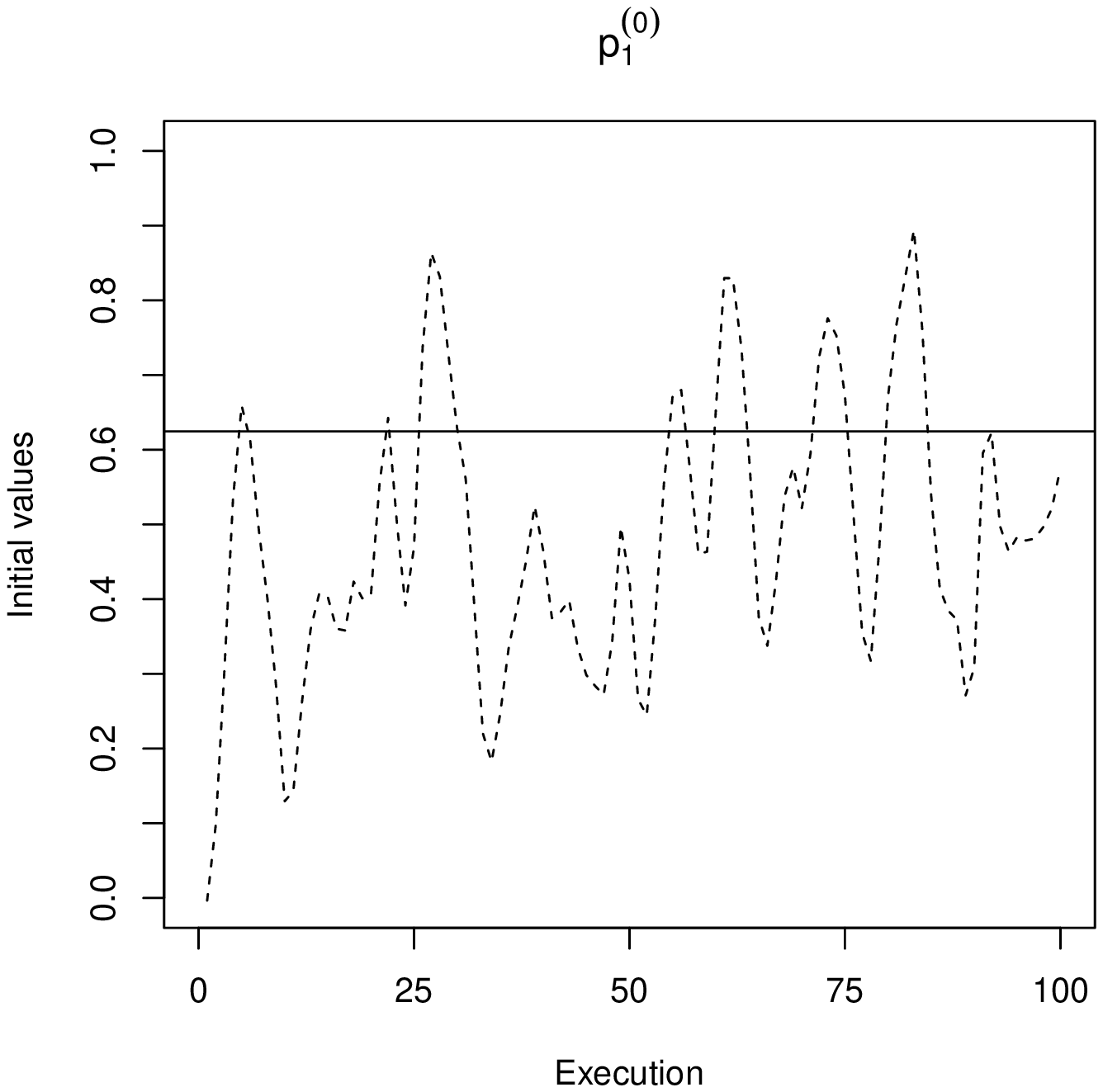}}~{\includegraphics[scale=0.37]{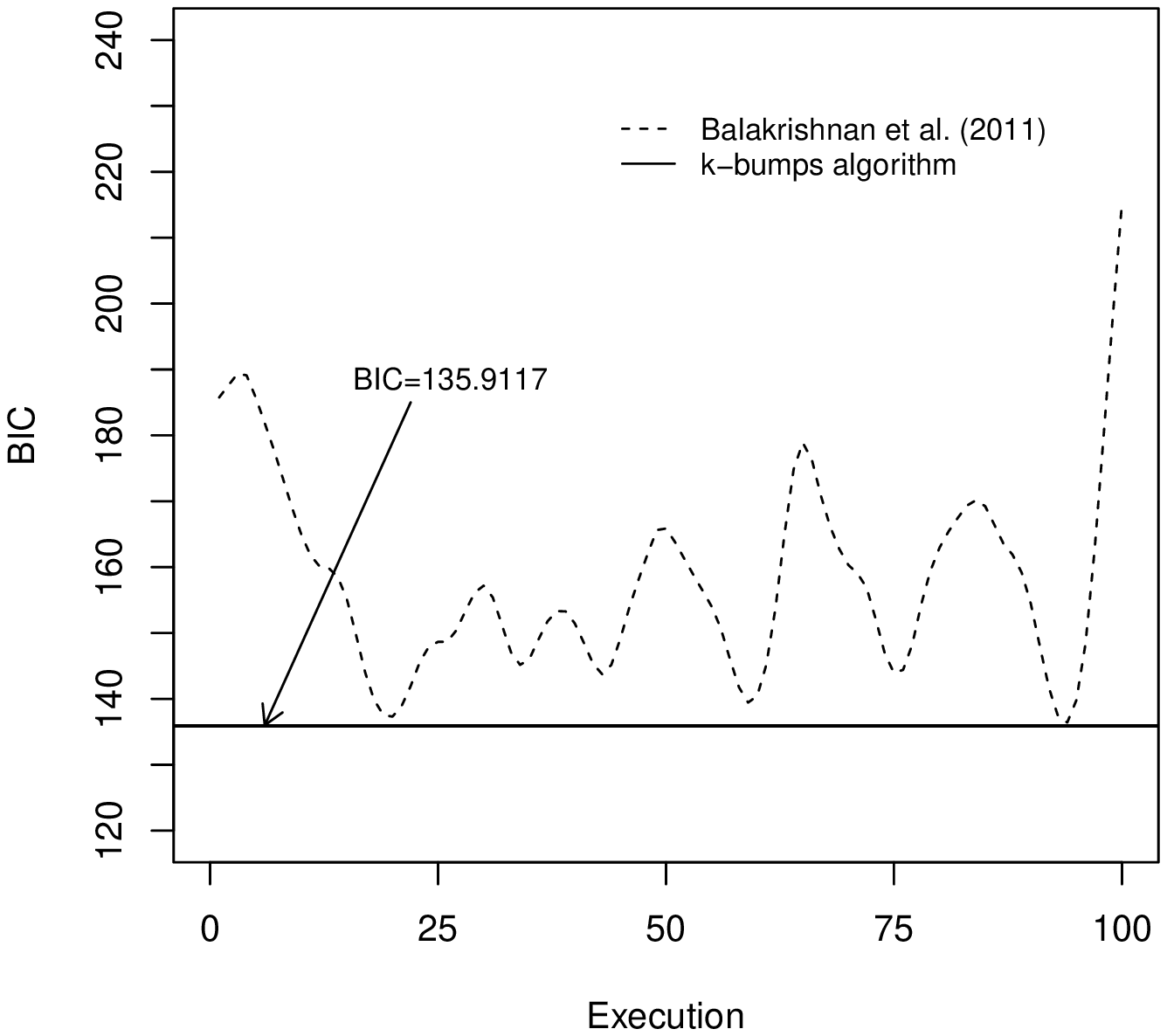}} \\
	\caption{Initialization strategy for the EM algorithm used in Balakrishnan et al. (2011) vs. $k$-bumps strategy for 100 runs. \label{fig-apli_initial_valuesENZYME3}}
\end{figure}
The graphical visualization shows that  the FM-BS model ($G=2$, $G=3$ or $G=4$) adapts to the shape of the histogram very accurately.

 \begin{table}[!t]
 	\begin{center}
 		\caption{Comparison of log-likelihood maximum and BIC for fitted FM-BS model using the enzyme data. The number of parameters and the rate of convergence are denoted by $m$ and $r$, respectively.}\label{tableENZYMEindicators}
 		\vskip 3mm
 		\small{	\begin{tabular}{ccccccccc}
 				\hline
 	  	              & $G$       & m      &log-lik   & AIC              &BIC       & Iterations & r	\\
 				\hline
 				BBS   & 1         & 3      & -86.2856 & 178.5713         & 189.075  &  40       &  0.2697 \\ 		
              \hline
 				FM-BS & 1         & 2      & -105.5071& 215.0141         & 222.0167  &   2        & - \\
 				FM-BS & 2         & 5      & -54.2027 &\textbf{118.4054} & \textbf{135.9117}  &   7        & 0.2236\\ 
 				FM-BS & 3         & 8      & -51.6763 & 119.2884         & 147.3627  & 759        & 0.9958\\
 				FM-BS & 4         & 11     & -39.7009 & 122.7353         & 139.9157  & 275        & 0.9888\\
 				\hline
 		\end{tabular}}
 	\end{center}
 \end{table}
 As reported by \cite{turner2000}, we can use parametric or semiparametric bootstrapping to test the hypothesis concerning the number of components in the mixture. Following the method proposed by \cite{turner2000}, we considered 1000 bootstrap statistics to test $G=1$ versus $G=2$. The $p$-value is 0.031 for the parametric bootstrap.

\begin{center}
	\def\stackalignment{l}
	\topinset{\bfseries(a)}{\includegraphics[width=3.3in]{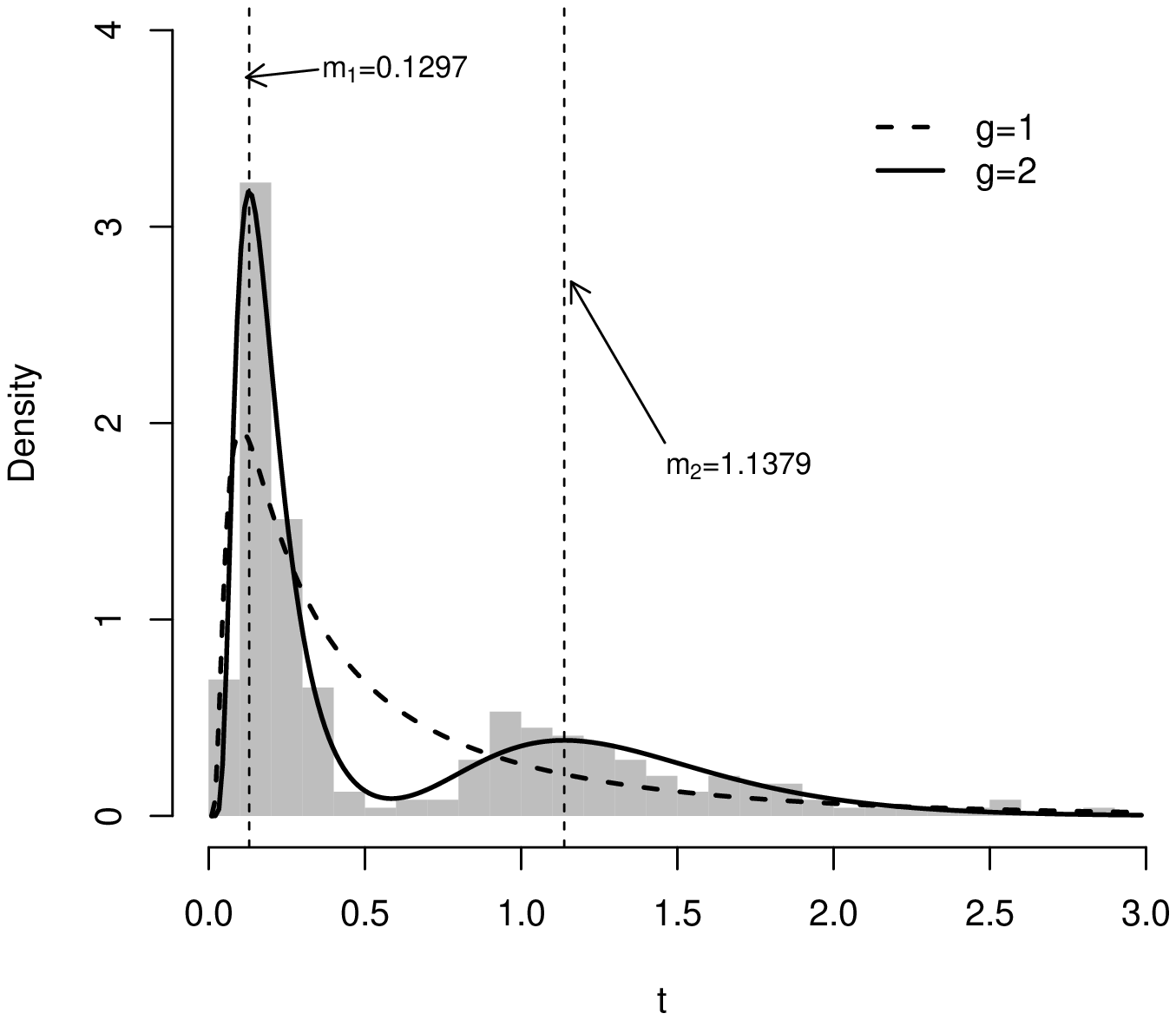}}{0.4in}{0.02in}~\topinset{\bfseries(b)}{\includegraphics[width=3.3in]{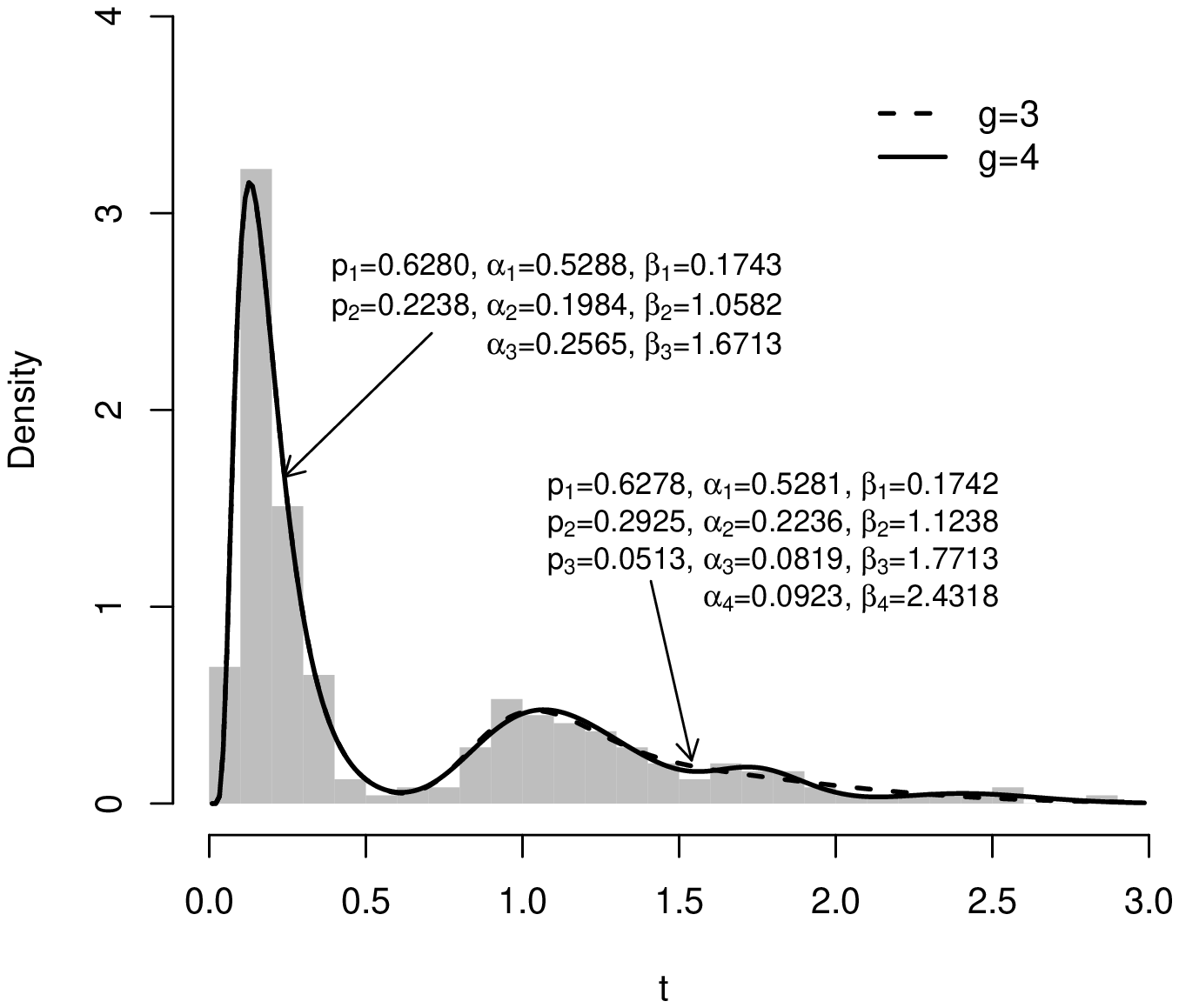}}{0.4in}{0.02in} \vspace*{-1.4cm} \\
	\topinset{\bfseries(c)}{\includegraphics[width=3.3in]{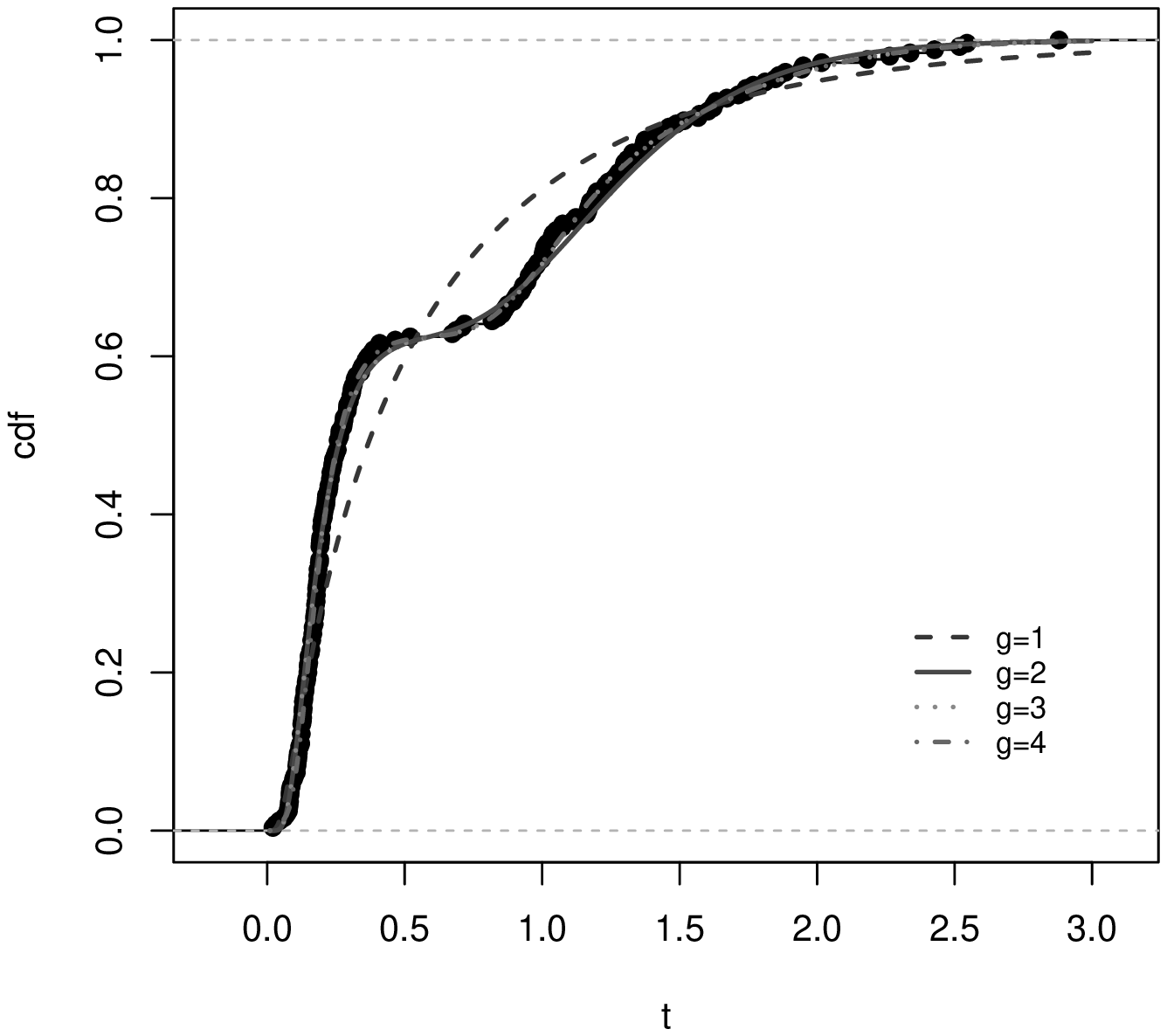}}{0.4in}{0.02in}~\topinset{\bfseries(d)}{\includegraphics[width=3.3in]{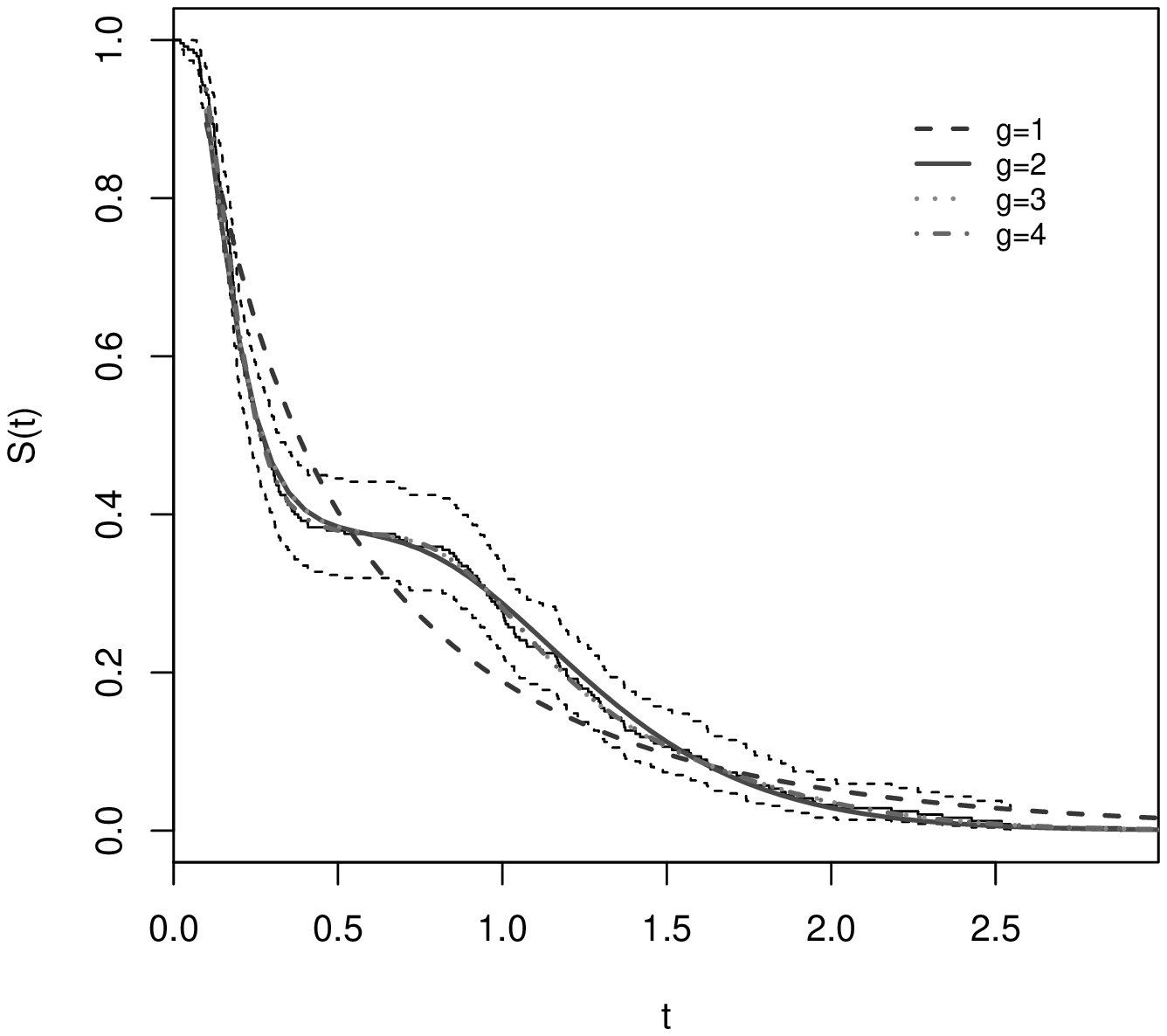}}{0.4in}{0.02in}
	\begin{flushleft}
		Figure 9: (a) and (b) Histogram of the enzyme data overlaid with ML-fitted densities (c) Cu- mulative and (d) estimated survival functions and the empirical survival function for four fitted FM-BS models. \vspace*{0.2cm}
	\end{flushleft}
\end{center}
Accordingly, there is strong evidence  that there are at least two components. For the $G=2$ versus $G=3$ test, the bootstrap $p$-value is 0.415, thus there is no evidence that more than two components are required. Moreover, the results based on AIC, BIC and $r$ (see equation \ref{rates}) to test hypotheses and figures indicate that the FM-BS model with $G=2$ provides much better fit of the data than the other models considered. Note also in Table \ref{tableENZYMEindicators} that the BIC obtained using the proposed method is lower than that reported in Table 3 of \cite{Balakrishnan:11} for the mixture distribution of two different BS (146.02). This is due to the method used to obtain the initial values for the EM algorithm used in  \cite{Balakrishnan:11}, which change every time the algorithm is executed, resulting in different estimates and BIC values, as can be seen in Figure \ref{fig-apli_initial_valuesENZYME3}.
\begin{table}[H]
\begin{center}
\caption{Estimated parameter values via the EM algorithm with corresponding standard errors (SE) and two-sided 95\% confidence interval for the FM-BS model applied to the enzyme data.}\vskip 3mm
		\label{tableENZYMEfit}
		\small{		\begin{tabular}{c@{\hskip 0.20in}c@{\hskip 0.20in}c@{\hskip 0.20in}c@{\hskip 0.20in}c@{\hskip 0.20in}c@{\hskip 0.20in}c@{\hskip 0.20in}c@{\hskip 0.20in}} \hline
				Parameter   & Estimates & SE     &$\mbox{SE}_{b}$ &  L     & U      & $\mbox{L}_{b}$ & $\mbox{U}_{b}$ \\
				\hline
				$\alpha_1$	& 0.5239    & 0.0231 & 0.0232         & 0.4788 & 0.5689 & 0.4187 & 0.6289   \\
				$\alpha_2$  & 0.3231    & 0.0284 & 0.0539         & 0.2677 & 0.3785 & 0.2779 & 0.3683         \\
				$\beta_1$	& 0.1734    & 0.0083 & 0.0070         & 0.1572 & 0.1896 & 0.1597 & 0.1870          \\
				$\beta_2$	& 1.2669    & 0.0464 & 0.0445         & 1.1764 & 1.3574 & 1.1801 & 1.3537           \\
				$p_1$		& 0.6259    & 0.0312 & 0.0395         & 0.5651 & 0.6867 & 0.5489 & 0.7029          \\
				\hline
\end{tabular}}
\end{center}
\end{table}

The results clearly show that the 2-component FM-BS model has the best fit. Based on Section \ref{notesimplementation}, the initial values are $p_1^{(0)}=0.6408$, $\alpha_1^{(0)}=0.5630$, $\alpha_2^{(0)}=0.3017$, $\beta_1^{(0)}=0.1802$ and $\beta_2^{(0)}=1.3008$. Table \ref{tableENZYMEfit} presents the ML  estimates  of $p_1$, $\alpha_1$, $\alpha_2$, $\beta_1$ and $\beta_2$  for the FM-BS model along with the corresponding standard errors (SE), obtained via the information-based procedure presented in Section \ref{imtheo}, which is used to obtain the lower (L) and upper (U) confidence limits. Moreover, the bootstrap approach, developed by \cite{efron1986bootstrap}, provides another way of deriving confidence intervals. This table presents the bootstrap estimated standard errors ($\mbox{SE}_{b}$), and the two-sided  95\%  confidence intervals ($\mbox{L}_{b}$ and $\mbox{U}_{b}$ are bootstrap confidence limits with 400 bootstrap replicates).
\begin{table}[H]
	\begin{center}
		\caption{Comparison of log-likelihood, AIC and BIC for fitted FM-BS, FM-logN and FM-SN models using the BMI data.}\label{tableBMIindicators}
		\vskip 3mm
		\small{	\begin{tabular}{ccccccc}
				\hline
				& $G$     &log-lik   & AIC        & BIC       	\\
				\hline		
				FM-LogN   & 1       &-86.2856  &14212.65   & 14223.95   \\ 			
				FM-LogN   & 2       &-86.2856  &14283.20   & 17165.83   \\ 		
				FM-LogN   & 3       &-86.2856  &15134.94   & 18895.17  \\ 	
				\hline
				FM-SN     & 1       &-7234.190 &14474.38  & 14491.34   \\ 			
				FM-SN     & 2       &-6911.778 &13837.56  & 13877.13   \\ 		
				FM-SN     & 3       &-6862.755 &13804.74  & 13809.69  \\ 		
				\hline
				FM-BS     & 1       &-7099.455 &14202.91  & 14214.22  \\
				FM-BS     & 2       &-6886.495 &13782.99  & 13811.26   \\
				FM-BS     & 3       &-6858.605 &\textbf{13733.21}  & \textbf{13778.43}   \\
				\hline
		\end{tabular}}
	\end{center}
\end{table}
\subsection{Real dataset II}
As a second application, we consider the body mass index for 2107 men aged between 18 to 80 years. The dataset comes from the National Health and Nutrition Examination Survey, conducted by the National Center for Health Statistics (NCHS) of the Centers for Disease Control (CDC) in the USA. These data have been analyzed by \cite{Basso2010}, who fitted them to finite mixture of skewed distributions, for example finite mixture of skew-normal (FM-SN). The estimation algorithm is implemented in the \textbf{R package} \texttt{mixsmsn} where the $k$-means clustering algorithm is used to obtain the initial values. We performed the EM algorithm to carry out the ML estimation for the FM-BS, finite mixture of Log-normal (FM-logN), defined by \cite{Mengersen.Robert.Titterington2011}, and FM-SN, for model comparison. Table \ref{tableBMIindicators} contains the log-likelihood together with AIC and BIC for several components. Figure 10 (a) and (b)  displays histograms with estimated pdfs for the data superimposed with $G=1-4$ components. The graphs show that  the FM-BS model ($G=3$ or $G=4$) adapts to the shape of the histogram very accurately. The results indicate that the FM-BS model with $G=3$ components provides a better fit than the other models considered and this is verified by hypothesis testing using parametric bootstraping like in application 1.
\begin{center}
\def\stackalignment{l}
\topinset{\bfseries(a)}{\includegraphics[width=3.3in]{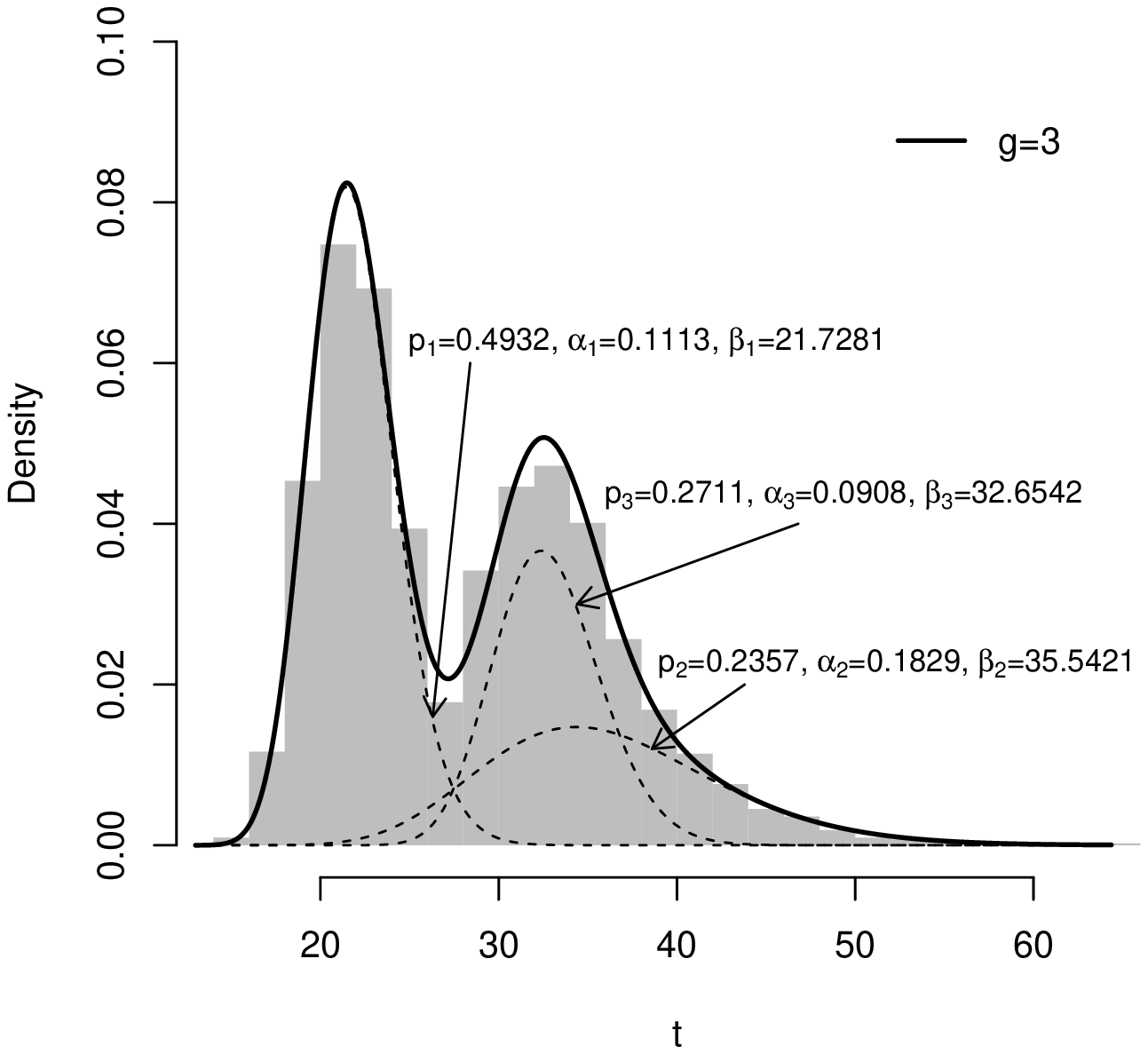}}{0.4in}{0.02in}~\topinset{\bfseries(b)}{\includegraphics[width=3.3in]{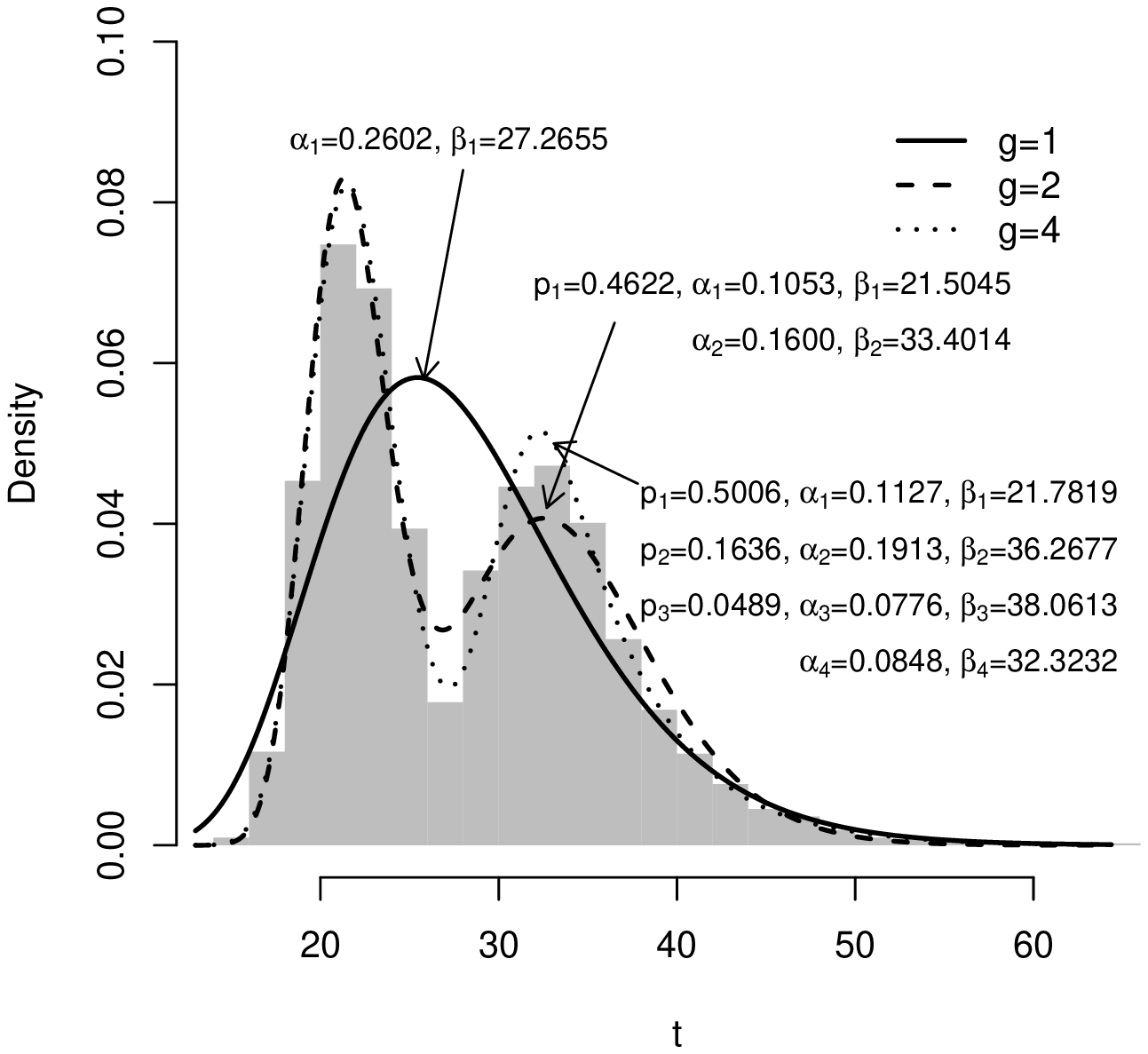}}{0.4in}{0.02in} \vspace*{-1.4cm} \\
\topinset{\bfseries(c)}{\includegraphics[width=3.3in]{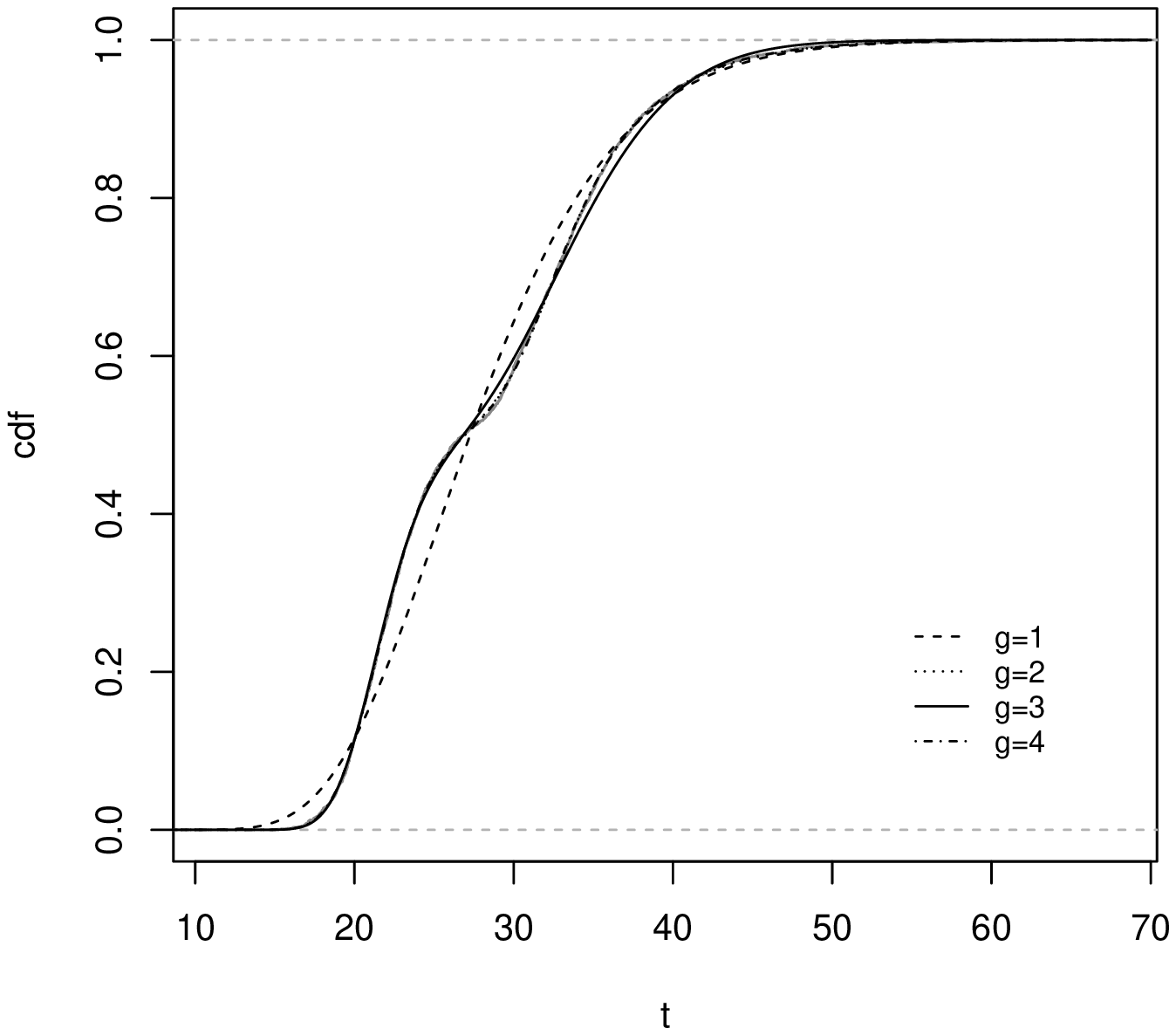}}{0.4in}{0.02in}~\topinset{\bfseries(d)}{\includegraphics[width=3.3in]{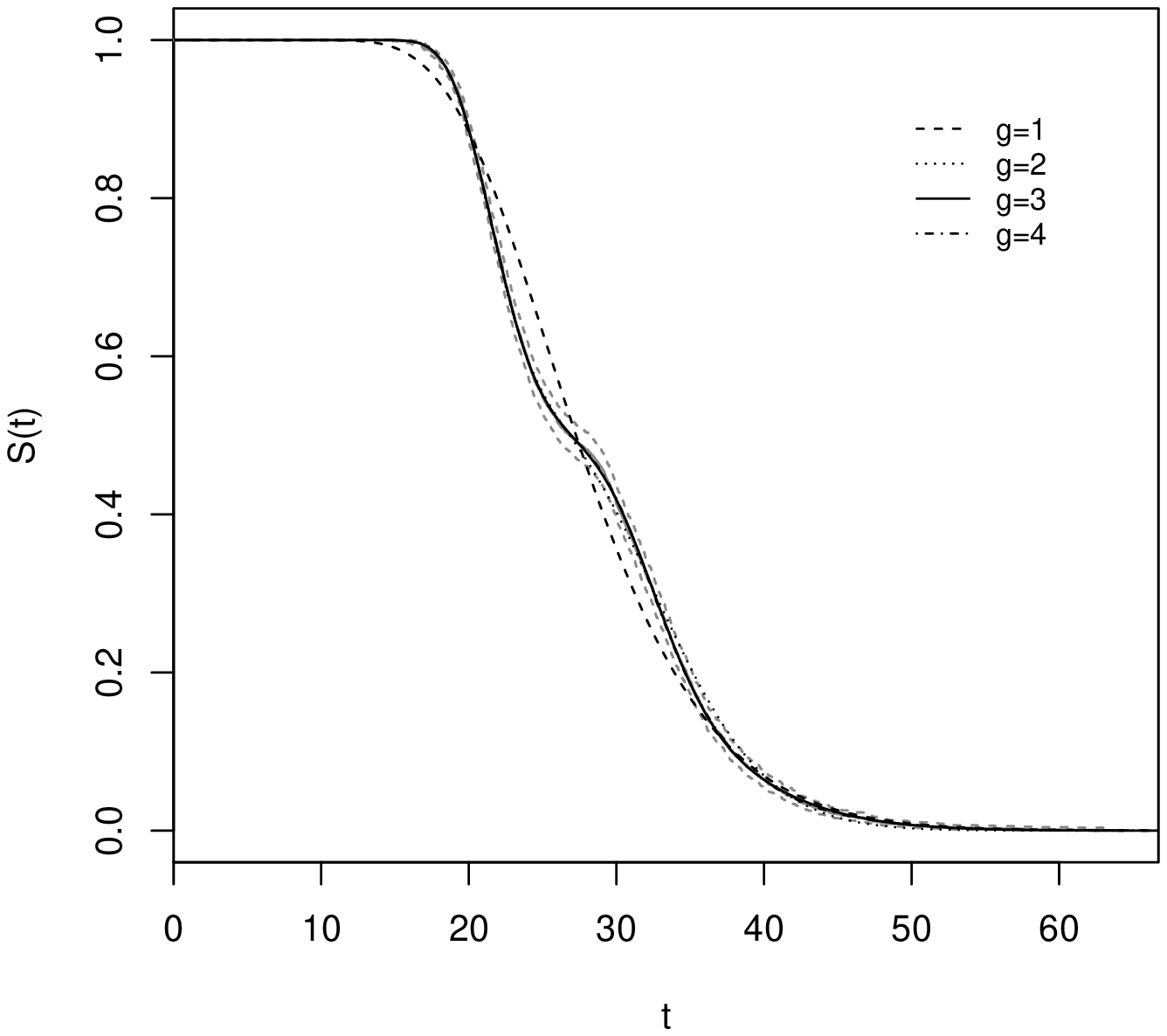}}{0.4in}{0.02in}
\begin{flushleft}
Figure 10: (a) and (b) Histogram of the BMI data with overlaid ML-fitted densities (c) Cumu- lative  and  (d) estimated survival  functions and the empirical \,\,survival\,\, function for\,\, four\,\, fitted FM-BS models.  \vspace*{1.0cm}
\end{flushleft}
\end{center}
 The $p$-value for the $G=2$ versus $G=3$ test is <0.000. Thus, there is strong evidence that at least three components exist. For the $G=3$ versus $G=4$, test the $p$-value is 0.132, so there is no evidence that more than three components are required. In conclusion, results based on AIC, BIC, hypothesis testing and figures indicate that the FM-BS model with $G=3$ components provides the best fit.
Table \ref{tableBMIfit} presents the MLE of $p_1$, $p_2$, $\alpha_1$, $\alpha_2$ , $\alpha_3$, $\beta_1$, $\beta_2$ and $\beta_3$ for the FM-BS model  along with the corresponding standard errors (SE) and L and U confidence limits. In Figure 10 (c) and (d) we show the cumulative and estimated survival functions and the empirical survival function of BMI data for four fitted FM-BS models respectively.
\begin{table}[H]
	\begin{center}
		\caption{Estimated parameter values via the EM algorithm and with the corresponding standard errors (SE) for the FM-BS model applied to the BMI data.}\vskip 3mm
		\label{tableBMIfit}
		\small{		\begin{tabular}{c@{\hskip 0.20in}c@{\hskip 0.20in}c@{\hskip 0.20in}c@{\hskip 0.20in}c@{\hskip 0.20in}c@{\hskip 0.20in}c@{\hskip 0.20in}c@{\hskip 0.20in}} \hline
				Parameter   & Estimates & SE     &$\mbox{SE}_{b}$ &  L     & U      & $\mbox{L}_{b}$ & $\mbox{U}_{b}$ \\
				\hline
				$\alpha_1$	& 0.1113    & 0.0050 & 0.0541         & 0.1016 & 0.1210 & 0.0058 & 0.2168   \\
				$\alpha_2$  & 0.1829    & 0.0270 & 0.0366         & 0.1302 & 0.2356 & 0.1115 & 0.2543         \\
				$\alpha_3$  & 0.0908    & 0.0170 & 0.0650         & 0.0577 & 0.1240 &-0.0359 & 0.2176        \\			
				$\beta_1$	&21.7281    & 0.2025 & 1.9123         &21.3332 &22.1230 &17.9991 &25.4571          \\
				$\beta_2$	&35.5421    & 3.6312 & 3.0561         &28.4613 &42.6229 &29.5827 &41.5015           \\
				$\beta_3$	&32.6542    & 0.3640 & 1.7669         &31.9444 &33.3640 &29.2087 &36.0997           \\			
				$p_1$		& 0.4932    & 0.0330 & 0.0234         & 0.4288 & 0.5576 & 0.4476 & 0.5388          \\
				$p_2$		& 0.2357    & 0.1570 & 0.0305         &-0.0704 & 0.5418 & 0.1581 & 0.3133          \\				
				\hline
		\end{tabular}}
	\end{center}
\end{table}

\section{Conclusions}
This work proposes finite mixture of Birnbaum-Saunders distributions, extending some results proposed by \cite{Balakrishnan:11} and providing important supplementary findings regarding mixture of BS distributions. The  resulting model simultaneously accommodates multimodality and skewness, thus allowing practitioners from different areas to analyze data in an extremely flexible way.

We pointed out some important characteristics and properties of FM-BS  models that allow us to obtain qualitatively better ML estimates and efficiently compute them by using the proposed  EM-algorithm, which can be easily implemented and coded with existing statistical software such as the R language. The efficiency of   the  EM algorithm  is  supported  by  the  use  of  the $k$-bumps algorithm  to obtain the initial values  of   model  parameters. We noted interesting advantages in comparison with the other algorithms ($k$-mean and $k$-medoids), because the final  estimates do not  change each time  the algorithm is executed.


The FM-BS model  can be extended to multivariate settings, following the  recent proposal of \cite{khosravi2014} for mixtures of bivariate Birnbaum-Saunders distributions. We intend to pursue this in future research. Another worthwhile task is to develop a fully Bayesian inference via the Markov chain Monte Carlo method.\\

\noindent \textbf{Acknowledgements}
\noindent  {Roc\'{i}o Maehara and Luis Benites were supported by CNPq-Brazil. The computer program, coded in the R language, is available from the first author upon request.}\\


\end{document}